\documentclass[a4paper,UKenglish,thm-restate]{lipics-v2021}
\nolinenumbers
\usepackage{amsmath,amssymb,amsfonts,todonotes}
\usepackage{thmtools}
\usetikzlibrary{decorations.pathreplacing,shapes.geometric}
\usepackage[ruled,vlined,linesnumbered]{algorithm2e}
\hypersetup{colorlinks=true, allcolors=blue}

\newcommand\DGP{\textsc{Dense Graph Partition}}

\newenvironment{proofofclaim}{%
  \begin{proof}%
}{%
  \end{proof}%
}

\title{Connected (Dense) Partition for Tree-Like Graphs}

\author{Katrin {Casel}}{Humboldt-Universität Berlin, Germany}{katrin.casel@hu-berlin.de}{https://orcid.org/0000-0001-6146-8684}{}
\author{Archontia C. Giannopoulou}{Department of Informatics and Telecommunications, National and Kapodistrian University of Athens, Greece}{archontia.giannopoulou@gmail.com}{https://orcid.org/0009-0001-4368-2852}{}
\author{Aikaterini {Niklanovits}}{Hasso Plattner Institute, Germany}{Aikaterini.Niklanovits@hpi.de}{https://orcid.org/0000-0002-4911-4493}{}

\authorrunning{K. Casel, A.\,C. Giannopoulou and A. Niklanovits}
\Copyright{Katrin Casel, Archontia C. Giannopoulou and Aikaterini Niklanovits}

\begin{document}
\maketitle

\begin{abstract}
We focus on two variants of graph partitioning problems, connected partition and dense partition. Formally, given a graph $G=(V,E)$ and a partition of its vertices $\mathcal P=\{P_1,\ldots, P_k\}$ we say that $\mathcal P$ is a connected partition of $G$ if each $P_i$ induces a connected graph in $G$. Many classical variants of this problem impose additional restrictions both on the number of parts as well as on the size of each part.
Moreover, given a partition $\mathcal P=\{P_1,\ldots, P_k\}$  we define its density by $d(\mathcal P):=\sum_{i=1}^k |E(P_i)|/|V(P_i)|$. The problem \textsc{Maximum Dense Graph Partition} asks to construct a partition of maximum density. We study this problem both with and without fixed number of sets $k$.
We prove the following~results:
\begin{enumerate}
    \item A polynomial time algorithm for \textsc{Maximum Dense Graph Partition} of thick forests, a subclass of chordal graphs, generalizing the previously known polynomial time algorithm on block graphs.
    \item A generic dynamic programming algorithm to construct (if possible) a connected partition into $k$ sets of prescribed sizes on graphs with bounded treewidth. This yields algorithms for both variants of Dense Graph Partition and an efficient construction for the Gy\H{o}ri-Lov\'{a}sz theorem.
    \item The $\mathsf{NP}$-hardness of \textsc{Maximum Dense Graph Partition}  restricted to split graphs, showing that thick trees are a \emph{boundary} for the polynomial computability of this problem.
\end{enumerate}
\end{abstract}

\newpage
\section{Introduction}

Graph partitioning problems constitute a central topic in graph algorithms and combinatorial optimization, with applications ranging from clustering \cite{celebi2015partitional} and community detection \cite{newman2004detecting} to network decomposition \cite{yahia2018network} and job allocation \cite{ali1992graph}. In order for those problems to find applications certain additional restrictions arise naturally with the most general being each part of the partition to induce a connected graph, see for example \cite{DBLP:journals/dam/LucertiniPS93, DBLP:journals/jea/MohringSSWW06, DBLP:journals/tods/FanYXZLYLCX18, DBLP:conf/osdi/GonzalezLGBG12}. 
Connectivity alone however is not a sufficient objective for a useful partitioning. Prominent additional constraints are fixing the number and sizes of parts or optimizing for high density. 

The most general form of size-constraints is the option of requesting a partition into $k$ connected subgraphs of specified sizes $n_1,\dots, n_k$.   
We call these types of partition \emph{Gy\H{o}ri-Lov\'{a}sz Partition} in reference to  Gy\H{o}ri \cite{Gyori1976} and Lov\'{a}sz \cite{Lovasz1977} who independently proved that $k$-connectivity is both a sufficient and necessary condition on a graph for such a partition to exist for any choice of  $n_1,\dots, n_k$. Both proofs of existence unfortunately  do not translate to an efficient computation of such a partition, and it still remains open if this is possible for $k>4$ \cite{DBLP:journals/ipl/ZuzukiTN90, GLcase3, DBLP:conf/wg/WadaK93, DBLP:journals/ipl/NakanoRN97, Hoyer2016}.
 Further, without the restriction to $k$-connected graphs,  already the specialized problem of partitioning into parts of almost equal size is $\mathsf{NP}$-hard~\cite{ItoZN06}.

An alternative to specified sizes is optimizing for the classical definition of the density (see for example~\cite{goldberg1984finding,DBLP:conf/icalp/KhullerS09,DBLP:journals/dam/DarlayBM12}), that is, the density of a graph  is equal to its average degree, meaning the amount of its edges divided by the amount of its vertices. Formally, this objective gives the Maximum Dense Graph Partition problem (MDGP), in which an undirected graph is given and the goal is to find a partition of its vertices $\mathcal P=\{P_1,\ldots, P_k\}$ such that the sum of the densities of the subgraphs induced by each $P_i$ is maximized. Note that in MDGP the number of parts of the partition is not specified, and the objective of maximizing density implies that all parts are connected. While finding a subgraph of largest density is famously known to be solvable in polynomial time~\cite{goldberg1984finding}, MDGP is $\mathsf{NP}$-hard~\cite{DBLP:journals/dam/DarlayBM12}. To understand the difference between densest subgraph and  MDGP, consider for example a graph on six vertices where three of them form a triangle and each of the remaining three is adjacent to a distinct vertex of the triangle, as shown in Figure~\ref{fig:example1}.
Taking greedily the densest sugraph as a part for a MDGP partition would result in just considering the triangle as one part and then also have three isolated vertices. The optimal MDGP solution is a perfect matching of the graph, that is, the three edges that are adjacent to the vertices of degree~$1$. Good solutions to MDGP thus cannot use individual greedy ideas but need to partition more globally. 
 Under a different name, this problem is also studied in game theory in the context of social fairness. Specifically, Aziz et al.~\cite{DBLP:conf/ijcai/AzizGGMT15} studied the problem Fractional Hedonic Game, and more particularly the Maximum Utilitarian Welfare problem (MUW). In this setting, instead of parts we talk about coalitions where a formation of coalitions corresponds to a partition of the graph. The value of the optimal formation of coalitions equals to twice the value of the optimal dense partition.

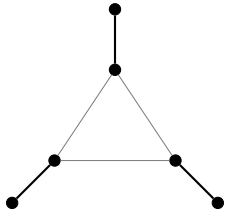
\begin{figure}[h]
\begin{center}
\begin{tikzpicture}[every node/.style={fill=black,inner sep=1.5pt},scale=0.8]
\node[shape=circle,draw=black] (a1) at (-1,0){};
\node[shape=circle,draw=black] (b1) at (1,0){};
\node[shape=circle,draw=black] (c1) at (0,1.5){};

\node[shape=circle,draw=black] (a2) at (-1.7,-0.7){};
\node[shape=circle,draw=black] (b2) at (1.7,-0.7){};
\node[shape=circle,draw=black] (c2) at (0,2.5){};

\path [-, thick] (a1) edge (a2);
\path [-, thick] (b1) edge (b2);
\path [-, thick] (c1) edge (c2);
\path [-,gray] (a1) edge (b1);
\path [-,gray] (a1) edge (c1);
\path [-,gray] (b1) edge (c1);

\end{tikzpicture}

\caption{If we split the graph in such a way that a part contains exactly the clique of size $3$ the overall density is $1$. If we choose the perfect matching as a partition then the overall density is $3\cdot\frac{1}{2}$.}
\label{fig:example1}
\end{center}
\end{figure}

Both  Maximum Dense Graph Partition and Gy\H{o}ri-Lov\'{a}sz Partition have been mostly studied under restricted graph classes to identify tractable inputs.

MDGP is known to be polynomial time solvable on trees~\cite{DBLP:journals/jair/Bilo0FMM18} and  graphs with minimum degree $(n-3)$~\cite{DBLP:journals/jcss/BazganCC25}. Further,  Hanaka et al.~\cite{DBLP:journals/jco/HanakaIO25} provided a polynomial algorithm for block graphs and also showed that MDGP can be solved in $n^{\mathcal O(\tau)}$ time where $\tau$ is the vertex cover number of the input graph. Moreover, they provide an $(nW)^{\mathcal O(w)}$ time algorithm where $w$ is the treewidth of the input graph and $W$ the maximum absolute edge weight.
On the negative side, MDGP is known to be $\mathsf{NP}$-hard  even restricted to bipartite, cubic, or $(n-4)$-regular graphs~\cite{DBLP:journals/jcss/BazganCC25}.
MDGP is also studied in the variation where the number of parts to be built is specified, and we denote this version by $\text{MDGP}_k$. 
Most  hardness results listed for MDGP transfer to $\text{MDGP}_k$, and is known that $\text{MDGP}_k$ is solvable in polynomial time on trees \cite{DBLP:conf/ictcs/Li21}.

It is known that Gy\H{o}ri-Lov\'{a}sz Partition is solvable in polynomial time when restricted to $k$-connected chordal graphs~\cite{DBLP:conf/wg/CaselFINZ23}, or $k$-connected biconvex graphs~\cite{DBLP:conf/esa/NiklanovitsSVZ25}.
The specialized variation with all $n_i$ being equal, as decision problem called Equitable Connected Partition (ECP), is known to be fixed parameter tractable when parameterized by the vertex cover number and XP when parameterized by treewidth~\cite{equi}. On the negative side, this specialized ECP is known to be  W[1]-hard when parameterized by pathwidth, feedback vertex set  and $k$ combined, even for planar graphs. Further, ECP is already $\mathsf{NP}$-hard on grid graphs \cite{DBLP:journals/networks/BeckerLLS98}.

For approximation, restriction to graph classes was used to derive better results. 
The $2$-approximation for MDGP on general graphs by Aziz et al.~\cite{DBLP:conf/ijcai/AzizGGMT15} was improved by Bazgan et al.~\cite{DBLP:journals/jcss/BazganCC25} to ratio $4/3$ on cubic graphs, and even to an EPTAS on graphs with minimum degree $n-t$ for any constant $t\geq 4$. The optimization problem of Equitable Connected Partition is usually called Balanced Connected Partition $\text{BCP}_k$ with the objective to either to minimize the size of the largest, or to maximize the size of the smallest part. The $3$-approximation for both versions of $\text{BCP}_k$ by \cite{DBLP:conf/approx/BorndorferCINSZ21} was improved to ratio 2 on claw-free graphs.

\subsection*{Our Contribution}
We extend the line of research on Maximum Dense Graph Partition and Gy\H{o}ri-Lov\'{a}sz Partition on specified graph classes, particularly exploring the gaps left open by the existing literature. 

First, we study the structure of optimal solutions for MDGP on well partitioned chordal graphs.
To this end, we define the notion of a \emph{generalized star}. Informally, a graph is a generalized star, if it can be obtained by a star after blowing each vertex up to be a clique, with the leaves connected to at least half the vertices in the center clique. With this notion, we show the following structural property for MDGP on well partitioned chordal graphs.

\begin{restatable}{theorem}{genstars}  
\label{thm::genstars}
For $\text{MDGP}$  restricted to well partitioned chordal graphs, there exists an optimum solution where each part induces a generalized star or a clique.   
 \end{restatable}
Observe that both generalized stars and cliques have diameter at most 2. Thus, at least on well partitioned chordal graphs, we can resolve the open question, if there always exists an optimal MDGP solution into parts of diameter at most 2. Interestingly, one of the steps the proof for Theorem~\ref{thm::genstars} is split into, is exactly proving that each part in an optimum solution for a well partitioned chordal graph has diameter at most 2. 
 
Theorem~\ref{thm::genstars} in particular gives us a structural insight into solutions for any subclass of well partitioned chordal graphs.
Not allowing two minimal separators to share a vertex, unless they are identical, turns well partitioned chordal graphs into the smaller graph class of \emph{thick forests}. 
Restricted to thick forests, we show that generalized stars have an even more specific structure which allows us to formulate the density of a thick tree based on the sizes of its leaves and its center. We argue that because the set of vertices of each pair of minimal separators of a thick tree are either equal or disjoint we are able to derive a representation of a thick tree as a block-cut tree. Via dynamic programming on this tree, we show that,

\begin{restatable}{theorem}{thicktrees}  
\label{thm::thicktrees}
 $\text{MDGP}$ is solvable in polynomial time on thick trees, and hence on thick forests.   
 \end{restatable}

This generalizes the result of Hanaka et al.~\cite{DBLP:journals/jco/HanakaIO25} who provided a polynomial time algorithm for MDGP on block graphs.
Note that this result cannot be extended to $\text{MDGP}_k$ since it is based on the structural insight given by Theorem~\ref{thm::genstars}. With the simple example of any well partitioned chordal graph that is not itself a generalized star or a clique, and the choice of $k=1$, one can see that Theorem~\ref{thm::genstars} does not hold for $\text{MDGP}_k$.

To argue why the extra restriction to thick trees might be needed, we give a complementing hardness result on split graphs, a subclass of well partitioned chordal graphs. In particular, we show the following.

\begin{restatable}{theorem}{splitnph}\label{thm::splitnph}
 $\text{MDGP}$ restricted to split graphs is $\mathsf{NP}$-hard.     \end{restatable}
 Hence, it seems that the amount of vertices shared by a pair of minimal separators of a graph essentially encodes the difficulty of $\text{MDGP}$.
 
 The proof of Theorem~\ref{thm::splitnph}  is a reduction from the \textsc{Dominating Set} problem such that a yes-instance is translated to a split graphs that can be partitioned into $k$ complete split graphs that are \emph{balanced}. In this context, we say that a partition $P_1,\dots, P_k$ of the vertices of a split graph $G=(I\cup K,E)$ is balanced, if $\frac{|P_i\cap I|}{|P_i\cap K|}=\frac{|I|}{|K|}$ for each $i$; this is somewhat the split graph version of the optimal partition characterization for complete bipartite graphs (Corollary 5 in~\cite{DBLP:journals/jcss/BazganCC25}). 
It is however not true that such balanced partitions of complete split graphs give the best possible density. This makes the proof of Theorem~\ref{thm::splitnph} significantly harder than the result for bipartite graphs.

In order to not fall back onto reducing just to $\text{MDGP}_k$, we need some involved computations and insights to bound the number of parts. But even with this, the proof is far from trivial as one still cannot assume perfect balance. Thus we need to argue about loss of density for parts that are too unbalanced or have too many edges missing from being a complete split graph.

At last, we consider parameterization by treewidth, and derive a generic algorithm for connected partitions with the following theorem.

\begin{restatable}{theorem}{conpart}  
\label{thm::conpart}
Given a graph $G=(V,E)$ with a tree decomposition of width $w$ and $k$ integers $n_1,\ldots, n_k$ with $\sum_{i=1}^kn_i=|V|$, a connected partition $P_1,\dots, P_k$ of $G$ with $|P_i|=n_i$ for each $i$ can be computed in 
 $\mathcal O(nw^2k^{2(w+1)}(n+1)^{2k}4^{\binom{w+1}{2}})$.
 \end{restatable}

The algorithm used for Theorem~\ref{thm::conpart} is the expected technical dynamic programming algorithm on tree decompositions, that tracks all information needed about forgotten edges.  Incorporating forbidden partial solutions where given terminals are misplaced directly gives a constructive algorithm for the  Gy\H{o}ri-Lov\'{a}sz theorem, generalizing the result of Enciso et al.~\cite{equi} for the special case problem ECP.

\begin{restatable}{corollary}{glpart}  
\label{corr::glpart}
Given a graph $G=(V,E)$ with a tree decomposition of width $w$, $k$ integers $n_1,\ldots, n_k$ with $\sum_{i=1}^kn_i=|V|$ and terminals $v_1,\dots,v_k\in V$, computing, if exists, a connected partition $P_1,\dots, P_k$ of $G$ with $|P_i|=n_i$ and $v_i\in P_i$ for each $i$ can be done in time 
$\mathcal O(nw^2k^{2(w+1)}(n+1)^{2k}4^{\binom{w+1}{2}})$.
 \end{restatable}

Further, the information tracked about edges, can be used to optimize for density (instead of tracking sizes and connectivity). Thus this generic algorithm can also be used for MDGP$_k$. Further, any optimum solution for MDGP only contains parts that induce connected subgraphs, thus it is also straightforward to adjust the generic algorithm to MDGP - at any point, there are at most treewidth many unfinished clusters to be tracked. These simple observations directly yield the following, improving on the results of Hanaka et al.~\cite{DBLP:journals/jco/HanakaIO25}.

\begin{restatable}{corollary}{mdgpkk}  
\label{corr::mdgpkk}
Given a graph $G=(V,E)$ with a tree decomposition of width $w$, MDGP$_k$ and MDGP can be solved in time
$\mathcal O(nw^2k^{2(w+1)}(n+1)^{2k}4^{\binom{w+1}{2}})$ and $\mathcal O(nw^{2(w+1)}(n+1)^{2w}4^{\binom{w+1}{2}})$, respectively.
 \end{restatable}

The generalized algorithm we present is somewhat the generic way one would approach MDGP directly. The relationship to connected partitions given here, together with the W[1]-hardness of ECP for parameter pathwidth might be an indication that MDGP is not FPT parameterized by treewidth. On the other hand, the only thing missing from giving an FPT algorithm for MDGP with parameter treewidth is an efficient way to describe parts in an optimum solution that do not have a size bounded by some function in the treewidth. Parts that induce a star can unfortunately have such unbounded sizes. 

While we do not settle the open question about treewidth for MDGP with our algorithm, it still gives valuable insight: If MDGP is not FPT parameterized by treewidth, a reduction must use optimum solutions with large parts. If one can show that only rather simple structures like stars are problematic together with an efficient way to track their density progress, our algorithm gives FPT time.

Due to space restrictions some proofs are moved to the Appendix. Those proofs are denoted by a $\star$ next to the statement to be proved.

\section{Preliminaries}\label{sec: prelim}
Most of the notions used in this paper are standard graph theory notions so due to space restrictions a large part of them is presented in Appendix~\ref{appendix:prelim}. Here we only present the definitions of the relevant graph classes.

 We say that a graph is \emph{chordal} if all its induced cycles are triangles. A subclass of chordal graphs are well partitioned chordal graphs. In particular, a connected graph $G$ is a \emph{well partitioned chordal graph} if there exist a partition $\mathcal{P}$ of $V(G)$ and a tree $\mathcal{T}'$ having $\mathcal{P}$ as a vertex set such that the following hold: 
 \begin{enumerate}
 \item each part $X\in\mathcal P$ is a clique in $G$, 
 \item for each edge $XY\in E(\mathcal T')$, there are subsets $X'\subseteq X$ and $Y'\subseteq Y$ such that $E(G[X,Y])=X'\times Y'$, and 
 \item for each pair of distinct $X, Y\in V(\mathcal T')$ with $XY\not\in E(\mathcal T')$, $E(G[X,Y])=\emptyset$.
 \end{enumerate}
Moreover, a graph $G$ is a \emph{thick tree} if it is chordal and every two minimal separators of $G$ are either equal or disjoint. It is easy to notice that thick trees are a subclass of well partitioned chordal graphs, and also that block graphs are a subclass of thick trees. We say that a graph $G$ is a \emph{block graph} if every biconnected component of it is a clique. We also say that a graph $G$ is \emph{bipartite} if its vertices can be partitioned into two sets each inducing an independent set. Lastly, we also work on split graphs, where $G$ is a \emph{split graph} if there is a partition of its vertices into two parts $K,I$, where $I$ induces an independent set and $K$ induces a clique. For both bipartite and split graphs we say that a graph is \emph{complete} if it has the maximum amount of edges while still being bipartite or split respectively.

\section{Dense Partition On Well Partitioned Chordal Graphs}\label{section: well-partitioned chordal}
In this section we analyze the structure of each part of an optimal dense partition on well partitioned chordal graphs. In particular, we prove that each part has diameter at most $2$ and induces either a clique, or a specific structure we call generalized star.
We say that a set of vertices induces a generalized star if both its center and its leaves are cliques of various sizes, and leaves are well connected to the center. Formally,
\begin{definition}
We call a graph $G$, a generalized star and denote it by $S_{l,l_1,\ldots, l_r}$ if there is a partition $\mathcal{P}=\{P_{0},P_{1},\dots,P_{r}\}$ with $r\geq 1$, of $G$ such that
\begin{itemize}
\item $G[P_{i}]$ induces a clique, for every $0\leq i\leq r$,
\item $|V(G[P_{0}])|=l$,
\item for every $i\in [r]$, $|V(G[P_{i}])|=l_{i}$ and $N_{G}(P_{i})\subseteq P_{0}$,
and
\item for every $i\in [r]$, every vertex $v\in P_{i}$ has at least $\lfloor\frac{l}{2}\rfloor+1$ neighbors in $P_{0}$.
\end{itemize}
We refer to the clique of size $l$ induced by $P_{0}$ as the center of the generalized star and to each clique of size $l_{i}$ induced by $P_{i}$ as a leaf of the generalized star.
Finally, for any such partition for which all but the last condition hold, we call $G$ an incomplete star.
\end{definition}
Notice that since each vertex of each leaf of a generalized star is adjacent to more than half of the vertices of the center clique, the diameter of a generalized star is $2$.

The rest of this section is dedicated to proving the following theorem.
\genstars*
Towards proving Theorem~\ref{thm::genstars}, we first compare the density of keeping an incomplete leaf attached to the center of a generalized star or considering it as a separate part (or considering the leaf plus one central vertex as a separate part). We see that when the leaf vertices are adjacent to at most half of the center then separating it from the rest of the star does not decrease the overall density of the partition.

\begin{restatable}{lemma}{stardecomposition}\label{lemma:star-decomposition}
Let $G$ be an incomplete star with partition
$\mathcal{P}=\{P_{0},P_{1},\dots,P_{r}\}$ such that for every $i\in[r]$ every
vertex in $P_{i}$ has at least one neighbor in $P_{0}$,
$\bigcup_{i\in[r]}N(P_{i})=P_{0}$, and for every $u,v\in P_{i}$,
$N(u)\cap P_{0}=N(v)\cap P_{0}$. Then $G$ can be partitioned into generalized
stars and cliques whose densities sum up to at least $d(G)$. 
\end{restatable}
\begin{proof}
Write $n_{2}=|P_{0}|$, $\ell_{i}=|N(P_{i})|$ for $i\in[r]$, $n=|V(G)|$ and
$m=|E(G)|$, so that $d(G)=m/n$. Let also $N=\sum_{i\in[r]}|P_{i}|$ and note that
$n=n_{2}+N$. We call a leaf $P_{i}$ \emph{weak} if
$\ell_{i}\leq\lfloor n_{2}/2\rfloor$ and let $B=\{i\in[r]:P_{i}\text{ is weak}\}$.
We describe two partitions and show that at least one of them is as desired.

\medskip
\noindent\emph{Partition into cliques.}
Let $(Z_{1},Z_{2},\dots,Z_{r})$ be a partition of $P_{0}$ into (possibly empty sets) such that 
$Z_{i}\subseteq N(P_{i})$.
Indeed, since $\bigcup_{i\in[r]}N(P_{i})=P_{0}$, we may assign every vertex of $P_{0}$
to a leaf in whose neighborhood it lies. Every vertex of $P_{i}$ is adjacent to every
vertex of $Z_{i}\subseteq N(P_{i})$, and both $P_{i}$ and $Z_{i}$ are cliques,
so each $P_{i}\cup Z_{i}$ is a clique. 
Moreover, the sets
$P_{1}\cup Z_{1},\dots,P_{r}\cup Z_{r}$ form a partition of $V(G)$. 
Since a clique on $s$ vertices has density $(s-1)/2$, their densities sum to
\[
\sum_{i\in[r]}\frac{|P_{i}|+|Z_{i}|-1}{2}=\frac{n-r}{2}.
\]
Thus this partitioning into cliques is as desired whenever $\frac{n-r}{2}\geq\frac{m}{n}$, that is,
whenever
\begin{equation}\label{eq:Awins}
n(n-r)\geq 2m.
\end{equation}

\medskip
\noindent\emph{Partition into a star and cliques.}
For $i\notin B$, it holds that $\ell_{i}\geq\lfloor n_{2}/2\rfloor+1$, so
$\Sigma:=G\!\left[P_{0}\cup\bigcup_{i\notin B}P_{i}\right]$ is a generalized star.
In particular, it is a clique if $B=[r]$, that is, if every leaf is weak. 
We partition $G$ into that following parts.
$\Sigma$ is one part and each weak leaf $P_{i}$,
$i\in B$, is kept as a separate part which induces a clique.
This yields a partition of $G$ into a generalized star and cliques.
The total density of these parts is $d(\Sigma)+\sum_{i\in B}\frac{|P_{i}|-1}{2}\geq d(\Sigma)$.
Let $N_{B}=\sum_{i\in B}|P_{i}|$ and let $e_{B}$ be the number of edges incident
to $\bigcup_{i\in B}P_{i}$. As these leaves are cliques adjacent only to $P_{0}$,
$e_{B}=\sum_{i\in B}\big(\binom{|P_{i}|}{2}+|P_{i}|\ell_{i}\big)$. Since $\Sigma$
has $n-N_{B}$ vertices and $m-e_{B}$ edges,
\[
d(\Sigma)\geq d(G)
\iff\frac{m-e_{B}}{\,n-N_{B}\,}\geq\frac{m}{n}
\iff N_{B}\,m\geq n\,e_{B}
\iff \frac{m}{n}\geq\frac{e_{B}}{N_{B}} .
\]
Thus partitioning into a generalized star and cliques is as desired whenever $d(G)\geq e_{B}/N_{B}$.

\smallskip
\noindent\emph{One of the two inequalities holds.}
Assume, in order to reach a contradiction, that neither does. As
\eqref{eq:Awins} fails, $d(G)>\frac{n-r}{2}$. 
Moreover, since partitioning into a generalized star and cliques fails, 
$d(\Sigma)<d(G)$ and hence $d(G)<e_{B}/N_{B}$, as otherwise the total
density of the partition would be at least $d(G)$. Moreover, $B\neq\emptyset$,
so $N_{B}>0$. Let $p_{\max}=\max_{i\in[r]}|P_{i}|$ and recall that $\ell_{i}\leq\frac{n_{2}}{2}$ for
$i\in B$ and $e_{B}=\sum_{i\in B}|P_{i}|\big(\frac{|P_{i}|-1}{2}+\ell_{i}\big)$. 
Therefore,
\[
\frac{n_{2}+p_{\max}-1}{2}\,N_{B}-e_{B}
=\sum_{i\in B}|P_{i}|\left(\frac{n_{2}-2\ell_{i}}{2}
+\frac{p_{\max}-|P_{i}|}{2}\right)\geq 0,
\]
which holds since it is a sum of non-negative values, each factor $|P_{i}|$ is
non-negative, and for $i\in B$ both $n_{2}-2\ell_{i}\geq 0$ (as
$\ell_{i}\leq\lfloor n_{2}/2\rfloor$) and $p_{\max}-|P_{i}|\geq 0$. Dividing by
$N_{B}>0$ yields $e_{B}/N_{B}\leq\frac{n_{2}+p_{\max}-1}{2}$. Combining with the
above and recalling $n=n_{2}+N$,
\[
\frac{n_{2}+N-r}{2}\;<\;d(G)\;<\;\frac{n_{2}+p_{\max}-1}{2},
\]
so that $N<r+p_{\max}-1$. But one leaf has $p_{\max}$ vertices and each of the
remaining $r-1$ leaves has at least one vertex, so $N\geq p_{\max}+(r-1)$, a
contradiction.
Hence one of the partitions has density at least the density of $G$.

\end{proof}

\begin{restatable}{lemma}{wellpartitioneddiameter}\label{lemma: well-partitioned diameter}[$\star$]
    Let $G$ be a well partitioned chordal graph and $\mathcal P$ a partition of $V(G)$. Let $C$ be a part of $\mathcal P$ of diameter at most $2$. Then $C$ can be partitioned to generalized stars and cliques whose densities sum up to at least the density of $C$.
\end{restatable}



In order to finally prove Theorem~\ref{thm::genstars}, it only remains to show that there is always an optimal partition for MDGP such that each part has diameter $2$.

\begin{restatable}{lemma}{wellpartitioneddiametertwo}\label{lemma: well-partitioned diameter 2}[$\star$]
    Let $G$ be a well partitioned chordal graph and $\mathcal P$ a partition of $V(G)$. Let $C$ be a part of $\mathcal P$ of diameter at least $3$. Then we can partition $C$ into parts of diameter at most $2$ whose densities sum up to at least the density of $C$.
\end{restatable}

It is easy to see that the density of a clique of size $l$ equals to $(l-1)/2$. For the case of a generalized star however, things are a bit more complicated as we need to take into account how many edges connect each leaf to the center, which may vary. Luckily this is not the case for thick trees.
It is important to note, and crucial to the algorithm of the next section being polynomial, that when we restrict these results to the class of thick trees, because each pair of minimal separator cliques is either equal or disjoint then each vertex of each leaf of a generalized star is adjacent to all vertices of the central clique.

\subparagraph*{Density of generalized stars on thick trees} As already mentioned when we work on thick trees, every leaf clique of a generalized star is completely adjacent to the center clique. Hence, in order to compute the density of a generalized star it suffices to know the size of its center $l$ and the size of each of its $r$ leaves, $l_1,\ldots, l_r$. We denote the overall leaf vertices by $L=\sum_{i=1}^rl_i$ and the overall leaf edges by $Q=\sum_{i=1}^r\binom{l_i}{2}$. Then each such star has $L+l$ vertices and $Q+lL+\binom{l}{2}$ edges. Hence the density of such a generalized star is $\frac{Q+lL+\binom{l}{2}}{L+l}$.

\section{MDGP on Thick Forests}\label{section: algorithm thick trees}
Since we now know exactly the structure of an optimal solution of MDGP on well partitioned chordal graphs, it is tempting to try to generalize the algorithm of Hanaka et al.~\cite{DBLP:journals/jco/HanakaIO25} by replacing stars by generalized stars and vertex separators by clique separators.  
This however would require to consider all the possible intersections of the separator, which makes the running time exponential. In this section we develop in detail a polynomial time algorithm for dense partition on thick trees and then the desired algorithm for thick forests is immediate since a graph $G$ is a thick forest if each connected component of $G$ is a thick tree. Hence, we may apply the algorithm of this section in each connected component of a thick forest $G$ separately and sum up the densities afterwards. In order to develop the desired algorithm we generalize the algorithm in \cite{DBLP:journals/jco/HanakaIO25} for block graphs in the following manner. 

Our dynamic programming algorithm for thick trees follows a similar structure as the dense partition algorithm for block graphs, but also extends it to handle clique separators and generalized star parts. As in the block graph setting, the computation is carried out on a block-cut tree decomposition, and the focus of each subtree is on how its solution may interact with the rest of the graph. In block graphs, every such interaction goes through a cut vertex and the dynamic programming states encode the role of that vertex. 
In thick trees, cut vertices are replaced by clique separators, and the corresponding dynamic programming states therefore counts how many vertices of the clique separator participate in a part that continues upward (with respect to the block-cut tree) instead of the role of a single vertex.

More importantly, the non-clique parts are now generalized stars whose centers are cliques. We prove in Lemma~\ref{lemma: center on separator} that whenever such a star crosses a separator its center lies inside that separator and hence every generalized star can be regarded as centered at a cut node and having at most one leaf in each block incident to that cut node. Our states store a generalized star at the cut node that owns its center, keeping only the total number of vertices and edges contributed by its leaves so far, rather than a single reserved leaf clique per cut.

Another difference to the algorithm for block graphs is the non-clique parts. 
In block graphs, every non-clique part crossing a block boundary is a star with a single vertex center and it suffices to track the number of leaves. 
In thick trees, optimal solutions might contain generalized stars (stars whose centers are cliques and whose leaves are cliques too, possibly varying in sizes).
Our dynamic programming algorithm replaces the parameter that counts the leaves of each star by storing the information of the size of the center of each generalized star and also the information needed from each leaves to compute the density when the generalized star is finalized (number of vertices and edges).
Moreover, as in the block graph algorithm, we prove that at each block it is sufficient to consider at most one generalized star whose center might extend or be connected to some leaf across the boundary. This preserves the same tree structure that underlies the correctness and efficiency of the original block graph algorithm, while also allowing it to operate on the broader class of thick trees (and hence thick forests). In order to formally state the algorithm we first need to describe in detail the block-cut tree of a thick tree.

\paragraph*{Block-cut tree} Let $G$ be a thick tree. We represent $G$ by a block-cut tree $\mathcal T=(\mathcal X, \mathcal E)$, where the node set $\mathcal X$ is the disjoint union of a set $\mathcal B$ of \emph{block nodes} and a set $\mathcal C$ of \emph{cut nodes}. Each block node $B\in\mathcal B$ corresponds to a maximal clique of $G$, while each cut node $C\in\mathcal C$ corresponds to a minimal clique separator of $G$. The tree $\mathcal T$ is defined so that a block node $B$ is adjacent to a cut node $C$ if and only if the maximal clique of $B$ properly contains the separator clique of $C$. Because two minimal separators of a thick tree are either equal or disjoint, this incidence structure is a tree.

In Appendix~\ref{appendix:thick trees} we provide an algorithm that given a thick tree constructs its block-cut tree in polynomial time. 
 See Figure~\ref{fig: block-cut tree} for an example of a block-cut tree of a thick tree. We now proceed in analyzing the notation on the block-cut tree of some graph that we use throughout our algorithm for computing the desired partition on some thick tree.

\begin{figure}
	\begin{center}
		\begin{tikzpicture}[every node/.style={fill=black,inner sep=1.5pt},scale=0.8]
			
			\node [shape=circle, white] (ph) at (0,-6.5){};
			\node[shape=circle,draw=black] (a) at (0,0) {};
			\node[shape=circle,draw=black] (b) at (1,0) {};
			\node[shape=circle,draw=black] (c) at (2,-1) {};
			\node[shape=circle,draw=black] (d) at (1,-2) {};
			\node[shape=circle,draw=black] (e) at (0,-2) {};
			\node[shape=circle,draw=black] (f) at (-1,-1) {};
			\node[shape=circle,draw=black] (g) at (-2,-2.5) {};
			\node[shape=circle,draw=black] (h) at (-1,-3) {};
			\node[shape=circle,draw=black] (i) at (0,-3) {};
			\node[shape=circle,draw=black] (j) at (-1.5,-4) {};
			\node[shape=circle,draw=black] (k) at (-0.5,-5) {};
			\node[shape=circle,draw=black] (l) at (0.5,-4) {};
			\node[shape=circle,draw=black] (m) at (3,-0.5) {};
			\node[shape=circle,draw=black] (n) at (3,-1.5) {};
			
			\path [-] (a) edge (b);
			\path [-] (a) edge (c);
			\path [-] (a) edge (d);
			\path [-] (a) edge (e);
			\path [-] (a) edge (f);
			\path [-] (b) edge (c);
			\path [-] (b) edge (d);
			\path [-] (b) edge (e);
			\path [-] (b) edge (f);
			\path [-] (b) edge (m);
			\path [-] (b) edge (n);
			\path [-] (c) edge (d);
			\path [-] (c) edge (e);
			\path [-] (c) edge (f);
			\path [-] (c) edge (m);
			\path [-] (c) edge (n);
			\path [-] (d) edge (e);
			\path [-] (d) edge (f);
			\path [-] (d) edge (m);
			\path [-] (d) edge (n);
			\path [-] (e) edge (f);
			\path [-] (e) edge (g);
			\path [-] (e) edge (h);
			\path [-] (e) edge (i);
			\path [-] (f) edge (g);
			\path [-] (f) edge (h);
			\path [-] (f) edge (i);
			\path [-] (m) edge (n);
			\path [-] (h) edge (i);
			\path [-] (h) edge (j);
			\path [-] (h) edge (k);
			\path [-] (h) edge (l);
			\path [-] (i) edge (j);
			\path [-] (i) edge (k);
			\path [-] (i) edge (l);
			\path [-] (j) edge (l);
			\path [-] (j) edge (k);
			\path [-] (k) edge (l);
		\end{tikzpicture}
		\begin{tikzpicture}[every node/.style={fill=black,inner sep=1.5pt},scale=0.5]
			\node [shape=circle, white] (ph) at (0,-6){};
			\node [shape=circle, white] (ph1) at (1,-6){};
		\end{tikzpicture}
		\begin{tikzpicture}[every node/.style={fill=black,inner sep=1.5pt},scale=0.5]
			
			\node [shape=circle, white,label=above:{$B_r$}] (ph) at (0.5,1){};
			\node[shape=circle,draw=black, red] (a) at (0,0) {};
			\node[shape=circle,draw=black, red] (b) at (1,0) {};
			\node[shape=circle,draw=black, red] (c) at (2,-1) {};
			\node[shape=circle,draw=black, red] (d) at (1,-2) {};
			\node[shape=circle,draw=black, red] (e) at (0,-2) {};
			\node[shape=circle,draw=black, red] (f) at (-1,-1) {};
			
			\path [-] (a) edge (b);
			\path [-] (a) edge (c);
			\path [-] (a) edge (d);
			\path [-] (a) edge (e);
			\path [-] (a) edge (f);
			\path [-] (b) edge (c);
			\path [-] (b) edge (d);
			\path [-] (b) edge (e);
			\path [-] (b) edge (f);
			\path [-] (c) edge (d);
			\path [-] (c) edge (e);
			\path [-] (c) edge (f);
			\path [-] (d) edge (e);
			\path [-] (d) edge (f);
			\path [-] (e) edge (f);
			
			\draw (0.5,-1) circle[radius=1.8];
			
			\node[shape=circle,draw=black] (b') at (3.5,0) {};
			\node[shape=circle,draw=black] (c') at (4.5,-1) {};
			\node[shape=circle,draw=black] (d') at (3.5,-2) {};
			
			\path [-] (b') edge (c');
			\path [-] (c') edge (d');
			\path [-] (d') edge (b');
			
			\begin{scope}[xshift=4cm, yshift=-1cm, rotate=90]
				\node[draw, fill=none, minimum width=3cm, minimum height=2cm, transform shape, green] {};
			\end{scope}
			
			\path [thick] (2.3,-1) edge (3,-1);
			\path [thick] (5,-1) edge (5.7,-1);
			\path [thick] (-1.56,-2.51) edge (-1.03,-1.93);
			
			\node[shape=circle,draw=black] (b'') at (6.5,0) {};
			\node[shape=circle,draw=black] (c'') at (7.3,-1) {};
			\node[shape=circle,draw=black] (d'') at (6.5,-2) {};
			\node[shape=circle,draw=black, red] (n) at (8.5,-0.5) {};
			\node[shape=circle,draw=black, red] (m) at (8.5,-1.5) {};
			
			\draw (7.3,-1) circle[radius=1.6];
			
			\path [-] (b'') edge (c'');
			\path [-] (b'') edge (d'');
			\path [-] (b'') edge (n);
			\path [-] (b'') edge (m);
			\path [-] (c'') edge (d'');
			\path [-] (c'') edge (n);
			\path [-] (c'') edge (m);
			\path [-] (d'') edge (n);
			\path [-] (d'') edge (m);
			\path [-] (m) edge (n);
			
			\node[shape=circle,draw=black] (e') at (-1.5,-3.5) {};
			\node[shape=circle,draw=black] (f') at (-2.5,-2.5) {};
			
			\path [-] (e') edge (f');
			
			\begin{scope}[xshift=-2cm, yshift=-3cm, rotate=-45]
				\node[draw, fill=none, minimum width=2.3cm, minimum height=1.3cm, transform shape, green] {};
			\end{scope}
			\path [thick] (-0.35,-4.5) edge (-1.17,-3.79);
			\path [thick] (-3,-3.95) edge (-2.5,-3.4);
			
			\node[shape=circle,draw=black] (e'') at (-3,-5) {};
			\node[shape=circle,draw=black] (f'') at (-4,-4) {};
			\node[shape=circle,draw=black, red] (g) at (-4.5,-5) {};
			
			\path [-] (e'') edge (f'');
			\path [-] (e'') edge (g);
			\path [-] (g) edge (f'');
			
			\draw (-3.8,-4.7) circle[radius=1.1];
			
			\node[shape=circle,draw=black] (e''') at (0,-4.5) {};
			\node[shape=circle,draw=black] (f''') at (1,-4.5) {};
			\node[shape=circle,draw=black, red] (h) at (0,-5.5) {};
			\node[shape=circle,draw=black, red] (i) at (1,-5.5) {};
			
			\path [-] (e''') edge (f''');
			\path [-] (e''') edge (h);
			\path [-] (e''') edge (i);
			\path [-] (f''') edge (h);
			\path [-] (f''') edge (i);
			\path [-] (h) edge (i);
			
			\draw (0.5,-5) circle[radius=1];
			
			\node[shape=circle,draw=black] (h') at (0,-7.5) {};
			\node[shape=circle,draw=black] (i') at (1,-7.5) {};
			
			\path [-] (h') edge (i');
			
			\begin{scope}[xshift=0.5cm, yshift=-7.5cm, rotate=90]
				\node[draw, fill=none, minimum width=1.3cm, minimum height=2.3cm, transform shape] {};
			\end{scope}
			
			\path [thick] (0.5,-6.85) edge (0.5,-6);
			
			\node[shape=circle,draw=black] (h'') at (0,-9.5) {};
			\node[shape=circle,draw=black] (i'') at (1,-9.5) {};
			\node[shape=circle,draw=black, red] (j) at (-0.5,-10.5) {};
			\node[shape=circle,draw=black, red] (k) at (1.5,-10.5) {};
			\node[shape=circle,draw=black, red] (l) at (0.5,-11.5) {};
			
			\path [-] (h'') edge (i'');
			\path [-] (h'') edge (j);
			\path [-] (h'') edge (k);
			\path [-] (h'') edge (l);
			\path [-] (i'') edge (j);
			\path [-] (i'') edge (k);
			\path [-] (i'') edge (l);
			\path [-] (j) edge (k);
			\path [-] (j) edge (l);
			\path [-] (k) edge (l);
			
			\draw (0.5,-10.5) circle[radius=1.5];
			\path [thick] (0.5,-9) edge (0.5,-8.15);
		\end{tikzpicture}
	\end{center}
	\caption{Block-cut tree $\mathcal T$ of a thick tree. Block nodes correspond to maximal cliques (circle bags). Cut nodes correspond to clique separators (rectangle bags). Distinct separator cliques are disjoint. Red vertices correspond to the vertices assigned to each bag (their ``highest" appearance on the tree). The bag $B_r$ is the root of the tree and the green rectangles the set $\text{Child}_\mathcal T(B_r)$.}
	\label{fig: block-cut tree}
\end{figure}


The block-cut tree $\mathcal T$ is a tree, and every vertex of $G$ appears in a connected set of nodes of $\mathcal T$. We root $\mathcal T$ at an arbitrary block node $B_r$. For a node $x\in\mathcal X$, we denote by $p_{\mathcal T}(x)$ its parent in the rooted tree (unless $x=B_r$), by $\text{Child}_{\mathcal T}(x)$ the set of its children (unless $x$ is a leaf), and by $\mathcal T_x$ the subtree rooted at $x$. Moreover, we denote by $V_x$ the set of vertices of $G$ that appear in nodes of $\mathcal T_x$.

For a block node $B\neq B_r$ we write $C_P(B)=p_{\mathcal T}(B)$ for its parent cut node. This is a clique separator which contains all the edges $V(\mathcal T_B)$ and the rest of the graph, so there are no edges between different child subtrees of a cut node except through its vertices. 
Since a vertex of a minimal separator clique belongs to several adjacent blocks, we assign every vertex to a unique node of $\mathcal T$ in order to avoid double counting: a vertex that lies in a minimal separator is assigned to the corresponding cut node, and every other vertex is assigned to the unique block whose maximal clique contains it. For a block node $B$ we denote by $R(B)$ the set of vertices assigned to $B$. These vertices are called \emph{private} vertices of $B$, that is, the vertices of the clique of $B$ that belong to no separator. For any subtree $\mathcal T_x$, let $A_x$ be the set of vertices that are assigned to $x$ or to some descendant of $x$. The values of our dynamic programming states are computed using the vertices of $A_x$, which ensures that each vertex is considered exactly once.

\paragraph*{Structure of an optimal partition}
Recall from Section~\ref{section: well-partitioned chordal} since thick trees form a subclass of well partitioned graphs, that there is an optimal partition of $G$ for MDGP in which every part induces either a clique or a generalized star.
In the following lemma we show that the center of each generalized star is entirely contained in a cut node.

\begin{restatable}{lemma}{centeronseparator}\label{lemma: center on separator}[$\star$]
	Let $P$ be a part of the optimal partition that induces a generalized star with at least two leaves, and let $P_0$ be its center. Then there is a minimal clique separator $C$ of $G$ with $P_0\subseteq C$, and every leaf of $P$ lies in a distinct block incident to the cut node of $C$.
\end{restatable}

By Lemma~\ref{lemma: center on separator} we may assume that the center of every generalized star is a subclique of some cut node $C$, and that the star has at most one leaf in each of the blocks incident to $C$ (the child blocks of $C$ and its parent block). 
In particular a clique is contained in a single block, and the only way a part crosses a separator $C$ is either as a clique that uses some vertices of $C$ but is finalized inside one incident block, or as the generalized star centered at $C$, whose leaves below are in some child blocks of $C$ and whose one additional leaf, if any, is supplied by the parent block of $C$. 
This, together with the following lemma, which proves that at most one of a block node's children can have vertices that are the center of a generalized star with leaves both in the upward and in the downward part of the block-cut tree, is what allows the polynomial size of our states.

\begin{restatable}{lemma}{atmostone}\label{lemma: at most one}[$\star$]
	Let $B\in \mathcal B$ be a block node of $\mathcal T$. Then there exists a partition that satisfies the optimal dense partition value of $G$ such that for at most one node $u\in\text{Child}_{\mathcal T}(B)$, the center of a generalized star that has leaves in both $V_B\setminus B$ and $B$ lies in $u$.
\end{restatable}

We denote by $\Phi(l,L,Q):=\frac{\binom{l}{2}+lL+Q}{l+L}$ the density of a generalized star whose center clique has $l$ vertices and whose leaves have $L$ vertices and $Q$ internal edges in total, every leaf being completely joined to the center. Whenever a set of $h$ vertices of a single block remains unused by the reservations and star structures described below, it is optimal to keep them as one clique part. We denote its contribution by $\text{Left-clq}(h)$, which is $0$ if $h=0$ and $\frac{h-1}{2}$ if $h\geq 1$.

For the detailed presentation of the algorithm and the proof of correctness and runtime analysis we refer the reader to Appendix~\ref{appendix:thick trees}.

\section{$\text{MDGP}$ is $\mathsf{NP}$-hard on Split Graphs}\label{section: NP-hardness}
First we prove that given a complete split graph, the density of the whole graph increases with the number of the vertices on the independent set. 
\begin{restatable}{lemma}{completesplitdensity}\label{lemma: complete split density}[$\star$]
    Let $G=(K\cup I, E)$ be a complete split graph where $K$ is a clique and $I$ an independent set. If $G'=(K\cup (I\cup\{v\}), E\cup\{\{v,u\}\mid u\in V(K)\})$ is the graph obtained from $G$ by adding one more vertex to the independent set and connecting it to all of $K$ then $d(G')>d(G)$.  
\end{restatable}

We proceed in expressing the density of a complete split graph $G=(K\cup I, E)$ with an independent set $I$, proportional to the size of the clique $K$ within a factor $\alpha$ as a sum of a factor dependent to both $\alpha$ and $|K|$ and one only dependent on $\alpha$.
 \begin{restatable}{lemma}{splitdense}\label{lem::splitdense}[$\star$]
    Let $G=(K\cup I, E)$ be a complete split graph with $|K|=n$ and $|I|=\alpha n$ for some $\alpha >0$. 
    Then $d(G)=\frac 12(n(1+\frac{\alpha}{1+\alpha})-\frac{1}{1+\alpha})$.
\end{restatable}

To better analyze the density of a split graph $G$ and its clustering, we break the expression from Lemma~\ref{lem::splitdense} into the part that depends on $n$, denoted by $d_l$, and the part that does not, denoted by $d_c$. In particular, given a complete split graph $G=(K\cup I,E)$ such that $K$ is a clique of size $n$ and $I$ is an independent set of size $\alpha n$ for some constant $\alpha$, denoting $d_l(n,\alpha)=\frac{1}{2}n(1+\frac{\alpha}{1+\alpha})$ and $d_c(\alpha)=-\frac{1}{2(1+\alpha)}$.
If it is clear from the context we will omit the index that indicates the graph the functions $d_l$ and $d_c$ are applied on.

Looking at the proof of Lemma~\ref{lem::splitdense}, we can directly express the effect that missing edges have of a (non-complete) split graph with clique of size $n$ and independent set of size $\alpha n$ with the functions $d_l$ and $d_c$ as follows.

 \begin{observation}\label{obs::splitdense_missing}
    If $G=(K\cup I, E)$ is a split graph with $|K|=n$, $|I|=\alpha n$ for some $\alpha >0$, and $|E\cap (K\times I)|=\alpha n^2-r$, then
$d(G)=d_l(n,\alpha)+d_c(\alpha)-\frac{r}{n+\alpha n}$.     
\end{observation}

Our goal now is to obtain explicit expressions of the density change when we partition a complete split graph into more than one part. We consider the case of a partition that breaks the clique-part into two non-empty sets of sizes $n_1$ and $n_2$, respectively. Due to Lemma~\ref{lemma: complete split density}, we know that the sum of the densities is higher if we include all of the vertices of the independent set. Consider any distribution of the independent-set vertices to the two parts. If breaking in two like this worsens the density, we can conclude by iterative application of this argument that any clustering with more than one cluster has worse density. 

 \begin{restatable}{lemma}{splitdensetwo}\label{lem::splitdense2}[$\star$]
    Let $G=(K\cup I, E)$ be a complete split graph with $|K|=n$ and $|I|=\alpha n$ for some $\alpha >0$ and let $n_1,n_2\in\mathbb N$, and $\alpha_1,\alpha_2\geq 0$ be such that $n_1+n_2=n$ and $ \alpha_1n_1+\alpha_2n_2=\alpha n$. For the complete subgraphs $G_i=(K_i\cup I_i, E_i)$ with $|K_i|=n_i$ and $|I_i|=\alpha_i n_i$:
    \begin{enumerate}
     \item If $\alpha=\alpha_1=\alpha_2$, then $d_l(n,\alpha)-d_l(n_1,\alpha_1)-d_l(n_2,\alpha_2)=0$. 
    \item If   $\alpha\neq \alpha_1$, then $d_l(n,\alpha)-d_l(n_1,\alpha_1)-d_l(n_2,\alpha_2)=c_1n_1=c_2n_2$ for some positive constants $c_1,c_2$ that depend only on $\alpha,\alpha_1,\alpha_2$. 
     \end{enumerate}
     Specifically, it holds that $c_1=\gamma(\alpha_2-\alpha_1)(\alpha-\alpha_1)$ and $c_2=\gamma(\alpha_2-\alpha_1)(\alpha_2-\alpha)$,  where $\gamma=\gamma(\alpha,\alpha_1,\alpha_2)=\frac{1}{2(1+\alpha)(1+\alpha_1)(1+\alpha_2)}$.
\end{restatable}

Looking at the definition of $d_c(\alpha)$, we 
see immediately that 
$d_c(\alpha)-d_c(\alpha_1)-d_c(\alpha_2)>0$ for 
all $\alpha_1\leq\alpha\leq \alpha_2$ thus by 
Lemma~\ref{lem::splitdense2} we can conclude that 
no partition into more than one part of a 
complete split graph is optimal. Formally:
\begin{observation}\label{obs::onepart}
For any complete split graph $G=(V,E)$ and any partition $\mathcal P$ of $V$ with $|\mathcal P|\geq 2$, we have $d(G)>d(\mathcal P)$.
\end{observation}

It is tempting to think that complete split graphs behave similar to complete bipartite graphs in the sense that proportional partitions are always best. More precisely, when searching for a best choice of $n_1,n_2,\alpha_1,\alpha_2$ to partition a complete split graph $G=(K\cup I, E)$  with $|K|=n$ and $|I|=\alpha n$ into $G_i=(K_i\cup I_i, E_i)$ with $|K_i|=n_i$ and $|I_i|=\alpha_i n_i$, one would expect that such an optimal choice comes from picking $\alpha=\alpha_1=\alpha_2$. This seems to be supported by  Lemma~\ref{lem::splitdense2} which suggests that the loss incurred by the  $d_l$-parts grows with the number of vertices while $d_c\in (0,1]$. This however is unfortunately not true, as the constants $c_i$ that only depend on $\alpha,\alpha_1,\alpha_2$ can grow with the number of vertices as well. A counterexample to the proportional intuition is the complete split graph with $n=3$ and $\alpha=5$ where the unproportional split with $n_1=2$ and $\alpha_1=\frac{9}{2}$ and $n_2=1$ and $\alpha_2=6$ (density $\frac{199}{77}$) is better than the proportional $n'_1=2$, $n'_2=1$ and $\alpha'_1=\alpha'_2=5$  (density $\frac{31}{12}$). 

The overall intuition that a best partition is at least almost proportional however is enough to construct our $\mathsf{NP}$-hardness reduction. Since the $d_l$-part of the density difference between keeping $G$ and splitting it is positive unless the split is proportional, we aim in our $\mathsf{NP}$-hardness proof to link a dominating set to a proportional split. Such proportional splits also give an easy expression for the change with respect to the function $d_c$.

Formally, we call a partition $\{P_1,\dots, P_k\}$ of any split graph $G=(K\cup I, E)$ with $|K|=n$ and $|I|=\alpha n$ \emph{proportional}, if $\frac{|P_i\cap I|}{|P_i\cap K|}=\alpha$ for all $1\leq i\leq k$. The vanishing $d_l$-parts for proportional splits by  Lemma~\ref{lem::splitdense2} immediately give us the following.

 \begin{observation}\label{obs::splitdense_proportional}
    Let $G=(K\cup I, E)$ be a split graph with $|K|=n$ and $|I|=\alpha n$ for some $\alpha >0$. Any proportional partition $\mathcal P$ of $G$ into $k$ complete split graphs satisfies: 
     $d(\mathcal P)=d_l(n,\alpha)+kd_c(\alpha)$.  
\end{observation}
Aiming to use this proportional split value, we name
\[\rho_k(n,\alpha)=d_l(n,\alpha)+kd_c(n,\alpha)=n-\frac{n+k}{2(1+\alpha)}\]

In order to control the number of parts that can reach this target density, we need an additional structural lemma.
We call a part of a partition \emph{non-trivial} if its induced subgraph contains at least one edge; note that for split graphs, this implies that such a part contains at least one clique-vertex.

\begin{restatable}{lemma}{kparts}[$\star$]\label{lemma::kparts} For any split graph $G=(K\cup I,E)$ with $|K|=n$, $|I|=\alpha n$ with $\alpha>0$ and any $k$ with $k+1<n$, there is no partition into more than $k$ non-trivial parts that reaches a density of at least $\rho_k(n,\alpha)$.
\end{restatable}

\splitnph*
 We give a reduction from the $\mathsf{NP}$-hard problem  {\sc Dominating Set} to  $\text{MDGP}$. In the  {\sc Dominating Set} problem, the input is a pair $(G,k)$ where $G=(V,E)$ is a graph and $k\in\mathbb N$, and the task is to decide if there exists a subset $C\subseteq V$ with $|C|\leq k$ such that $V=N[C]$. Indeed, {\sc Dominating Set} is $\mathsf{NP}$-hard \cite{DBLP:conf/coco/Karp72}.
 
Let $(G,k)$ with $G=(V,E)$ and $V=\{v_1,\dots,v_n\}$ be an instance of {\sc Dominating Set}. We assume without loss of generality that $k\leq \frac{n}{2}$ (observe that if $k> \frac{n}{2}$, we can delete isolated vertices decrementing $k$ and simply reply yes if some  $(G',k')$ remains where the minimum degree in $G'$ is 1 and the number of vertices is still less than $2k$).
  As reduction, we construct a split graph $G':=(K\cup I,E')$ as instance of $\text{MDGP}$ as follows. A full proof can be found in Appendix~\ref{section: NP-hardnessAP}.

\begin{proof}[Sketch of Proof of Theorem~\ref{thm::splitnph}]
Let $(G,k)$ with $G=(V,E)$ and $V=\{v_1,\dots,v_n\}$ be an instance of {\sc Dominating Set}. We assume without loss of generality that $k\leq \frac{n}{2}$.
  As reduction, we construct a split graph $G':=(K\cup I,E')$ as instance of $\text{MDGP}$.

Overall idea: The sets of vertices of $G'$ contain copies of $V$ on both sides plus $n$ large sets of filler-vertices to force a certain structure of any (almost) proportional split. An illustration of this construction can be found in Figure~\ref{npfigESA}.

The edge set of this graph is the union of the following sets, where we use values $y=n^{20}$, $h=n^4$, $r=n^7$ and $r'=n^{14}$:
\begin{eqnarray*}
    E_{K} & = & \{\{u,v\}\mid u,v\in {K}\}\\
    E_{V} & = & \{\{u,v\}\mid u\in V^{K},v\in Y\cup Y'\cup Z\}\\
    E_{d} & = & \{\{v_i^t,\overline{v}_j^s\}\mid \{v_i,v_j\}\in E, t\in [r], j\in [r'] \text{ or } i=j \}\\
    E_{V^{I}w} & = & \{\{v,w\}\mid v\in V^{I},w\in W\}\\
    E_{wy'} & = &\{\{w,y'\}\mid w\in W, y'\in Y'\}\\
    E_{wy} & = & \bigcup_{i\in [k]} \{\{w_{i}^{j},y_{i}^{t}\}\mid j\in [h], t\in [y]\}\\
    E_{xz} & = & \bigcup_{i\in [n-k]} \{\{x_{i}^{j},z_{i}^{t}\}\mid j\in [h], t\in [nr'+y]\}
\end{eqnarray*}

Given all the chosen values, the construction can be done in polynomial time. Moreover,
observe that $|K|=(r+h)n$ and $|I|=(nr'+y)n$ and therefore $\alpha:= \frac{|I|}{|K|}=\frac{r'n+y}{r+h}$. This concludes the construction.

\begin{figure}[!h]
\begin{center}
\begin{tikzpicture}[every node/.style={fill=black,inner sep=1.5pt}, scale=0.8]
\draw [red, thick] (0,4)--(6,0);
\draw [red, thick] (0,6.6)--(6,3);
\draw [red, thick] (0.3,4.9)--(5.7,1.3);
\draw [red, thick] (0.3,5.5)--(5.7,2);
\draw [black, thick] (0,8.3)--(6,5.5);
\draw [black, thick] (0,10.6)--(6,9);
\draw [black, thick] (0.3,9)--(5.7,7);
\draw [black, thick] (0.3,9.6)--(5.7,7.8);
\draw [black, thick] (0,12.1)--(6,12);
\draw [black, thick] (0,14.3)--(6,16.5);
\draw [black, thick] (0.3,13)--(5.7,13.8);
\draw [black, thick] (0.3,13.5)--(5.7,14.5);
\draw [black!30, thick] (0.2,7.9)--(6,20);
\draw [black!30, thick] (0,10.6)--(6,20);
\draw [black!30, thick] (0.3,9)--(6,20);
\draw [black!30, thick] (0.3,9.6)--(6,20);
\draw [black!30, thick] (0,8.3)--(6,0);
\draw [black!30, thick] (0,8.3)--(6,3);
\draw [black!30, thick] (0,8.3)--(5.7,1.5);
\draw [black!30, thick] (0,10.6)--(6,0);
\draw [black!30, thick] (0,10.6)--(6,3);
\draw [black!30, thick] (0,10.6)--(5.7,1.5);
\draw [black!30, thick] (0.3,9)--(6,0);
\draw [black!30, thick] (0.3,9)--(6,3);
\draw [black!30, thick] (0.3,9)--(5.7,1.5);
\draw [black!30, thick] (0.3,9.6)--(6,0);
\draw [black!30, thick] (0.3,9.6)--(6,3);
\draw [black!30, thick] (0.3,9.6)--(5.7,1.5);
\filldraw[fill=white, draw=black] (0,4) ellipse (0.55cm and 0.7cm);
\node[fill=none, draw=none, text=black] at (0,4) {$V_1^K$};
\node[fill=none, draw=none, text=black] at (0,5.1) {$\cdot$};
\node[fill=none, draw=none, text=black] at (0,5.3) {$\cdot$};
\node[fill=none, draw=none, text=black] at (0,5.5) {$\cdot$};
\filldraw[fill=white, draw=black] (0,6.6) ellipse (0.55cm and 0.7cm);
\node[fill=none, draw=none, text=black] at (0,6.6) {$V_n^K$};
\draw [decorate,decoration={brace,amplitude=8pt}]
      (-0.5,3.1) -- (-0.5,7.4);
\node[fill=none, draw=none, text=black] at (-2.8,5.6) {$V^K=\cup_{i=1}^nV_i^K=$};
\node[fill=none, draw=none, text=black] at (-2.8,4.9) {$\cup_{i=1}^n\{v_i^j|1\leq j\leq r\}$};
\filldraw[fill=white, draw=black] (0,8.3) ellipse (0.55cm and 0.55cm);
\node[fill=none, draw=none, text=black] at (0,8.3) {$W_1$};
\node[fill=none, draw=none, text=black] at (0,9.2) {$\cdot$};
\node[fill=none, draw=none, text=black] at (0,9.4) {$\cdot$};
\node[fill=none, draw=none, text=black] at (0,9.6) {$\cdot$};
\filldraw[fill=white, draw=black] (0,10.6) ellipse (0.55cm and 0.55cm);
\node[fill=none, draw=none, text=black] at (0,10.6) {$W_k$};
\draw[decorate, decoration={brace, amplitude=8pt}]
  (-0.5,7.7) -- (-0.5,11.2);
\node[fill=none, draw=none, text=black] at (-2.8,9.75) {$W=\cup_{i=1}^k W_i$=};
\node[fill=none, draw=none, text=black] at (-2.8,9.05) {$\cup_{i=1}^k\{w_i^j|1\leq j\leq h\}$};
\filldraw[fill=white, draw=black] (0,12.1) ellipse (0.55cm and 0.55cm);
\node[fill=none, draw=none, text=black] at (0,12.1) {$X_1$};
\node[fill=none, draw=none, text=black] at (0,13) {$\cdot$};
\node[fill=none, draw=none, text=black] at (0,13.2) {$\cdot$};
\node[fill=none, draw=none, text=black] at (0,13.4) {$\cdot$};
\filldraw[fill=white, draw=black] (0,14.3) ellipse (0.55cm and 0.55cm);
\node[fill=none, draw=none, text=black] at (0,14.3) {$X_{n-k}$};
\draw[decorate, decoration={brace, amplitude=8pt}]
  (-0.5,11.5) -- (-0.5,14.9);
\node[fill=none, draw=none, text=black] at (-2.8,13.5) {$X=\cup_{i=1}^{n-k} X_i$=};
\node[fill=none, draw=none, text=black] at (-2.8,12.8) {$\cup_{i=1}^{n-k}\{x_i^j|1\leq j\leq h\}$};
\node[fill=none, draw=none, text=black] at (0,-2) {clique $K$};
\node[fill=none, draw=none, text=black] at (6,-2) {independent set $I$};
\filldraw[fill=white, draw=black] (6,0) ellipse (0.55cm and 0.9cm);
\node[fill=none, draw=none, text=black] at (6,0) {$V_1^I$};
\node[fill=none, draw=none, text=black] at (6,1.3) {$\cdot$};
\node[fill=none, draw=none, text=black] at (6,1.5) {$\cdot$};
\node[fill=none, draw=none, text=black] at (6,1.7) {$\cdot$};
\filldraw[fill=white, draw=black] (6,3) ellipse (0.55cm and 0.9cm);
\node[fill=none, draw=none, text=black] at (6,3) {$V_n^I$};
\draw[decorate, decoration={brace, mirror, amplitude=8pt}]
  (6.5,-1) -- (6.5,4);
\node[fill=none, draw=none, text=black] at (9,1.9) {$V^I=\cup_{i=1}^nV_i^I$=};
\node[fill=none, draw=none, text=black] at (9,1.2) {$\cup_{i=1}^n\{\overline{v}_{i}^{j}|1\leq j\leq r\}$};
\filldraw[fill=white, draw=black] (6,5.5) ellipse (0.55cm and 1.1cm);
\node[fill=none, draw=none, text=black] at (6,5.5) {$Y_1$};
\node[fill=none, draw=none, text=black] at (6,7) {$\cdot$};
\node[fill=none, draw=none, text=black] at (6,7.2) {$\cdot$};
\node[fill=none, draw=none, text=black] at (6,7.4) {$\cdot$};
\filldraw[fill=white, draw=black] (6,9) ellipse (0.55cm and 1.1cm);
\node[fill=none, draw=none, text=black] at (6,9) {$Y_k$};
\draw[decorate, decoration={brace, mirror, amplitude=8pt}]
  (6.5,4.3) -- (6.5,10.2);
\node[fill=none, draw=none, text=black] at (9,7.7) {$Y=\cup_{i=1}^k Y_i=$};
\node[fill=none, draw=none, text=black] at (9,7) {$\cup_{i=1}^k\{{y}_{i}^{j}|1\leq j\leq y\}$};
\filldraw[fill=white, draw=black] (6,12) ellipse (0.55cm and 1.5cm);
\node[fill=none, draw=none, text=black] at (6,12) {$Z_1$};
\node[fill=none, draw=none, text=black] at (6,14) {$\cdot$};
\node[fill=none, draw=none, text=black] at (6,14.2) {$\cdot$};
\node[fill=none, draw=none, text=black] at (6,14.4) {$\cdot$};
\filldraw[fill=white, draw=black] (6,16.5) ellipse (0.55cm and 1.5cm);
\node[fill=none, draw=none, text=black] at (6,16.5) {$Z_{n-k}$};
\draw[decorate, decoration={brace, mirror, amplitude=8pt}]
  (6.5,10.5) -- (6.5,18);
\node[fill=none, draw=none, text=black] at (9.5,14.5) {$Z=\cup_{i=1}^{n-k} Z_i=$};
\node[fill=none, draw=none, text=black] at (9.5,13.7) {$\cup_{i=1}^{n-k}\{{z}_{i}^{j}|1\leq j\leq nr'+y\}$};
\filldraw[fill=white, draw=black] (6,20) ellipse (0.55cm and 1.3cm);
\node[fill=none, draw=none, text=black] at (6,20) {$Y'$};
\node[fill=none, draw=none, text=black] at (9.5,20.35) {$Y'=$};
\node[fill=none, draw=none, text=black] at (9.5,19.65) {$\{{y'}_{i}|1\leq i\leq n(k-1)r'\}$};
\end{tikzpicture}
\caption{Illustration of the construction for Theorem~\ref{thm::splitnph}. Black and gray edges denote complete connections, red edges are drawn according to the edges in $E$. The edges connecting the $v_i^j$ to everything else and the edges of the clique $K$ are omitted.}
\label{npfigESA}
\end{center}
\end{figure}

We claim that $(G,k)$ is a yes-instance for {\sc Dominating Set} if and only if $(G',n,D)$ is a yes-instance for \DGP{} with $D=d_l(|K|,\alpha)+nd_c(\alpha)$. 

\medskip

(\textbf{Forward direction}) For the forward direction, 
let us assume  that $(G,k)$ with $G=(V,E)$ and 
$V=\{v_1,\dots,v_n\}$  and $k<\frac n2$ is a 
yes-instance for  {\sc Dominating Set}. Let
$S\subset V$ be a dominating set for $G$ with 
$|S|=k$. Let $S=\{v_{p_1},\dots, v_{p_k}\}$ and 
let $V_1,\dots, V_k$ be any partition of $V$ such 
that $\{v_{p_i},u\}\in E$ or $u=v_{p_i}$ for each 
$u\in V_i$, for all $1\leq i\leq k$. Such a 
partition exists, because $S$ is a dominating set 
and thus each $u$ in $V$ is equal or adjacent to 
at least one $v_{p_i}$-vertex. 
For every $i\in [k]$, let $l_{i}=|V_{i}|$. Denote 
further 
$V\setminus S=\{v_{p_{k+1}},\dots, v_{p_{n}}\}$. 
Let $Y'_{i}$, $i\in [k]$ be a partition of $Y'$ 
as follows:
\begin{linenomath}
    $$Y'_{i}=\left\{y'_{t}\mid r'\left(n(i-1)-\sum_{p=1}^{i-1}l_{p}\right)+1\leq t \leq r'\left(ni-\sum_{p=1}^{i}l_{p}\right)\right\},$$
\end{linenomath}
where $\sum_{p=1}^{0}l_{p}=0$. That is $Y'_{i}$ contains $(n-l_{i})r'$ elements of $Y'$. Observe that by definition $Y'_{i}\cap Y'_{j}=\emptyset$ for every $i\neq j$. Moreover, since $\sum_{i=1}^{k}|Y_{i}|=|Y'|$, they indeed form a partition of $Y'$.
We construct a partition $\mathcal P= \{D_1,\dots, D_k, P_1,\dots, P_{n-k}\}$  with $n$ sets of density $D$ for the corresponding graph $G'$ as follows. For every $i\in [k]$, let
\begin{linenomath}
$$\displaystyle D_i  =V^{K}_{p_{i}} \cup W_{i}  \cup Y_{i} \cup Y'_{i}\cup(\cup_{v_{t}\in V_{i}}V^{I}_{t}).$$
\end{linenomath}
\noindent Moreover, for every $i\in [n-k]$, let
\begin{linenomath}
$$P_i =V^{K}_{p_{k+i}} \cup X_{i} \cup Z_{i}.$$   
\end{linenomath}
One can prove that $\mathcal P$ is a proportional partition into $n$ complete split graphs, which by Observation~\ref{obs::splitdense_proportional} shows that $d(\mathcal P)=d_l(|K|,\alpha)+nd_c(\alpha)=D$.

Therefore, $(G',n,D)$ is a yes-instance for \DGP{} and this concludes the forward direction of the reduction. 

\medskip

(\textbf{Reverse direction}) Let us now assume that  $(G',n,D)$ is a yes-instance and thus there exists a partition $\mathcal P$ of $G'$ with $|\mathcal P|=n$ and $d(\mathcal P)\geq D$.

We use the estimates form Observation~\ref{obs::splitdense_missing} and  Lemma~\ref{lem::splitdense2} to derive bounds on missing edges and the ratio of each part in $\mathcal P$:
\begin{restatable}{claim}{missingedges}
\label{clm:missingedges}
The following bounds hold:
\begin{itemize}
\item Every part $P$ in $\mathcal{P}$ has strictly less than  $n^9$ missing edges.
\item Every $P$ in $\mathcal{P}$ with $K\cap P\neq \emptyset$ satisfies $\frac{|P\cap I|}{|P\cap K|}\geq \frac{\alpha}{2n}$
\end{itemize}
\end{restatable}

We then use these estimates to establish a particular structure of $\mathcal P$ with respect to the buffer sets $W,X,Y,Z$.

\begin{restatable}{claim}{structural}\label{clm:structural}
The partition $\mathcal{P}$ can be written as $\mathcal P=\{D_1,\dots, D_k, P_{1},\dots, P_{n-k},R\}$ with $W_{i}\subseteq D_{i}$, $i\in [k]$, and $X_{i}\subseteq P_{i}$, $i\in [n-k]$, and $R\subset I$.
Furthermore, these sets admit the following bounds:
\begin{enumerate}
\item $|R|\in\mathcal O(\frac{r'}{n})$,
\item $|Y_{i}\cap D_i|\in \Theta(y)$, $i \in [k]$, and
\item $|Z_{i}\cap P_i|\in \Theta(y)$, $i\in [n-k]$.
\end{enumerate}
\end{restatable}

Finally, we show how the parts of $\mathcal{P}$ interact with the sets $V^{K}$ and $V^{I}$. In particular,
\begin{restatable}{claim}{ViVkbounds}\label{clm:ViVkbounds}
The following bounds hold:
\begin{enumerate}
\item $|P_{i} \cap V^{I}| < n^6$ for all $i\in [n-k]$ and
\item $|P\cap V^{K}| > \frac{r}{n^3}$ for each $P\in\mathcal P$.
\end{enumerate}
\end{restatable}

Assuming the above claims hold, we show how to construct a dominating set of $G$ of size at most $k$.

\medskip
{\bf Creating the Dominating Set:}
 By the second item of Claim~\ref{clm:ViVkbounds} each part $D_i$, $i\in [k]$, contains at least $\frac{r}{n^3}$ vertices from $V^{K}$.
 Let $\mathcal D=\{v_{r_{1}},v_{r_{2
}},\dots, v_{r_{k}}\}$ be some set of $k$ vertices of $G$ such that $v_{r_{i}}^{j}\in D_{i}$ for some $j\in [r]$ (where it might be the case that $v_{r_{i}}=v_{r_{j}}$ for $i\neq j$).

By $|P_{i} \cap V^{I}| < n^6$ from the first item of   Claim~\ref{clm:ViVkbounds}  and $|R|\in\mathcal O(\frac{r'}{n})$ by Claim~\ref{clm:structural}, we get that $\mathcal O(r')$ of the vertices in $V_i^I$ (for each $i$) needs to be contained in the $D_j$-sets. The sizes of these sets, and the bound on the missing edges form Claim~\ref{clm:missingedges} can then be used to show that $v_i$ needs to be adjacent to at least one of the vertices $v_{r_{j}}$. This shows that $\mathcal D$ is a dominating set of $G$,
implying that $(G,k)$ is a yes-instance 
for \textsc{Dominating Set}. 
\end{proof}

\section{Bounded Treewidth Connected Partition}\label{section: treewidth connected partition}
In this section we develop an algorithm such that given a graph $G$ we are able to compute any connected partition of size $k$ with parameters the size of the partition and the treewidth of $G$. Moreover our algorithm receives as an input $k$ integers $n_1,\ldots, n_k$ that sum up to the order of $G$ and are the required size of each part of the desired connected partition.

\conpart*

Adding one specific vertex to each part is possible and as a result this algorithm can also realize Gy\H{o}ri-Lov\'{a}sz theorem on graphs with bounded treewidth. We will also discuss how small adjustments to the arguments used for this result yields a similar runtime for MDGP.

The main difficulty of a dynamic programming algorithm for computing a balanced partition of a graph $G$ is preserving the connectivity of each part even though we process local information at each step. 
Specifically it is possible that a path connecting two vertices has been forgotten making them locally appear as non-connected. In order to tackle this, our algorithm assigns each vertex of a bag in one of the $k$ possible parts and also keeps track of how many vertices belonging in each part have already been forgotten. To ensure that the connectivity information is preserved, every time a vertex is forgotten we {\it force}
some of its neighbors to be a clique by adding respective edges. Join nodes combine two independent subgraphs by taking into account both forgotten vertex counts and adjacency information.

The running time of the algorithm occurs from enumerating all possible ways in which the vertices of a bag can be distributed among $k$ parts of the partition and all possible added edges to encode connectivity inside the bag. Join nodes are particularly expensive, as they require combining states from two subtrees. This however is based on the size of the bags and the numbers of parts. 

Formally, given a graph $G=(V,E)$ and a nice tree decomposition $\mathcal T=(\mathcal X, \mathcal E)$, we define for each node $t\in\mathcal X$ the dynamic programming table state 
\begin{eqnarray*}
  \mathcal T_t&=&(S_1,\ldots, S_k, f_1,\ldots,f_k, E)  
\end{eqnarray*}
where $(S_1,\ldots, S_k)$ is a partition of the vertices of the bag $X_t$ into $k$ parts, $f_i$ is the number of vertices each of those parts has already forgotten below $t$ and $E$ the set of edges that inside the bag that encodes the connectivity. Each state stores a value 
\begin{eqnarray*}
    T_t[S_1,\ldots, S_k, f_1,\ldots, f_k, E]\in\{0,1\}
\end{eqnarray*}
indicating whether the partial solution is potentially valid.

Moreover, during the algorithm we process the nodes of the nice tree decomposition bottom-up.
The detailed description of this Algorithm is defered to Appendix~\ref{appendix:bounded-treewidth}.



\section{Conclusion}\label{sec:conclusion}
As we showcase in this work, dense partitioning problems on graphs are of structural and algorithmic interest even when restricted to graph classes that are close to trees, such as well partitioned chordal graphs and thick forests. We provided a structural understanding of optimal dense partitions on well partitioned chordal graphs showing that we can only focus on parts of diameter $2$ and specifically generalized stars and cliques. Building on this, we developed a polynomial time algorithm for MDGP on thick forests. We complemented this, by showing that $\text{MDGP}$ is $\mathsf{NP}$-hard even when restricted to split graphs. In addition we developed a dynamic programming framework for constructing connected partitions into $k$ parts of specific sizes on graphs of bounded treewidth.

Our polynomial time algorithm for thick forests relies on the structural results for well partitioned chordal graphs and hence on the freedom to choose the number of parts of our partition. Hence, it cannot be extended to $\text{MDGP}_k$. Thus, whether $\text{MDGP}_k$ is computable in polynomial time on thick forests  remains open. 

The question if MDGP is in FPT parameterized by treewidth also remains open. Generally, one could investigate the efficiency of the bounded treewidth algorithm for Connected Partition with specific part sizes; specifically whether the dependency on the size of the partition and the treewidth of the graph can be improved. The W[1]-hardness of Equitable Connected Partition parameterized by pathwidth~\cite{equi} indicates that the tracking-the-partition approach we took in this paper  cannot be successful to place MDGP in FPT; any variation of our dynamic program that stores connectivity and cluster-size information could also solve GL-partition and thus Equitable Connected Partition.
\bibliography{refs}

\appendix
\section{Extended Preliminaries}\label{appendix:prelim}
All the graphs we refer to in this paper are simple, unweighted and finite. We mainly follow standard graph theoretic notation unless stated otherwise. In particular we denote an undirected graph $G$ by $G=(V,E)$, where $V$ is its vertex set and $E$ its edge set. Moreover we use $|V|$ and $|E|$ to denote the number of vertices and edges of $G$ respectively, when the graph we refer to is clear from the context. If we need to specify the graph we refer to we use $V(G)$ and $E(G)$ instead of $V$ and $E$ respectively.
Given a graph $G=(V,E)$ we also denote the \emph{open neighborhood} of $V\subseteq V(G)$ by $N_G(V):=\cup_{v\in V}\{u\notin V:vu\in E\}$ and the \emph{closed neighborhood} by $N_G[V]:=V\cup N_G(V)$. Again we may omit the subscript when the graph is clear from the context. We say that a vertex $v\in V(G)$ \emph{dominates} $u\in V(G)$ if $u\in N[v]$.
Moreover, for $v\in V(G)$ we call $|N(v)|$ the \emph{degree} of $v$ and denote it by $\text{deg}_G(v)$. Furthermore $\max_{v\in V}\{\text{deg}(v)\}$ is called the \emph{maximum degree of $G$} and is denoted by $\Delta (G)$.

We say that a graph $G$ is \emph{connected} if there is a path connecting any pair of its vertices. Given a graph $G=(V, E)$ and a vertex set $S\subseteq V$, we say that $S$ is a \emph{separator} of $G$ if $G[V\setminus S]$ is not connected. We say that a connected graph $G$ is $k$-connected if for every $S\subseteq V$ with $|S|\leq k-1$ the graph $G\setminus S$ remains connected. Given a graph $G$ we denote the number of connected components of $G$ by $cc(G)$. Moreover, a graph $G$ has \emph{diameter} $d$, denoted by $\text{diam}(G)$ if the maximum shortest path between any pair of vertices in $G$ has length $d$.

A set of vertex sets $\mathcal P=\{P_1,\ldots, P_k\}$ such that $P_i \subseteq V$ for all $i\in [k]$ is a \emph{partition of size $k$} of $G$ if and only if $V(G)=\cup_{i=1}^k P_i$ and $P_i\cap P_j=\emptyset$ for all $i, j\in [k]$ such that $i\neq j$. We also denote by $G[S]$ for some subset $S\subseteq V(G)$ the \emph{graph induced by $S$}, that is, $G[S]:=(S, \{vu:v, u\in S\text{ and }vu\in E(G)\})$. Given a partition $\mathcal{P}$ of size $k$ of $G$, we say that $\mathcal P$ forms a \emph{connected partition of $G$} if $G[P_i]$ is connected for all $i\in [k]$.

The \emph{density} of a graph $G=(V,E)$ is denoted by $d(G)$ and is the ratio between the number of edges of a graph $G$ and the number of its vertices, meaning $d(G)=|E|/|V|$. Moreover, for any $S\subseteq V(G)$ with $S\neq \emptyset$, $d(S)=d(G[S])$. Given a partition $\mathcal P$ of a graph $G$ we define the \emph{density of the partition $\mathcal P$} as $d(\mathcal P)=\sum_{P\in\mathcal P}d(P)$.

We say that a graph is \emph{chordal} if all its induced cycles are triangles. A subclass of chordal graphs are well partitioned chordal graphs. In particular, a connected graph $G$ is a \emph{well partitioned chordal graph} if there exist a partition $\mathcal{P}$ of $V(G)$ and a tree $\mathcal{T}'$ having $\mathcal{P}$ as a vertex set such that the following hold: 
 \begin{enumerate}
 \item each part $X\in\mathcal P$ is a clique in $G$, 
 \item for each edge $XY\in E(\mathcal T')$, there are subsets $X'\subseteq X$ and $Y'\subseteq Y$ such that $E(G[X,Y])=X'\times Y'$, and 
 \item for each pair of distinct $X, Y\in V(\mathcal T')$ with $XY\not\in E(\mathcal T')$, $E(G[X,Y])=\emptyset$.
 \end{enumerate}
Moreover, a graph $G$ is a \emph{thick tree} if it is chordal and every two minimal separators of $G$ are either equal or disjoint. It is easy to notice that thick trees are a subclass of well partitioned chordal graphs, and also that block graphs are a subclass of thick trees. We say that a graph $G$ is a \emph{block graph} if every biconnected component of it is a clique. We also say that a graph $G$ is \emph{bipartite} if its vertices can be partitioned into two sets each inducing an independent set. Lastly, we also work on split graphs, where $G$ is a \emph{split graph} if there is a partition of its vertices into two parts $K,I$, where $I$ induces an independent set and $K$ induces a clique. For both bipartite and split graphs we say that a graph is \emph{complete} if it has the maximum amount of edges while still being bipartite or split respectively.

Given a graph $G=(V,E)$, a tree decomposition of $G$ is a pair $\mathcal T=(T,\mathcal X)$, where $T=(I,F)$ is a tree and $\mathcal X=(X_i)_{i\in I}\subseteq \mathcal P(V)$ is a family of subsets $X_i$ of $V$ called \emph{bags}, such that the following properties hold.
    \begin{enumerate}
        \item $\cup_{i\in I}X_i=V$
        \item For all $uv\in E$, there exists an $i\in I$ with $\{u,v\}\subseteq X_i$
        \item For all $v\in V$, $X^{-1}:=\{i\in I|v\in X_i\}$ is connected in $T$.
    \end{enumerate}
    Moreover let $G$ be a graph and $\mathcal T=(T, X=(X_i)_{i\in I})$ a tree decomposition of $G$. The \emph{width of $\mathcal T$} is defined as $\max\{|X_i||i\in I\}-1$ and the \emph{treewidth of $G$} as $\text{tw}(G):=\min\{width(\mathcal T)\mid \mathcal T \text{ is a tree decomposition of } G\}$.
    Lastly, we define a \emph{nice tree decomposition} as a rooted binary tree with four different types of nodes:
		\begin{enumerate}
			\item Leaf nodes have no children and bag size $1$.
			\item Introduce nodes have one child. The child has the same vertices as the parent with one deleted.
			\item Forget nodes have one child. The child has the same vertices as the parent with one added.
			\item Join nodes have two children, both identical to the parent.
		\end{enumerate}
    We note that there is an algorithm that converts a tree decomposition into a nice one with $\mathcal{O}(nk)$ bags in time $\mathcal{O}(nk^2)$ (see, for example~\cite{CyganFKLMPPS15}).

\section{Omitted Proofs of Section~\ref{section: well-partitioned chordal}}
\begin{figure}
	\begin{center}
	\begin{tikzpicture}
		\draw (0,0) circle (1cm);

		\draw [dashed](-0.5,0.87) to [bend left] (-0.5,-0.87);
		\node[fill=none, draw=none, text=black] at (-0.6,0) {$\ell_1$};
		
		\draw (-4,0) circle (1cm);
		\node[fill=none, draw=none, text=black] at (-4,0) {$n_1$};
		\draw  (-0.5,0.87) to [bend right] (-4,1);
		\draw  (-4,0.-1) to [bend right] (-0.5,-0.87);
		
		\draw (1,3) circle (0.5cm);
		\draw (4,1) circle (1.2cm);
		\draw (2,-3) circle (0.8cm);
		\draw  (0.5,3) to [bend right] (-1,0);
		\draw  (1.5,3) to [bend left] (1,0);
		\draw  (4,2.2) to [bend right] (0,1);
		\draw  (4,-0.2) to [bend left] (0,-1);
		\draw  (2.8,-3) to [bend right] (1,0);
		\draw  (1.2,-3) to [bend left] (-1,0);
		
		
	\end{tikzpicture}
	\end{center}
\label{fig: proof of incomplete star}
\caption{Illustration of the notation used in Lemma~\ref{lemma:star-decomposition}}
\end{figure}

\wellpartitioneddiameter*
\begin{proof}
    Let $G$ be a well partitioned chordal graph and $\mathcal P$ a partition
    of $V(G)$. If $C$ has diameter 1 then $C$ already induces a clique and the 
    statement of the lemma holds.
    Let $C$ be a part of $\mathcal P$ of diameter $2$. Since $G$ is a well 
    partitioned chordal graph and $G[C]$ is an induced subgraph of $G$, $G[C]$ 
    is also well partitioned and 
    there exists a partition $\mathcal{P}^{wp}$ and 
    a tree $\mathcal{T}^{wp}$ certifying that $G[C]$ is well partitioned. 
    Out of 
    all such partitions let $\mathcal{P}^{wp}$ be one maximal number of parts. 
    Then, since $C$ has diameter 2, from the third condition of the definition 
    of well partitioned chordal graphs $\mathcal{T}^{wp}$ has diameter at most 2 and thus there exist sets 
    $P^{wp}_{0}$ and $P^{wp}_{i}$, $i\in [r]$, in $\mathcal{P}^{wp}$, 
    that induce a star in $\mathcal{T}^{wp}$ with center $P^{wp}_{0}$ and such 
    that the set $P^{wp}_{0}\cup\bigcup_{i\in [r]} P^{wp}_{i}$ contains all 
    vertices of $C$. Finally, recall that for every pair 
    $P^{wp}_{0}$ and $P^{wp}_{i}$, $i\in [r]$, there exist sets 
    $X_{i}\subseteq P^{wp}_{0}$ and $Y_{i}\subseteq P^{wp}_{i}$ such that the 
    edges between $P^{wp}_{0}$ and $P^{wp}_{i}$ are the edges of the complete 
    bipartite graph with $X_{i}$ and $Y_{i}$ as parts. 
    We first assume that for some $i\in [r]$, $Y_{i}\subsetneq P_{i}$. 
    Then we can modify the partition $\mathcal{P}^{wp}$ and the tree 
    $\mathcal{T}^{wp}$ so that they still certify that $G[C]$ is 
    well partitioned chordal but with more parts, contradicting to the maximality of the parts of the partition.
    In order to do so, we remove the leaf $P_{i}$ from the tree and replace it
    with a path $P_{Y_{i}}$ and $\overline{P}_{Y_{i}}$ such that $P_{Y_{i}}$ contain exactly the vertices of $Y_{i}$ which have neighbors in $P_{0}$ and
    is a neighbor of $P_{0}$ and $\overline{P}_{Y_{i}}$ contains the vertices 
    of $P_{i}\setminus Y_{i}$ which have neighbors only in $P_{Y_{i}}$ and
    is now a neighbor of $P_{Y_{i}}$. 
    Therefore, $Y_{i}=P_{i}$ for every $i\in [r]$.

    Moreover, again by the maximality of the parts of $\mathcal{P}^{wp}$ 
    observe that for every vertex $v$ of $P_{0}$ there exists some $X_{i}$ such that $v\in X_{i}$. Indeed, if otherwise, let 
    $X=P_{0}\setminus \bigcup_{i\in [r]} X_{i}$ and notice that we may obtain a partition with one more part by removing $X$ from $P_{0}$ and adding
    it as a new leaf of the tree that is a neighbor of $P_{0}$.

    Observe that from all of the above arguments we also obtain that if $r=1$ then
    the tree is an edge since both of its vertices are leaves the graph $G[C]$
    is a clique. In this case the lemma holds trivially.

    
    From now on we assume that $r\geq 2$.
    If for every $i\in [r]$, $|X_{i}|\geq \lfloor\frac{|P_{0}|}{2}\rfloor+1$ then $G[C]$ is already a generalized star with center $P_{0}$ and leaves $P_{1},\dots,P_{r}$, concluding the lemma. Therefore, it is enough to consider the case where for some $i\in [r]$, 
    $|X_{i}|\leq \frac{|P_{0}|}{2}$. In other words, $C$ induces an incomplete star. 
    Then the desired partition can be immediately obtained from Lemma~\ref{lemma:star-decomposition}. 
    \end{proof} 

\wellpartitioneddiametertwo*
\begin{proof} First of all notice that, without loss of generality we may assume that $G[C]$ is
    connected.
     Let $G$ be a well partitioned chordal graph and $\mathcal P$ a partition of $V(G)$. Let $C$ 
     be a part of $\mathcal P$ of diameter at least $3$.
     Since $G$ is a well partitioned chordal graph and $C$ is an induced subgraph of $G$, $G[C]$ 
     is also a well partitioned chordal graph and thus there exists a partition $\mathcal{P}^{wp}$ 
     and a tree $\mathcal{T}^{wp}$ certifying that $G[C]$ is well partitioned. Without loss
     of generality we may assume that out of all such partitions $\mathcal{P}$ has the
     greatest number of parts. Given two neighboring nodes $t_{1}$ and $t_{2}$ with parts 
     $P_{t_{1}}$ and $P_{t_{2}}$ denote by $X_{t_{1},t_{2}}\subseteq P_{t_{1}}$ and $X_{t_{2},t_{1}}\subseteq P_{t_{2}}$ 
     the sets for which $E(G[C](P_{t_{1}},P_{t_{2}}))=X_{t_{1},t_{2}}\times X_{t_{2},X_{1}}$.  
     Let also $\mathcal{T}_{1}$ and $\mathcal{T}_{2}$ be the subtrees of $\mathcal{T}^{wp}$ that 
     occur after the removal of the edge $t_{1}t_{2}$ where without loss of generality 
     $\mathcal{T}_{i}$ is the tree that contains $t_{i}$.
     We set $C_{i}=\cup_{t\in V(T_{i})} P_{t}$, $i\in [2]$.

    Moreover, we use throughout the following criterion. For a partition of $C$ in two parts $A$
    and $B$, with $m_A=|E(G[A])|$, $m_B=|E(G[B])|$ and $c=|E(G[C](A,B))|$, the inequality
    $d(G[A])+d(G[B])\geq d(G[C])$ is equivalent to
    \begin{equation}\label{eq:crit}
    m_A|B|^2+m_B|A|^2\geq c\,|A||B|,
    \end{equation}

    \medskip
	\begin{figure}[ht]
	\begin{center}
		\begin{tikzpicture}
			\draw [red](0,0) circle (1cm);
			\draw [dashed](0.2,0.87) to [bend right] (0.2,-0.87);
			\node[fill=none, draw=none, text=black] at (0.5,0) {$X_{t_1,t_2}$};
			\node[fill=none, draw=none, text=black] at (-1.4,0.6) {$P_{t_1}$};
			
			\draw [blue](3,-2) circle (1cm);
			\draw [dashed](3.3,-1.15) to [bend left] (2.3,-2.5);
			\node[fill=none, draw=none, text=black] at (2.55,-1.8) {$X_{t_2,t_1}$};
			\node[fill=none, draw=none, text=black] at (4.4,-1.4) {$P_{t_2}$};
			
			\draw [thick](0.92,-0.4) to (2.2,-1.4);
			\node[fill=none, draw=none, text=black] at (2,-0.6) {$t_1t_2$};
			
			\draw [red](0,2.5) circle (0.7cm);
			\draw [red,thick](0,1) to (0,1.8);
			
			\draw [red](0,-2) circle (0.5cm);
			\draw [red,thick](0,-1) to (0,-1.5);
			
			\draw [red](-3,-2) circle (1cm);
			\draw [red,thick](-0.92,-0.4) to (-2.2,-1.4);
			
			\draw [red](-3,-4) circle (0.5cm);
			\draw [red,thick](-3,-3.5) to (-3,-3);
			
			\draw [blue](0,-4) circle (0.8cm);
			\draw [blue,thick](0.7,-3.6) to (2.145,-2.5);
			
			\draw [blue](3,-4) circle (0.6cm);
			\draw [blue,thick](3,-3.4) to (3,-3);
			
			\draw [blue](6,-4) circle (0.8cm);
			\draw [blue,thick](5.3,-3.6) to (3.855,-2.5);
			
			\draw [blue](1.5,-7) circle (1cm);
			\draw [blue,thick](1.8,-6.05) to (2.8,-4.556);
			
			\draw [blue](4.5,-7) circle (0.5cm);
			\draw [blue,thick](4.5,-6.5) to (3.2,-4.556);
			
		\draw[decorate,
		decoration={brace,mirror,amplitude=5pt},blue]
		(7,-8) -- (7,-1)
		node[midway,right=8pt,align=center]
		{$\mathcal T_{t_2}$};
			
			\draw[decorate,
			decoration={brace,amplitude=5pt},red]
			(-4.2,-4.5) -- (-4.2,3.2)
			node[midway,left=8pt] {$\mathcal T_{t_1}$};
		\end{tikzpicture}
	\end{center}
\caption{Example of a tree certifying a partition of a well partitioned chordal graph. $P_{t_1}$, $P_{t_2}$ are two neighboring, through the edge $t_1t_2$, nodes with parts $P_{t_1}$, $P_{t_2}$. $X_{t_{1}t_{2}}\subseteq P_{t_{1}}$, $X_{t_{2}t_{1}}\subseteq P_{t_{2}}$ are the sets for which $E(G[C](P_{t_{1}},P_{t_{2}}))=X_{t_{1}t_{2}}\times X_{t_{2}t_{1}}$.}
\end{figure}

    \noindent \textbf{Case 1: There is an edge $t_{1}t_{2}$ such that 
    $X_{t_{1},t_{2}}\subsetneq P_{t_{1}}$ and $X_{t_{2},t_{1}}\subsetneq P_{t_{2}}$}.
    In other words, there exists a vertex $v_{1}\in P_{t_{1}}$ that does not have any neighbor in 
    $P_{t_{2}}$ and a vertex $v_{2}\in P_{t_{2}}$ that does not have any neighbor in $P_{t_{1}}$. 
    We argue that $d(G[C_{1}])+d(G[C_{2}])\geq d(G[C])$. Let $m_{i}$ and $n_{i}$ denote the number 
    of edges and the number of vertices of $G[C_{i}]$ respectively, $i\in [2]$. 
    Let also $|X_{1}|=\ell_{1}$ and $|X_{2}|=\ell_{2}$ and note that $|P_{t_{i}}|\geq \ell_{i}+1$, 
    $i\in [2]$.
    Then for the inequality of the densities to hold it should be the case that 
        $$\frac{m_{1}}{n_{1}}+\frac{m_{2}}{n_{2}}\geq \frac{m_{1}+m_{2}+\ell_{1}\ell_{2}}{n_{1}+n_{2}}.$$
    This, however is equivalent to 
    $$m_{1}n_{2}^{2}+m_{2}n_{1}^{2}\geq \ell_{2}n_{1}\ell_{1}n_{2}.$$
    Observe that $m_{i}\geq \binom{\ell_{i}+1}{2}$. Thus, for the above inequality to hold it is 
    enough to show that 
    $$\binom{\ell_{1}+1}{2}n_{2}^{2} + \binom{\ell_{2}+1}{2}n_{1}^{2}\geq \ell_{2}n_{1}\ell_{1}n_{2}.$$
    By expanding the binomials we obtain that the above inequality is equivalent to 
    $$(\ell_{1}n_{2}-\ell_{2}n_{1})^{2}+\ell_{1}n_{2}^{2}+\ell_{2}n_{1}^{2}\geq 0,$$
    which holds since it's a sum of non-negative values.

    \medskip
    
    Let us now assume that there is no such edge in $\mathcal{T}^{wp}$. 
    
    \medskip
    
    \noindent \textbf{Case 2: For every edge $t_{1}t_{2}$ either $X_{t_{1},t_{2}}=P_{t_{1}}$ or $X_{t_{2},t_{1}}= P_{t_{2}}$.} 
    \begin{claim}\label{clm:fullleaf}
    If $t_{1}$ is a leaf of 
    $\mathcal{T}^{wp}$ and $t_{2}$ is its unique neighbor then $X_{t_{1},t_{2}}=P_{t_{1}}$.
    \end{claim}

    \begin{proofofclaim}
    Towards a contradiction we assume that $X_{t_{1},t_{2}}\subsetneq P_{t_{1}}$
    Let us then remove the leaf $P_{1}$ from the tree and replace it
    with a path $t_{0},t_{1}'$ such that $t_{0}$ is a neighbour of $t_{1}'$, $t_{1}'$ is 
    a neighbour of $t_{2}$, $P_{1}^{'}=X_{t_{1},t_{2}}$ and $P_{0}=P_{1}\setminus X_{t_{1},t_{2}}$.
    This partition and the corresponding tree certify that $G[C]$ is a well-partitioned
    chordal graph but the partition has one part more than $\mathcal{P}^{wp}$. This is a 
    contradiction to the choice of $\mathcal{P}^{wp}$ as one with the maximum number of parts.
    \end{proofofclaim}

    \begin{claim}\label{clm:leafdichotomy}
    For every leaf $t_{1}$ of $\mathcal{T}^{wp}$ and its unique neighbor $t_{2}$ it holds that
    either $|X_{t_{2},t_{1}}|\geq\frac{|P_{t_{2}}|}{2}$ or there exists a partition of $C$ into
    two sets $C_{1}$ and $C_{2}$ such that the sum of the densities of $G[C_{1}]$ and $G[C_{2}]$
    is at least the density of $G[C]$.
    \end{claim}

    \begin{proofofclaim}
    From Claim~\ref{clm:fullleaf}, it holds that $X_{t_{1},t_{2}}=P_{t_{1}}$. 
    Therefore, every vertex of $P_{t_{1}}$ is adjacent to exactly the vertices of 
    $X_{t_{2},t_{1}}$ outside $P_{1}$. Let $n_{1}=|P_{t_{1}}|$, $n_{2}=|P_{t_{2}}|$,
    $n_{3}=|C|-n_{1}-n_{2}$, $\alpha=|X_{t_{2},t_{1}}|$, and let $m_{2}$ denote the number of edges 
    of $G[C\setminus P_{t_{1}}]$. Since $G[C]$ is connected, $\alpha\geq 1$. Thus, 
    the number of edges between $P_{t_{1}}$ and $C\setminus P_{t_{1}}$ equals $n_{1}\alpha$. Moreover $G[C\setminus P_{t_{1}}]$ is connected, 
    and hence we get that $m_{2}\geq\binom{n_{2}}{2}+n_{3}$.

    Let us assume that $\alpha<\frac{n_{2}}{2}$.

    \medskip
    \noindent\emph{Case $n_{1}\geq 2$.} We argue that $C_{1}=P_{t_{1}}$ and $C_{2}=C\setminus P_{t_{2}}$
    is the desired partition. By \eqref{eq:crit}, with $m_{1}=\binom{n_{1}}{2}$ and
    $c=n_{1}\alpha$, for the inequality of the densities to hold it should be the case that
    $$m_{1}(n_{2}+n_{3})^2+m_{2}n_{1}^2\geq\alpha n_{1}^2(n_{2}+n_{3}).$$
    Using $m_{2}\geq\binom{n_{2}}{2}+n_{3}$, $m_{1}\binom{n_{1}}{2}$ and $\alpha\leq\frac{n_{2}}{2}$, 
    it is enough to show,
    after dividing by $\frac{n_{1}}{2}$, that
    $$(n_{1}-1)(n_{2}+n_{3})^2+n_{1}n_{2}(n_{2}-1)+2n_{1}n_{3}\geq n_{1}n_{2}(n_{2}+n_{3}),$$
    that is, $(n_{1}-1)(n_{2}+n_{3})^2+2n_{1}n_{3}\geq n_{1}n_{2}(n_{3}+1)$. Since $n_{1}\geq2$ we have
    $n_{1}-1\geq\frac{n_{1}}{2}$, so it is enough to show that
    $\frac{(n_{2}+n_{3})^2}{2}+2n_{3}\geq n_{2}(n_{3}+1)$, which is equivalent to
    $$(n_{2}-1)^2+n_{3}^2+4n_{3}\geq 1$$
    and holds since $\alpha\geq 1$ implies $n_{2}\geq2$.

    \medskip
    \noindent\emph{Case $n_{1}=1$.} Write $P_{t_{1}}=\{v\}$, $n=|C|$ and $m=|E(G[C])|$, and
    consider first the partition $C_{1}=\{v\}$, $C_{2}=C\setminus\{v\}$. Here $m_{1}=0$ and
    $c=\alpha$, so \eqref{eq:crit} is equivalent to $\frac{m-\alpha}{n-1}\geq\frac{m}{n}$, 
    that is,
    $\alpha\leq\frac{m}{n}$. If this holds we are done, so assume
    \begin{equation}\label{eq:vfail}
    \alpha>\frac{m}{n}.
    \end{equation}
    Consider now the partition $C_{1}=P_{t_{2}}$, $C_{2}=C\setminus P_{t_{2}}$, and let $c_0$ be 
    the number of edges between $P_{t_{2}}$, $m_{2}=\binom{n_{2}}{2}$ be the number of edges of $G[C_{1}]$, and $C\setminus P_{t_{2}}$ and $m_B$ the number of 
    edges of $G[C\setminus P_{t_{2}}]$. 
    We note here that while it is counter-intuitive to consider as $C_{2}$ a part that is 
    not connected, the total density of its connected components is at least as much as
    the density of $C_{2}$.
    If \eqref{eq:crit} holds for this partition we are done, so let us assume that it fails. 
    Therefore,
    $$\frac{m_{2}}{n_{2}}+\frac{m_{B}}{n-n_{2}}<\frac{m_{2}+m_{B}+c_{0}}{n_{2}+(n-n_{2})}.$$
    By extending and simplifying the inequality, we obtain that
    $$\binom{n_{2}}{2}(n-n_{2})^{2}+m_{B}n_{2}^2<c_0 n_{2}(n-n_{2}),$$
    and thus, since $m_{B}\geq 0$, $$\binom{n_{2}}{2}(n-n_{2})^{2}<c_{0}n_{2}(n-n_{2}).$$
    Dividing by $n_{2}(n-n_{2})$ we obtain that
    $c_0>\frac{(n_{2}-1)(n-n_{2})}{2}$, and therefore
    $$m \geq \binom{n_{2}}{2}+c_0 > \frac{n_{2}(n_{2}-1)}{2}+\frac{(n_{2}-1)(n-n_{2})}{2}
    = \frac{n(n_{2}-1)}{2}.$$
    Together with \eqref{eq:vfail} this gives $\alpha>\frac{n_{2}-1}{2}$, hence
    $2\alpha\geq n_{2}$ since $\alpha$ is an integer, that is,
    $\alpha\geq\frac{n_{2}}{2}$, a contradiction to the assumption
    $\alpha<\frac{n_{2}}{2}$. Therefore one of the two partitions above satisfies
    \eqref{eq:crit}, and the second option of the claim holds.
\end{proofofclaim}

    From now on we may assume that \textbf{for every leaf $t_{1}$ of $\mathcal{T}^{wp}$ and its unique neighbor $t_{2}$ it holds that $|X_{t_{2},t_{1}}|\geq\frac{|P_{t_{2}}|}{2}$}, as
    otherwise from Lemma~\ref{clm:leafdichotomy} we may break $C$ into more parts whose total
    density is at least the density of $G[C]$.
    Let us now assume that for every two leaves 
    $t$ and $t'$ the path that joins them in $\mathcal{T}^{wp}$ has length at most $2$. Then, 
    from Lemma~\ref{lemma:star-decomposition}, $G[C]$ can be partitioned into cliques and 
    generalized stars whose total density is at least the density of $G[C]$. Since both 
    generalized stars and cliques have diameter at most 2 the lemma follows.

    Thus, from now on, we assume that \textbf{if $t$ and $t'$ are two anti-diametrical leaves 
    then the path that joins them in $\mathcal{T}^{wp}$ has length at least $3$}.

    Let $P$ be a path that joins two anti-diametrical vertices $t=t_{1}$ and $t'$ of 
    $\mathcal{T}^{wp}$.
    Let also $P'=t_{1}t_{2}\dots t_{k}$ be a maximal subpath of $P$ such that 
    $X_{t_{i+1},t_{i}}\geq \frac{|P_{t_{i+1}}|}{2}$ for every $1\leq i \leq k$. 
    Observe that from the assumption we made, 
    the vertices $t_{1}$ and $t_{2}$ form such a path and thus $k\geq 2$.
    Let us first consider the case where $t_{k}\neq t'$. Therefore, $t_{k}$ has a neighbor 
    $t_{k+1}$ in $P$.
    We consider the interaction of the edge $t_{k}t_{k+1}$. Let also $\mathcal{T}_{k}$ and 
    $\mathcal{T}_{k+1}$ the subtrees of $\mathcal{T}$ occurring after the removal of $t_{k}t_{k+1}$
    that contain $t_{k}$ and $t_{k+1}$ respectively. 
    We set $C_{i}=\cup_{t\in V(T_{i})} P_{t}$, $i\in \{k,k+1\}$.
    Since $X_{t_{k+1},t_{k}}\subsetneq P_{t_{k+1}}$ then we claim that 
    $d(G[C_{k}])+d(G[C_{k+1}])\geq d(G[C])$. 
    Observe that since $X_{t_{k+1},t_{k}}\subsetneq P_{t_{k+1}}$ we may assume that 
    $X_{t_{k},t_{k+1}}=P_{t_{k}}$. Denote $|X_{t_{k},t_{k+1}}|=\ell_{k}=|P_{t_{k}}|$ and 
    $|X_{t_{k+1},t_{k}}|=\ell_{k+1}$.
    Let, as before $m_{k}$ and $n_{k}$ denote the amount of edges and vertices of $C_{k}$ and
    $m_{k+1}$ and $n_{k}$ denote the amount of edges and vertices of $C_{k+1}$ respectively. 
    Then the inequality holds if and only if
    $$\frac{m_{k}}{n_{k}}+\frac{m_{k+1}}{n_{k+1}}\geq \frac{m_{k}+m_{k+1}+\ell_{k}\ell_{k+1}}{n_{k}+n_{k+1}}$$
    which as before is equivalent to 
    $$m_{k}n_{k+1}^{2}+m_{k+1}n_{k}^{2}\geq n_{k}\ell_{k+1}n_{k+1}\ell_{k}.$$
    Observe that as before
    $C_{k}$ contains all edges between the $\ell_{k}$ vertices of $P_{t_{k}}$ and at least 
    $\frac{\ell_{k}}{2}$ edges between the vertices $P_{t_{k-1}}$ and $P_{t_{k}}$ and thus
    $m_{k}\geq \frac{\ell_{k}^{2}}{2}$ (recall that $k\geq 2$).
    Moreover $C_{k+1}$ contain all edges among all of vertices of $P_{t_{k+1}}$. 
    However, $P_{t_{k+1}}$ contains at least $\ell_{k+1}+1$ vertices and thus $m_{k+1}\geq\binom{\ell_{k+1}+1}{2}$.
    By replacing $m_{k}$ and $m_{k+1}$ and expanding the binomials we obtain that the above inequality is equivalent to 
    $$(\ell_{k}n_{k+1}-\ell_{k+1}n_{k})^{2}+\ell_{k+1}n_{k}^{2}\geq 0,$$
    which holds since it's a sum of non-negative values.
    If $t_{k}$ is a leaf then $t_{k}=t'$ and thus $k\geq 4$ (by the assumption that $P$ has length at least 3.).
    We now consider the edge $t_{k-2}t_{k-1}$.
    Let also $\mathcal{T}_{k-2}$ and 
    $\mathcal{T}_{k-1}$ the subtrees of $\mathcal{T}$ occurring after the removal of $t_{k-2}t_{k-1}$
    that contain $t_{k-2}$ and $t_{k-1}$ respectively. 
    We set $C_{i}=\cup_{t\in V(T_{i})} P_{t}$, $i\in \{k-2,k-1\}$.
    We claim that $d(G[C_{k-2}])+d(G[C_{k-1}])\geq d(G[C])$. 
    Similarly, to the previous cases, we denote by $m_{k-2}$ and $n_{k-2}$
    the number of edges and vertices of $C_{k-2}$ and by 
    $m_{k-1}$ and $n_{k-1}$ the number of edges and vertices of $C_{k-1}$.
    Denote also by 
    $\ell_{k-2}=|X_{t_{k-2},t_{k-1}}|$ and $\ell_{k-1}=X_{t_{k-1},t_{k-2}}$.
    Observe that since $t_{k}$ is a leaf each vertex of $P_{t_{k}}$ has at least
    $\frac{|P_{t_{k-1}}|}{2}\geq \frac{\ell_{k-1}}{2}$ neighbours in $P_{t_{k}}$. 
    Similarly, by the definition of the path 
    $|X_{t_{k-2},t_{k-3}}|\geq \frac{|P_{t_{k-2}}|}{2}$ and thus there exist
    at least $\frac{|P_{t_{k-2}}|}{2}\geq \frac{\ell_{k-2}}{2}$ edges between
    $P_{t_{k-2},t_{k-3}}$ in $C_{k-2}$.
    Therefore, $m_{k-1}\geq \frac{\ell_{k-1}^{2}}{2}$ and $m_{k-2}\geq \frac{\ell_{k-2}}{2}$.
    Plugging all values in~\eqref{eq:crit} we obtain that the desired inequality holds if and only
    if $$\ell_{k-1}^{2}n_{k-2}^{2}+\ell_{k-2}^{2}n_{k-1}^{2}\geq 2\ell_{k-2}\ell_{k-1}n_{k-2}n_{k-1}.$$
    This is equivalent to
    $$(\ell_{k-1}n_{k-2}-\ell_{k-2}n_{k-1})^{2}\geq 0,$$
    which always holds.
    This completes the proof of the lemma.
\end{proof}

\section{Omitted Parts of Section~\ref{section: algorithm thick trees}}\label{appendix:thick trees}
\paragraph*{Structure of an optimal partition}{Block-cut tree}
We first show through Algorithm~\ref{alg:block-cut-tree} that given a thick tree $G$ we can construct its block-cut tree in polynomial time.
\begin{algorithm}[h]
	\caption{Construction of the block-cut tree $\mathcal T$ of a thick tree $G$}
	\label{alg:block-cut-tree}
	\KwIn{A thick tree $G=(V,E)$}
	\KwOut{The block-cut tree $\mathcal T=(\mathcal X,\mathcal E)$}
	Initialize $\mathcal T \leftarrow (\emptyset,\emptyset)$

	\While{$G$ is not a clique}{
	Find an arbitrary maximal clique $K$ of $G$

	Create a block node $B_K$ corresponding to $K$ (if it does not already exist) and add it to the node set of $\mathcal T$

	Compute the set of connected components of $G-K$, $\mathcal D$

		$S\gets \emptyset$
        
	\ForEach{$D \in \mathcal D$}{
		$S_D \leftarrow N_G(D) \cap K$
		
		$S\gets S\cup S_D$

		\If{there is no cut node $C_{S_D}$ in $\mathcal T$}{
			create cut node $C_{S_D}$ and add it to the node set of $\mathcal T$
		}
		add edge $B_K C_{S_D}$ to the edge set of $\mathcal T$

		Let $K_D$ be the maximal clique containing $S_D$ and intersecting $D$

		create block node $B_{K_D}$ (if it does not already exist) and add it to the node set of $\mathcal T$

		add edge $C_{S_D} B_{K_D}$ to the edge set of $\mathcal T$

	}
	$G\gets G-K\cup S$
}
\Return $\mathcal T$
\end{algorithm}

\subparagraph*{Running time.}
In each iteration of Algorithm~\ref{alg:block-cut-tree} we process one block node of the block-cut tree. For each maximal clique $K$ it processes, the algorithm computes the connected components of $G-K$ and, for each component $D$, computes the separator $S_D=N_G(D)\cap K$. These computations can be done in polynomial time. Moreover, each maximal clique, each minimal clique separator, and each edge of the block-cut tree is created at most once. Since a chordal graph has at most $n$ maximal cliques, we have at most $n$ iterations. Therefore, the total running time of the construction is polynomial in the size of $G$.

\subparagraph*{Correctness.}
Let $K$ be the maximal clique processed in some iteration of Algorithm~\ref{alg:block-cut-tree}. For every connected component $D$ of $G-K$, the set $S_D=N_G(D)\cap K$ is a clique separator, and in a thick tree it is in fact a minimal clique separator. Here it is essential that $G$ is a thick tree, because if $G$ was for example just chordal then $S_D$ would be a separator of $G$ but not necessarily minimal, since separators in chordal graphs can have a non empty intersection without being identical. Hence, the algorithm correctly creates a cut node for each separator adjacent to $K$. Furthermore, by the definition of thick trees, every component $D$ of $G-K$ is attached to $K$ through exactly such a separator, and there is a unique maximal clique $K_D$ on the side of $D$ containing $S_D$. Thus the algorithm correctly creates the corresponding child block node and connects it to the cut node of $S_D$. Recursing on the induced subgraph $G[D\cup S_D]$, so that the shared separator vertices are retained, constructs exactly the decomposition of $G$ into maximal cliques and minimal clique separators. Since every maximal clique of $G$ is eventually reached once, and every edge of the block-cut tree corresponds to the containment relation between a maximal clique and a separator clique, the resulting graph is exactly the block-cut tree of $G$.

\paragraph*{Structure of an optimal partition}

\centeronseparator*
\begin{proof}
	Let $P_1,\ldots, P_r$ with $r\geq 2$ be the leaves of $P$. Each $P_i$ is completely joined to $P_0$ and disjoint from it, so $P_0\cup P_i$ is a clique and is therefore contained in a maximal clique, that is, a block $B_i\supseteq P_0$, where $B_i$ is the block node for that contains the vertices of $P_i$ for all $i\in [r]$. 
	Since the leaves of a generalized star are pairwise non-adjacent, hence the blocks $B_1,\ldots, B_r$ are pairwise distinct. 
	Due to the definition of the block-cut tree, the blocks that contain the clique $P_0$ form a connected subtree $\mathcal T_{P_0}$ of $\mathcal T$ that contains every $B_i$. Since $r\geq 2$, the subtree $\mathcal T_{P_0}$ has at least two blocks. 
	The separator associated with any edge between two adjacent blocks $B, B'$ of $\mathcal T_{P_0}$ is $B\cap B'\supseteq P_0$, hence contains the non-empty set $P_0$. Therefore no two of these separators are disjoint,
	and since $G$ is a thick tree they are pairwise equal, to a single minimal clique
	separator $C\supseteq P_0$. Consequently, every block of $\mathcal T_{P_0}$ therefore has a neighbour in $\mathcal T_{P_0}$ through the separator $C$, and thus contains $C$. In particular each $B_i\supseteq C$, where $C$ is the node that contains $P_0$. Consequently $P_0\subseteq C$ and the leaves lie in the distinct blocks $B_1,\ldots, B_r$, all incident to the cut node of $C$.
\end{proof}

\atmostone*
\begin{proof}
	Let $G$ be a thick tree and $\mathcal T$ be its block-cut tree. Let also $B$ be a block node of $\mathcal T$ and assume, in order to reach a contradiction, that there are $t\geq 2$ vertices $c_1,\ldots c_t\in \text{Child}_{\mathcal T}(B)$ for some block node $B$ that are the centers of $t$ generalized stars having leaves in both $V_{\text{Child}_{\mathcal T}(B)}$ and $B$, in an optimal dense partition $D^*$. Assume also that among all the partitions with equal density, $D^*$ is the one that minimizes the number of such stars.
	Denote those $t$ generalized stars and by $S_1,\ldots, S_t$, by $B_1,\ldots, B_t$ the parts of $B$ that also belong in $S_i$ and not in $c_i$ respectively for $i\in [t]$ and by $C_i$ the vertices of the center of each star (that are also vertices of $B$ by the definition of the block cut tree). We focus on the subgraph $G[S_1\cup S_2]$.
	
	Denote by $S_1'$ and $S_2'$ the graphs $S_1-B_1$ and $S_2-B_2$ respectively and also by $T_1=d(S_1)-d(S_1')$ and $T_2=d(S_2)-d(S_2')$ the contribution of the leaves in $B$ to the density of each star. This equals to:
	\begin{eqnarray*}
		T_i&=&d(S_i)-d(S_i')=\frac{m_i}{n_i}-\frac{m_i'}{n_i'}=\frac{m_i'+b_ic_i+\binom{b_i}{2}}{n_i'+b_i}-\frac{m_i'}{n_i'}\\
		&=&\frac{n_i'm_i'+n_i'b_ic_i+n_i'\binom{b_i}{2}-n_i'm_i'-b_i'm_i'}{n_i'(n_i'+b_i)}=\frac{n_i'b_ic_i+n_i'\binom{b_i}{2}-b_im_i'}{n_i'(n_i'+b_i)}\\
		&=&\frac{b_ic_i+\binom{b_i}{2}-b_id(S_i')}{n_i'+b_i}=\frac{b_i(c_i-d(S_i'))+\binom{b_i}{2}}{n_i'+b_i}\\
	\end{eqnarray*}
	Notice that $T_{i}<\frac{b_i}{2}$ because the center is inside $S_{i}'$, meaning $m_{i}'\geq\binom{c_{i}}{2}$, hence $d(S_{i}')\geq\frac{c_{i}(c_{i}-1)}{2n_{i}'}$. This, combined with the fact that 
	$c_i\leq n_i'-1$ gives us the desired result.
	We now consider the following cases based on the size of the center of the stars.
	\begin{enumerate}
		\item There is a $C_i$, say $C_1$, such that $|V(C_1)|\leq d(S_1')$. Then,
		\begin{eqnarray*}
			d(S_1)\leq d(S_1')+\frac{|V(B_1)|-1}{2}&\Leftrightarrow&\frac{b_1(c_1-d(S_1'))+\binom{b_1}{2}}{n_1'+b_1}\leq\frac{b_1-1}{2}\Leftarrow\\
			& \Leftarrow& \frac{\binom{b_1}{2}}{n_1'+b_1}\leq\frac{b_1-1}{2}\Leftrightarrow\\
			& \Leftrightarrow&b_1(b_1-1)\leq (n_1'+b_1)(b_1-1)\Leftrightarrow\\
			& \Leftrightarrow&0\leq n_1'(b_1-1)
		\end{eqnarray*}
		Hence in this case we create the partition $S_1', B_1, S_2,\ldots, S_t$ which has at least equal density as $D^*$ but less stars with leaves both upwards and downwards on the block-cut tree which contradicts our original assumption.
		
		\item $|V(C_i)|\geq d(S_i')$ for both $i\in[2]$. 
		
		In this case we consider the contribution of each leaf in $B$ to the corresponding star, $\frac{T_i}{b_i}$, and we move the leaf that contributes less to its star to the other one. Specifically, if $\frac{T_1}{b_1}\leq \frac{T_2}{b_2}$ we consider the partition $S_1', S_2\cup B_1$ and prove that
		\begin{eqnarray*}
			d(S_1)+d(S_2)\leq d(S_1')+d(S_2\cup B_1)\Leftrightarrow\\
			T_1\leq d(S_2\cup B_1)-d(S_2)\Leftrightarrow\\
			\frac{b_1(c_1-d(S_1'))+\binom{b_1}{2}}{n_1'+b_1}\leq \frac{b_1(c_2-d(S_2))+\binom{b_1}{2}+b_1b_2}{n_2+b_1}\Leftrightarrow\\
			\frac{c_1-d(S_1')+\frac{b_1-1}{2}}{n_1'+b_1}\leq \frac{c_2-d(S_2)+\frac{b_1-1}{2}+b_2}{n_2+b_1}\overset{\frac{T_1}{b_1}\leq \frac{T_2}{b_2}}{\Leftarrow}\\
			\frac{c_2-d(S_2')+\frac{b_2-1}{2}}{n_2'+b_2}\leq \frac{c_2-d(S_2)+\frac{b_1-1}{2}+b_2}{n_2+b_1}\Leftrightarrow\\
			(n_2'+b_2+b_1)\left(c_2-d(S_2')+\frac{b_2-1}{2}\right)\leq (n_2'+b_2)\left(c_2-d(S_2)+\frac{b_1-1}{2}+b_2\right)\Leftrightarrow\\
            (n_2'+b_2)\left(T_2+\frac{b_2-b_1}{2}-b_2\right)\leq -b_1\left(c_2-d(S_2')+\frac{b_2-1}{2}\right)\Leftrightarrow\\
            (n_2'+b_2)\left(T_2+\frac{b_2-b_1}{2}-b_2\right)\leq-b_1\left(\frac{(n_2'+b_2)T_2}{b_2}\right)\Leftrightarrow\\
            T_2(n_2'+b_2)\left(1+\frac{b_1}{b_2}\right)\leq -(n_2'+b_2)\left(-\frac{b_1+b_2}{2}\right)\Leftrightarrow\\
            T_2\leq\frac{\frac{b_1+b_2}{2}}{\frac{b_1+b_2}{b_2}}\Leftrightarrow\\
            T_2\leq\frac{b_2}{2}
		\end{eqnarray*}
        which holds as proven above.
		Hence we have again found a partition, namely $S_1', S_2\cup B_1,S_3,\ldots, S_t$ with density at least equal to $D^*$ but with one generalized star less having leaves both upwards and downwards the block cut tree which contradicts the minimality assumption.
	\end{enumerate}   
\end{proof}

\paragraph*{Cut nodes} Let $C\in\mathcal C$ with children block nodes $\text{Child}_{\mathcal T}(C)=\{B_1,\ldots, B_d\}$. For an integer $t\in\{0,\ldots, |C|\}$, the number of vertices of $C$ reserved for the parent block of $C$, we define two families of values.

$f_C(\text{fin}, t)$ is the maximum dense partition value of a partition of $G[A_C]$ that uses every vertex of $C$ except the $t$ reserved ones, in which every part is finalized, that is, contained in $A_C$. This includes the option of finalizing a generalized star centered at $C$ that has no leaf in the parent block.
Parameter $t$ here plays the role of both the up-clq and up-leaf states of the block-graph algorithm of Hanaka et al.~\cite{DBLP:journals/jco/HanakaIO25}. It basically captures how many vertices of $C$ are handed to the parent block, while the decision of whether those vertices are absorbed into a clique, attached as a leaf of a generalized star, or passed further up is deferred to the block recurrence, where the vertices in question lie in a single maximal clique and merging them is always optimal.

$f_C(\text{star}, t, l, L, Q)$ is the maximum dense partition value of such a partition when, in addition, exactly one generalized star centered at $C$ is left \emph{open}. This means that its center consists of $l\geq 1$ non reserved vertices of $C$, and the leaves it has already collected from the child blocks amount to $L$ vertices and $Q$ internal edges. The open star is not counted in the value. It is completed at the parent block by attaching one further leaf clique of size $z\geq 0$, where it contributes $\Phi\left(l, L+z, Q+\binom{z}{2}\right)$.

In Figure~\ref{fig: generalized cut center} it is illustrated why it suffices to only consider possible star centers in the cut nodes of the block-cut tree.
\begin{figure}
	\begin{center}
		\begin{tikzpicture}[every node/.style={fill=black,inner sep=1.5pt},scale=1]
			
			\node[shape=circle,draw=black, label=above:$a$] (a) at (0,0) {};
			\node[shape=circle,draw=black, label=left:$b$] (b) at (-1,-1) {};
			\node[shape=circle,draw=black, label=right:$c$] (c) at (1,-1) {};
			\node[shape=circle,draw=black, label=left:$d$] (d) at (-2,-2) {};
			\node[shape=circle,draw=black, label=right:$e$] (e) at (2,-2) {};
			\node[shape=circle, white] (ph) at (2,-4) {};
			
			\path [-] (a) edge (b);
			\path [-] (a) edge (c);
			\path [-] (b) edge (c);
			\path [-] (b) edge (d);
			\path [-] (c) edge (d);
			\path [-] (b) edge (e);
			\path [-] (c) edge (e);
		\end{tikzpicture}
		\begin{tikzpicture}[every node/.style={fill=black,inner sep=1.5pt},scale=0.8]
			\node[shape=circle, white] (ph) at (2,-4) {};
			\node[shape=circle, white] (ph1) at (3,-4) {};
		\end{tikzpicture}
		\begin{tikzpicture}[every node/.style={fill=black,inner sep=1.5pt},scale=0.7]
			
			\node[shape=circle,draw=black, label=above:$a$] (a) at (0,0) {};
			\node[shape=circle,draw=black, label=left:$b$] (b) at (-1,-1) {};
			\node[shape=circle,draw=black, label=right:$c$] (c) at (1,-1) {};
			
			\path [-] (a) edge (b);
			\path [-] (a) edge (c);
			\path [-] (b) edge (c);
			
			\draw (0,-0.5) circle[radius=1.7];
			
			\node[shape=circle,draw=black, label=left:$b$] (b') at (-1,-3.5) {};
			\node[shape=circle,draw=black, label=right:$c$] (c') at (1,-3.5) {};
			
			\path [-] (b') edge (c');
			
			\begin{scope}[xshift=0cm, yshift=-3.5cm, rotate=90]
				\node[draw, fill=none, minimum width=1.2cm, minimum height=3cm, transform shape] {};
			\end{scope}
			\path [thick] (0,-2.9) edge (0,-2.2);
			\path [thick] (0,-4.1) edge (-1.25,-5);
			\path [thick] (0,-4.1) edge (1.25,-5);
			
			\node[shape=circle,draw=black, label=left:$b$] (b'') at (-3,-6) {};
			\node[shape=circle,draw=black, label=right:$c$] (c'') at (-1,-6) {};
			\node[shape=circle,draw=black, label=left:$d$] (d) at (-2,-7) {};
			
			\path [-] (b'') edge (d);
			\path [-] (b'') edge (c'');
			\path [-] (d) edge (c'');
			
			\draw (-2,-6.3) circle[radius=1.5];
			
			\node[shape=circle,draw=black, label=left:$b$] (b''') at (1,-6) {};
			\node[shape=circle,draw=black, label=right:$c$] (c''') at (3,-6) {};
			\node[shape=circle,draw=black, label=right:$e$] (e) at (2,-7) {};
			
			\path [-] (b''') edge (e);
			\path [-] (b''') edge (c''');
			\path [-] (e) edge (c''');
			
			\draw (2,-6.3) circle[radius=1.5];
		\end{tikzpicture}
	\end{center}
	\caption{The optimal partition here is considering the whole graph as one part. This can be viewed as a generalized star with center $\{b, c\}$ and leaves $\{a,d,e\}$ or equivalently as a generalized star with center $\{a,b,c\}$ and leaves $\{d,e\}$ and hence it suffices to only consider star centers in cut nodes.}
	\label{fig: generalized cut center}
\end{figure}
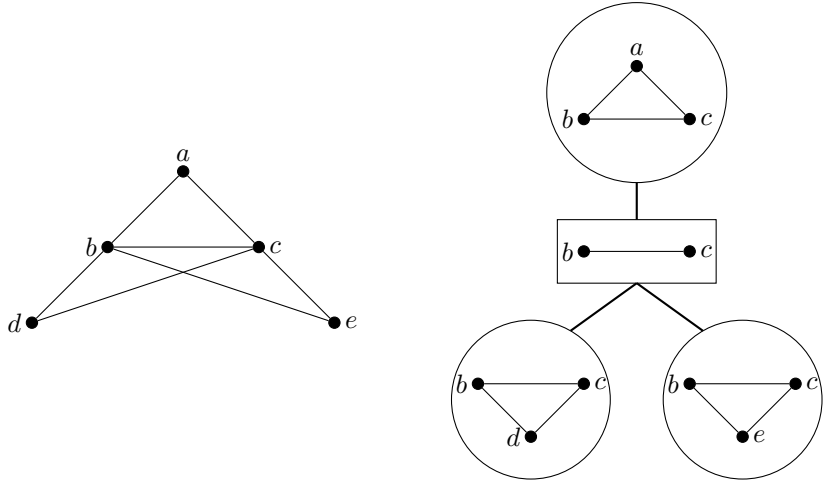

\paragraph*{Block nodes} Let $B\in\mathcal B$ with parent cut node $C_P(B)$. For $u\in\{0,\ldots, |C_P(B)|\}$, the number of vertices of $C_P(B)$ that $B$ may use as private, and for an integer $y\geq 0$, the size of a leaf clique contained in $B$ participating to a generalized star centered at $C_P(B)$. We define $g_B(u, y)$ to be the maximum dense partition value obtained on $G[A_B]$ together with the $u$ potentially private vertices, where the $y$ vertices are set aside and not counted, since they are counted through $\Phi$ once the star at $C_P(B)$ is completed, and every other part is finalized. For the root block $B_r$ we only use the value $g_{B_r}(0,0)$.

In contrast to the block-graph algorithm of Hanaka et al.~\cite{DBLP:journals/jco/HanakaIO25}, whose block states carry the center of a star (the states sc and scu), here a block never hosts a star center as proven in Lemma~\ref{lemma: center on separator} every generalized star is centered on a separator, so a block only ever gives a single leaf upward and receives vertices from its parent cut.

\paragraph*{Recurrences}
We compute the values bottom-up in $\mathcal T$ and each value $f_C(\cdot)$ and $g_B(\cdot)$ is computed once.

\subparagraph{Cut node recurrence} Let $C$ be a cut node that has children block nodes $B_1,\ldots, B_d$ and set $\bar t=|C|-t$. Remember that $t\in\{0,\ldots, |C|\}$ is the number of vertices of $C$ reserved for the parent block of $C$. 
Each child $B_j$ decides how many vertices $s_j\geq 0$ of $C$ to use into its own subtree and, when a star is centered at $C$, one leaf clique of size $y_j\geq 0$. Its contribution is $g_{B_j}(s_j, y_j)$. Since two minimal separators are either equal or disjoint and every vertex of $C$ plays the same role with respect to the incident blocks, it suffices to track cardinalities and not subsets.

When no open star is centered at $C$, every child contributes $y_j=0$ and the leftover vertices of $C$ form one clique:
\begin{eqnarray*}
	h_C(t)&=&\max_{\substack{s_1,\ldots, s_d\geq 0\\ s_1+\cdots+s_d\leq \bar t}}\left(\sum_{j=1}^d g_{B_j}(s_j, 0)\right)+\text{Left-clq}\left(\bar t-\sum_{j=1}^d s_j\right).
\end{eqnarray*}
When a star of center size $l$ is centered at $C$, we sum up over the children: over all choices $(s_j, y_j)_{j\in[d]}$ with $\sum_{j=1}^d s_j\leq \bar t-l$ we set $S=\sum_{j=1}^d s_j$, $L=\sum_{j=1}^d y_j$, $Q=\sum_{j=1}^d\binom{y_j}{2}$, and we let the remaining $r=\bar t-l-S$ vertices of $C$ form one clique part.
\begin{eqnarray*}
	f_C(\text{star}, t, l, L, Q)&=&\max\left(f_C(\text{star}, t, l, L, Q), \sum_{j=1}^d g_{B_j}(s_j,y_j) +\text{Left-clq}(r)\right).
\end{eqnarray*}
We do not need to also treat the leftover vertices as an extra leaf of the star: since they lie in $C$ and are therefore adjacent to every leaf belonging in a child block, placing them in the center instead is at least as good and is already covered by a larger value of $l$.
Finally, a star centered at $C$ that receives no leaf from the parent is finalized here, so
\begin{eqnarray*}
	f_C(\text{fin}, t)&=&\max\left(h_C(t), \max_{l, L, Q}\left[f_C(\text{star}, t, l, L, Q)+\Phi(l, L, Q)\right]\right).
\end{eqnarray*}

\subparagraph{Block node recurrence} Let $B$ be a block node and let it have children cut nodes $C_1,\ldots, C_e$ and write $\rho_{B(u)}=|R(B)|+u$. By Lemma~\ref{lemma: at most one}, $B$ completes at most one open star imported from a child cut, so we maximize over the choice of the child $C_b\in\{C_1,\ldots, C_e\}\cup\{\varnothing\}$ whose star $B$ completes. For a fixed such choice let
\begin{eqnarray*}
	F_b(S)&=&\max_{\substack{(s_i)_{i:\,C_i\neq C_b}\\ \sum_i s_i=S}}\sum_{i:\,C_i\neq C_b} f_{C_i}(\text{fin}, s_i)
\end{eqnarray*}
be the best combined finalized value of the other children when they reserve $S$ vertices for $B$ in total. Then $g_B(u, y)$ is the maximum over both of the following. If $C_b=\emptyset$, no imported star is completed and every remaining vertex in $\rho_{B(u)}$ that is not participating in the imported star joins one clique:
\begin{eqnarray*}
	g_B(u,y)=\max_S\left\{F_\emptyset(S)+\text{Left-clq}\left(\rho_{B(u)}+S-y\right)\right\},\qquad |R(B)|+S\geq y.
\end{eqnarray*}
Otherwise, $B$ completes the star of $C_b$: we choose the number $s_b$ of $C_b$-vertices drawn into $B$, an open triple $(l, L, Q)$ of $C_b$, a closing leaf of size $z\geq 1$ taken from $B$'s pool, and the contribution $y$, and with $\rho=\rho_{B(u)}+S+s_b$,
\begin{eqnarray*}
g_B(u, y)=\max\{F_b(S)+f_{C_b}(\text{star}, s_b, l, L, Q)+\Phi\left(l, L+z, Q+\binom{z}{2}\right)+\\\text{Left-clq}\left(\rho-y-z\right)\},
\end{eqnarray*}
ranging over all $S, s_b, (l, L, Q), z\geq 1$ and $y$ with $y+z\leq \rho$. All the vertices of $\rho$, that is, the private vertices of $B$, the vertices from the parent and the vertices drawn up from the child cuts, that are not used to close a star and not part of a star lie in the clique of $B$. Hence, merging them into a single clique part is optimal, which is what $\text{Left-clq}$ accounts for.

\subparagraph{Root} At the root block $B_r$ no vertices may be reserved upward and nothing is donated, so
\begin{eqnarray*}
	\text{OPT}(G)&=&g_{B_r}(0, 0).
\end{eqnarray*}

\paragraph*{Correctness}
By Lemma~\ref{lemma: center on separator} and Lemma~\ref{lemma: at most one} we know that there is an optimal partition of $G$ in which every part induces either a clique contained in a single block or a generalized star centered at a cut node $C$ with at most one completely joined leaf in each block incident to $C$, and in which every block completes at most one such star with a leaf of its own. We show that the recurrences compute the values $f_C$ and $g_B$ as defined above and then the optimal value follows.

We first note that at a cut node only cardinalities matter.

\begin{observation}\label{obs: interchange}
	Let $C$ be a cut node of the block-cut tree $\mathcal T$ of some thick tree $G$. Every vertex of $C$ is adjacent to every vertex of every block incident to $C$, because such a block is a clique containing $C$. Hence, in the partitions considered above, any two vertices of $C$ are interchangeable: for a fixed assignment of the vertices of $C$ to roles (reserved for the parent, center of the star, participating into a given child subtree, or leftover), every choice of which vertices of $C$ play each role yields a partition of the same density.
\end{observation}
We now proceed in proving that the algorithm indeed constructs a partition that realizes the maximum value of MDGP.
\begin{lemma}\label{lemma: dp correct}
	Let $C\in\mathcal C$ be a cut node with child blocks $B_1,\dots,B_d$, let
	$t\in\{0,\dots,|C|\}$, and let $l\geq 1$ and $L,Q\geq 0$.
	\begin{enumerate}
		\item $f_C(\text{fin},t)$ equals the maximum value of a partition
		of $A_C$ that leaves exactly $t$ vertices of $C$ in no part (all remaining
		parts being complete) and that each part induces either a clique or a generalized star.
		\item $f_C(\text{star},t,l,L,Q)$ equals the maximum value of a partition of $A_C$ that leave exactly $t$ vertices of $C$ in no part and that
		contain a designated part which is a generalized star whose center is formed by
		$l$ of the used vertices of $C$ and whose leaves lie in the child subtrees and
		total $L$ vertices and $Q$ internal edges, of the sum of the densities of all
		parts except the designated star.
	\end{enumerate}
	Moreover, let $B\in\mathcal B$ be a block node with parent cut $C_P(B)$, and let
	$u\in\{0,\dots,|C_P(B)|\}$ and $y\geq 0$. Then
	\begin{enumerate}\setcounter{enumi}{2}
		\item $g_B(u,y)$ equals the maximum, value of a partition of the vertex set
		$A_B$ after adding $u$ distinguished vertices of $C_P(B)$ that contain a
		designated part which is a clique on exactly $y$ vertices, of the sum of the
		densities of all parts except the designated clique.
	\end{enumerate}
\end{lemma}
\begin{proof}
	We proceed by induction on $\mathcal T$, moving from its leaves towards the root $B_r$. Since $\mathcal T$ alternates between block and cut nodes and its leaves are blocks, the base case is a leaf block node and the two inductive steps are one for cut nodes and one for block nodes.

	\emph{Base case.} Let $B$ be a block node with no child cut nodes, so $A_B=R(B)$. Every private vertex of $B$ and every vertex participating in the clique above, lies in the clique of $B$, so any two of them are adjacent. This means that no generalized star with two or more leaves can be formed among them, and every part is a subclique of $B$. Setting aside the $y$ vertices reserved as a leaf, the remaining $|R(B)|+u-y$ vertices are best kept as a single clique, of density $\text{Left-clq}(|R(B)|+u-y)$. The block recurrence, whose only possible branch is $C_b=\emptyset$ with $S=0$, returns exactly this value.

	\emph{Cut node.} Let $C$ have child blocks $B_1,\ldots, B_d$, and assume as induction hypothesis that part (iii) of Lemma~\ref{lemma: dp correct} holds for each $B_j$. This means that, $g_{B_j}(u,y)$ equals its optimum for all $u,y$, which is available since each $B_j$ lies strictly below $C$ in $\mathcal T$.
	Removing the vertices of $C$ disconnects $A_C$ into the child subtrees $A_{B_1},\ldots, A_{B_d}$, with no edge between distinct child subtrees. 
	Consider first a partition realizing $f_C(\text{fin}, t)$ in which no star is centered at $C$. 
	Every part that intersects two child subtrees $A_{B_j}$ and $A_{B_k}$ would, being a clique or a generalized star, have its center in the separator on the $\mathcal T$-path between $B_j$ and $B_k$, namely $C$ itself, by Lemma~\ref{lemma: center on separator} which would lead to a contradiction. 
	Hence, every part intersects at most one child subtree, and a part intersecting $A_{B_j}$ together with some vertices of $C$ is exactly a part of a partition of $A_{B_j}$ that is reserving those $s_j$ vertices of $C$ to be used inside $B$, contributing $g_{B_j}(s_j, 0)$, while the remaining $r=\bar t-\sum_j s_j$ vertices of $C$ form a clique of density $\text{Left-clq}(r)$. Conversely every choice of the $s_j$ yields such a partition. Maximizing gives $\sum_j g_{B_j}(s_j, 0)+\text{Left-clq}(r)=h_C(t)$, and by Observation~\ref{obs: interchange} knowing the cardinalities $s_j$ suffices.

	Now consider a partition corresponding to $f_C(\text{star}, t, l, L, Q)$, with one open star centered at $C$. By Lemma~\ref{lemma: center on separator} its center consists of $l$ non reserved vertices of $C$, and each of its leaves is a clique lying in one incident child block. Specifically, a leaf in $A_{B_j}$ is exactly a clique that $B_j$ donates, of some size $y_j$, and it is universal to the center because the whole clique of $B_j$ is adjacent to $C\supseteq$ center. Its contribution, together with the rest of the partition of $A_{B_j}$, is $g_{B_j}(s_j, y_j)$. The leaves lying in the children of $C$, therefore contribute $L=\sum_j y_j$ vertices and $Q=\sum_j\binom{y_j}{2}$ internal edges, and the remaining $r=\bar t-l-\sum_j s_j$ vertices of $C$ form a clique. The value of the finalized parts is $\sum_j g_{B_j}(s_j, y_j)+\text{Left-clq}(r)$, which is exactly what the recurrence maximizes into $f_C(\text{star}, t, l, L, Q)$ finalizing the star here, meaning, adding no leaf from the parent, contributes $\Phi(l, L, Q)$ and yields the second term of $f_C(\text{fin}, t)$. By the induction hypothesis every $g_{B_j}$ is optimal, hence so are $f_C(\text{fin}, t)$ and $f_C(\text{star}, t, l, L, Q)$.

	\emph{Block node.} Let $B\neq B_r$ have child cut nodes $C_1,\ldots, C_e$, and assume parts (i)–(ii) of Lemma~\ref{lemma: dp correct} hold for each child cut $C_i$, and hence for each $f_{C_i}$. 
	All interactions among the child subtrees of $B$, and between $B$ and its parent, pass through vertices of $B$. 
	In a partition corresponding to $g_B(u, y)$, each child subtree $A_{C_i}$ either contributes a finalized partition, of value $f_{C_i}(\text{fin}, s_i)$ where $s_i$ is the number of $C_i$ vertices it reserves for $B$, or, for at most one child $C_b$ by Lemma~\ref{lemma: at most one}, exports one open star centered at $C_b$ that $B$ completes. 
	In the latter case $B$ attaches to that star a leaf of size $z\geq 1$ taken from the vertices corresponding to $\rho$, contributing $f_{C_b}(\text{star}, s_b, l, L, Q)+\Phi(l, L+z, Q+\binom{z}{2})$. This leaf is completely joined to the center because it lies in the clique of $B\supseteq C_b\supseteq$ center. The vertices of $\rho$ consist of the private vertices of $B$, the $u$ vertices of $C_P(B)$ used inside $B$ and the $\sum_i s_i$ vertices drawn up from the child cuts. Out of these vertices, $y$ are reserved to be used upward and $z$ close the imported star, and the rest form one clique of density $\text{Left-clq}(\rho-y-z)$. Summing these contributions gives exactly the two branches of the block recurrence, and conversely every such choice is realizable. By the induction hypothesis every $f_{C_i}$ is optimal, hence so is $g_B(u, y)$ which completes our induction.
\end{proof}

Since at the root $B_r$ no vertex may be reserved for a parent, the value $g_{B_r}(0, 0)$ ranges over all partitions of $G=G[A_{B_r}]$ of the assumed structural form, and by Lemma~\ref{lemma: dp correct} it equals the optimal MDGP value of $G$. An optimal partition can be obtained by storing the maximizing choice with each computed value.

\paragraph*{Running time}
The block-cut tree has $\mathcal O(n)$ nodes, and each value is computed once and reused, so it suffices to bound the number of values and the cost per value. 
For a cut node $C$ the values $f_C(\text{fin}, \cdot)$ and $f_C(\text{star}, \cdot)$ are indexed by $t\leq|C|$, a center size $l\leq|C|$ and leaf totals $(L, Q)$ with $L\leq n$. 
Since $L\leq n$ and $Q\leq\binom{n}{2}$ the table $f_C(\text{star},t,l, L, Q)$ has size $\mathcal O(|C|^2n^3)$ .
Summing up over the $d$ child blocks of $C$ and the leaf total, costing $\mathcal O(d|C|n)$ per center size, hence $\mathcal O(d|C|^2n)$ in total for $C$.
For a block node $B$ the value $g_B(u, \cdot)$ is computed for each $u\leq|C_P(B)|$. Specifically, for fixed $u$ the recurrence branches over the completed child, our algorithm runs over the sum of the remaining children, and in the completing branch ranges over the open triple of the completed child, the closing leaf size and the leaf that $B$ provides to a star centered at its parent cut $C_P(B)$, which is polynomial in $n$. Summing over all nodes and all values of $u$, and using $\sum_{C}|C|,\ \sum_B(|C_P(B)|+\deg B)=\mathcal O(n)$, the total running time is polynomial in $n$.

\section{Omitted Proofs of Section~\ref{section: NP-hardness}}\label{section: NP-hardnessAP}

\completesplitdensity*
\begin{proof}
    We prove that given two complete split graphs $G=(K\cup I, E)$ and $G'=(K\cup (I\cup\{v\}), E\cup\{\{v,u\}\mid u\in V(K)\})$ it holds that $d(G)< d(G')$. Indeed, observe that $d(G)=\frac{\binom{|K|}{2}+|I||K|}{|I|+|K|}$ and $d(G')=\frac{\binom{|K|}{2}+(|I|+1)|K|}{|I|+|K|+1}$. Then,
    \begin{eqnarray*}
        d(G)&<& d(G')\Leftrightarrow\\
        \frac{\binom{|K|}{2}+|I||K|}{|I|+|K|}&<&\frac{\binom{|K|}{2}+(|I|+1)|K|}{|I|+|K|+1}\Leftrightarrow\\
        \left(\binom{|K|}{2}+|I||K|\right)\left(|I|+|K|+1\right)&<&\left(\binom{|K|}{2}+(|I|+1)|K|\right)\left(|K|+|I|\right)\Leftrightarrow\\
        \binom{|K|}{2}+|I||K|&<&|K|(|K|+|I|)\Leftrightarrow\\
        \binom{|K|}{2}&<&|K|^2,
    \end{eqnarray*}
    which holds for every $|K|\geq 1$.
\end{proof}

\splitdense*
\begin{proof}
Let $G=(K\cup I, E)$ be a complete split graph with $|K|=n$ and $|I|=\alpha n$ for some $\alpha >0$. 
\begin{eqnarray*}
    d(G)&=&\frac{|E|}{|K|+|I|}=
\frac{1}{2}\left(\frac{n(n-1)+2\alpha n^2}{n(1+\alpha)}\right)\\
&=&\frac{1}{2}\left(\frac{n(n+2\alpha n-1)}{n(1+\alpha)}\right)= \frac{1}{2}\left(\frac{(n+2\alpha n-1)}{1+\alpha}\right)\\
&=&\frac{1}{2}\left(n+\frac{(\alpha n-1)}{1+\alpha}\right)=\frac{1}{2}\left(n\left(1+\frac{\alpha}{1+\alpha}\right)-\frac{1}{1+\alpha}\right)
\end{eqnarray*}
concluding the proof of the lemma.\end{proof}

\splitdensetwo*
\begin{proof}
 By simply applying the definition of the density we see that  
 $$d_l(n,\alpha)-d_l(n_1,\alpha_1)-d_l(n_2,\alpha_2)
=\frac 12 \left( n\left(1+\frac{\alpha}{1+\alpha}\right)- n_1\left(1+\frac{\alpha_1}{1+\alpha_1}\right) -n_2\left(1+\frac{\alpha_2}{1+\alpha_2}\right)\right).$$
Using the fact that $n_1+n_2=n$, this equals to
$\frac 12 \left(\frac{n\alpha}{1+\alpha}-\frac{n_1\alpha_1}{1+\alpha_1} -\frac{n_2\alpha_2}{1+\alpha_2}\right)
$.
 It is now immediate that since $ \alpha_1n_1+\alpha_2n_2=\alpha n$  this expression is equal to $0$ when $\alpha=\alpha_1=\alpha_2$.
 
 The equality $ \alpha_1n_1+\alpha_2n_2=\alpha n$ implies that in case where $\alpha \neq \alpha_1$, one of the $\alpha_i$'s is larger while the other one is smaller than $\alpha$, so assume without loss of generality that $\alpha_1<\alpha<\alpha_2$. 
 Then, by setting $\gamma=\gamma(\alpha,\alpha_1,\alpha_2)=(2(1+\alpha)(1+\alpha_1)(1+\alpha_2))^{-1}$ we get that $\frac 12 \left(\frac{n\alpha}{1+\alpha}-\frac{n_1\alpha_1}{1+\alpha_1} -\frac{n_2\alpha_2}{1+\alpha_2}\right)$ equals to:
 $$\displaystyle\gamma(\alpha,\alpha_1,\alpha_2)\big(n\alpha(1+\alpha_1)(1+\alpha_2) - n_1\alpha_1(1+\alpha)(1+\alpha_2)-n_2\alpha_2(1+\alpha)(1+\alpha_1)\big).$$
 Leaving out the multiplicative $\gamma(\alpha,\alpha_1,\alpha_2)=(2(1+\alpha)(1+\alpha_1)(1+\alpha_2))^{-1}$, we can transform the remaining expression by multiplying some parts and using the fact that $n\alpha=n_1\alpha_1+n_2\alpha_2$ to get
\begin{align*}
 & &n\alpha(1+\alpha_1+\alpha_2+\alpha_1\alpha_2) - n_1\alpha_1(1+\alpha+\alpha_2+\alpha\alpha_2)-n_2\alpha_2(1+\alpha_1+\alpha+\alpha_1\alpha)\\
& = &n\alpha(\alpha_1+\alpha_2+\alpha_1\alpha_2) - n_1\alpha_1(\alpha+\alpha_2+\alpha\alpha_2)-n_2\alpha_2(\alpha_1+\alpha+\alpha_1\alpha).\end{align*}
Plugging in $n=n_1+n_2$ and reordering this equals to:
\begin{align*} & = &(n_1+n_2)\alpha(\alpha_1+\alpha_2+\alpha_1\alpha_2) - n_1\alpha_1(\alpha+\alpha_2+\alpha\alpha_2)-n_2\alpha_2(\alpha_1+\alpha+\alpha_1\alpha)\\
& = &n_1\alpha(\alpha_2+\alpha_1\alpha_2) +n_2\alpha(\alpha_1+\alpha_1\alpha_2) - n_1\alpha_1(\alpha_2+\alpha\alpha_2)-n_2\alpha_2(\alpha_1+\alpha_1\alpha)\\
& = &n_1\alpha\alpha_2  +(n_1+n_2)\alpha\alpha_1\alpha_2 +n_2\alpha\alpha_1 - n_1\alpha_1(\alpha_2+\alpha\alpha_2)-n_2\alpha_2(\alpha_1+\alpha_1\alpha)\\
& = &n_1\alpha\alpha_2  +n_2\alpha\alpha_1 - n_1\alpha_1\alpha_2-n_2\alpha_2\alpha_1\\
& = &\alpha_2n_1(\alpha-\alpha_1)  +\alpha_1n_2(\alpha-\alpha_2).\end{align*}
Observe that the equalities $n=n_1+n_2$ and $\alpha n=\alpha_1 n_1+\alpha_2n_2$ imply $\alpha n_1+\alpha n_2 = \alpha_1 n_1+\alpha_2n_2$ which gives $n_1(\alpha-\alpha_1)=n_2(\alpha_2-\alpha)$ and thus we can reorder to finally get:
\[d_l(n,\alpha)-d_l(n_1,\alpha_1)-d_l(n_2,\alpha_2)=\gamma(\alpha_2-\alpha_1)(\alpha-\alpha_1)n_1=\gamma(\alpha_2-\alpha_1)(\alpha_2-\alpha)n_2.\]
Lastly, observe that the inequalities $\alpha_1<\alpha<\alpha_2$  imply that $c_1=(\alpha_2-\alpha_1)(\alpha-\alpha_1)>0$ and $c_2=(\alpha_2-\alpha_1)(\alpha_2-\alpha)>0$.
\end{proof}
\paragraph*{Proof of Lemma~\ref{lemma::kparts}}
\kparts*
\begin{proof}
Assume that there is a combination of $G,n,\alpha$ and $k$ with a partition $\mathcal P$ that falsifies the statement. Observe that we can assume that $G$ is a complete split graph, since the density of $\mathcal P$ can only increase if we add edges to $G$. By Lemma~\ref{lemma: complete split density} and ~\ref{lem::splitdense2}, we can then also assume, that $\mathcal P$ contains exactly $k+1$ parts, each non-trivial (adding the otherwise isolated vertices from the independent set part, or merging if there are more than $k+1$ parts can only increase the overall density).

Thus we assume the counterexample is for $G=(K\cup I,E)$ such that $G$ is a complete split graph with $|K|=n$ and $|I|=\alpha n$ such that $k<n$. Define the best possible value of a partition into $k+1$ non-trivial parts for $G$, formally:
\[D^*_{k+1}=max\{d({\cal P})\mid {\cal P} =\{P_1,\dots,P_{k+1}\}, \bigcup_{i=1}^{k+1}P_i=K\cup I, P_i\cap K\neq \emptyset, 1\leq i\leq k+1\}\]

Let $k$ be the smallest value such that $D^*_{k+1}\geq \rho_k(n,\alpha)$. By Observation~\ref{obs::onepart} we know that $k>1$ and by $G$ being a counterexample, we have  $k+1<n$. 

Let ${\cal P}^*=\{P_1,\dots,P_{k+1}\}$ be a partition with $d({\cal P}^*)=D^*_{k+1}$, and define $|P_i\cap K|=n_i$ and $|P_i\cap I|=\alpha_in_i$ for all $1\leq i\leq k+1$. Sort the $P_i$ such that $\alpha_1\leq \alpha_2\leq\dots\leq\alpha_{k+1}$.

\begin{claim}\label{claimsize} $\alpha_2>\alpha$
\end{claim}
Assume towards contradiction that $\alpha_2\leq\alpha$. Consider the partition ${\cal P}'$ obtained from ${\cal P}$ by replacing $P_1,P_2$ by $P_1\cup P_2$.  Thus ${\cal P}'$ is a partition into $k$ pats with density $d({\cal P'})=d({\cal P})-d(P_1)-d(P_2)+d(P_1\cup P_2)$. By the assumptions on the counterexample, we have $d({\cal P})\geq \rho_k(n,\alpha)$. We further know that $P_1\cup P_2$ contains $n_1+n_2$ vertices of $K$ and has a ratio $\alpha'$ with $\alpha_1\leq \alpha'\leq \alpha_2$. By Lemma~\ref{lem::splitdense} ($d_l$-difference is non-negative), and with $\alpha'\geq \alpha_1$ we can estimate:

\begin{align*}&-d(P_1)-d(P_2)+d(P_1\cup P_2)\\
&=-d_l(n_1,\alpha_1)-d_l(n_2,\alpha_2)+d_l(n_1+n_2,\alpha')-d_c(n_1,\alpha_1)-d_c(n_2,\alpha_2)+d_c(n_1+n_2,\alpha')\\&\geq \frac{1}{2(1+\alpha_1)}+\frac{1}{2(1+\alpha_2)}-\frac{1}{2(1+\alpha')}\\&\geq \frac{1}{2(1+\alpha_2)}\end{align*}

Since we have $\alpha_2\leq\alpha$, this gives:
\begin{align*}d({\cal P'})&=d({\cal P})-d(P_1)-d(P_2)+d(P_1\cup P_2)\\
&\geq \rho_k(n,\alpha) -d(P_1)-d(P_2)+d(P_1\cup P_2)\\
&\geq \rho_k(n,\alpha)+\frac{1}{2(1+\alpha_2)}\\
&\geq \rho_k(n,\alpha)+\frac{1}{2(1+\alpha)}\\
&= \rho_{k-1}(n,\alpha)
\end{align*}

This contradicts the choice of $k$, since we now have a partition into $k$ non-trivial parts of density at least $\rho_{k-1}(n,\alpha)$.

\medskip

Thus by our choice of $k$, the partition ${\cal P}^*$ has one part $P_1$ with $\alpha_1<\alpha$ and $k$ parts $P_2,\dots, P_{k+1}$ with $\alpha_i>\alpha$. Note here that if $\alpha_1$ were equal to $\alpha$ then in order to distribute all vertices among the parts, we would immediately get that all $\alpha_i$ are equal to $\alpha$, thus a proportional split where we already know that the density is equal to $\rho_{k+1}(n,\alpha)<\rho_k(n,\alpha)$.

Consider the formal expression of the density of ${\cal P}^*$:
\begin{align*}
d({\cal P})&=\sum_{i=1}^{k+1}\frac{1}{n_i(1+\alpha_i)}\left(\frac12 n_i(n_i-1)+\alpha_in_i^2\right)\\
&= \frac12\sum_{i=1}^{k+1}\frac{1}{n_i(1+\alpha_i)}\left(n_i^2-n_i+2\alpha_in_i^2\right)\\
&= \frac12\sum_{i=1}^{k+1}\frac{1}{n_i(1+\alpha_i)}\left(2n_i^2(1+\alpha_i)-n_i^2-n_i\right)\\
&= n-\frac12\sum_{i=1}^{k+1}\frac{1}{n_i(1+\alpha_i)}\left(n_i(n_i+1)\right)\\
&= n-\frac12\sum_{i=1}^{k+1}\frac{1+n_i}{1+\alpha_i}\\
\end{align*}

Thus maximizing $d({\cal P}^*)$ corresponds to minimizing $\sum_{i=1}^{k+1}\frac{1+n_i}{1+\alpha_i}$ (dropping the constant $\frac12$ for the optimization step). We now consider $d({\cal P}^*)$  as continuous function as if we were allowed to cut vertices and edges in pieces which can only increase the density value. We will use methods from nonlinear optimization to learn something about the structure of local optima of $d({\cal P}^*)$.  Not to actually compute the optimum, since the expressions get quite complicated with the $2k+2$ variables. Luckily we know that there is only one part with $\alpha_i<\alpha$ and this will simplify the expressions we need to consider immensely.

Formally finding the best value for  $d({\cal P}^*)$ as optimization problem OP over $2k+2$ variables summarized as vector $x=(n_1,\dots, n_{k+1},\alpha_1,\dots,\alpha_{k+1})$ looks as follows:

\begin{align*} \text{minimize} \ \  & f(x)=\sum_{i=1}^{k+1}\frac{1+n_i}{1+\alpha_i} \\
\text{s.t.} \ \ & h_1(x)=\sum_{i=1}^{k+1} n_i -n =0\\
 & h_2(x)=\sum_{i=1}^{k+1}\alpha_i n_i-\alpha n=0\\
 & g_i(x)=1-n_i\leq 0 \text{ \ for all } \ 1\leq i\leq k+1\\
 \end{align*}
 
 The inequalities for $g_i$ are important to enforce that there are actually $k+1$ parts. Otherwise we already know that the optimum would chose to build only one part with all vertices. Since we are in the non-trivial case of $n>k+1$ and $\alpha>0$ (note that $k,n$ and $\alpha$ are constants in this optimization problem, not variables), there is no feasible point such that all of the $g_i$ are active at the same time or such that all $\alpha_i$ are zero. This implies that the set of active constraints for any feasible point satisfy the Linear Independence Constraint Qualification constraint (LICQ). Thus any minimum of the above OP needs to satisfy the KKT conditions (see a standard textbook containing methods for nonlinear optimization, like~\cite{NoceWrig06}).
 
 The Lagrangian for this optimization problem with $\lambda_i$ for $h_i$ and $\mu_i$ for $g_i$, and $\lambda=(\lambda_1,\lambda_2)$, $\mu=(\mu_1,\dots,\mu_{k+1})$ is:
 
 \[L(x,\lambda,\mu)=\sum_{i=1}^{k+1}\frac{1+n_i}{1+\alpha_i}+\lambda_1(\sum_{i=1}^{k+1} n_i -n)+\lambda_2(\sum_{i=1}^{k+1}\alpha_i n_i-\alpha n)+\sum_{i=1}^{k+1} \mu_i(1-n_i)\]
 
 A local minimum $x^*$ for OP satisfies the KKT constraints:
 
 \begin{itemize}
 \item[]
There exist $\lambda^*,\mu^*$ such that:
\begin{itemize}
\item the gradient of $L$ is zero at $(x^*,\lambda^*,\mu^*)$
\item $x^*$ is feasible for the OP
\item $\lambda^*\geq 0$
\item $\mu_i(1-n_i)=0$ for all $i$
\end{itemize}
\end{itemize}

The first condition gives the equations:
\begin{itemize}
\item  $\frac{\partial L}{\partial n_i}(x,\lambda,\mu)=\frac{1}{1+\alpha_i}+\lambda_1+\lambda_2\alpha_i-\mu_i=0$
\item  
$\frac{\partial L}{\partial \alpha_i}(x,\lambda,\mu)=-\frac{1+n_i}{(1+\alpha_i)^2}+\lambda_2n_i=0$

\end{itemize}

Reordering the second equation gives:
\begin{align*}&-\frac{1+n_i}{(1+\alpha_i)^2}+\lambda_2n_i=0\\
\Leftrightarrow \ \ \ &\frac{1+n_i}{(1+\alpha_i)^2}=\lambda_2n_i\\
\Leftrightarrow  \ \ \ & 1+n_i=\lambda_2n_i(1+\alpha_i)^2\\
\Leftrightarrow  \ \ \ &\frac{1}{n_i}+1=\lambda_2(1+\alpha_i)^2\\
\end{align*}

We observe that $\alpha_i$ is anti-proportional to $n_i$, and for all $i$ we arrive at the same expression $\alpha_i=t(n_i,\lambda_2)$ for some function $t$. If we plug this into the gradient condition for $n_i$:

\begin{align*}&\frac{1}{1+\alpha_i}+\lambda_1+\lambda_2\alpha_i-\mu_i=0\\
\Leftrightarrow  \ \ \ &\frac{1}{1+t(n_i,\lambda_2)}+\lambda_1+\lambda_2t(n_i,\lambda_2)-\mu_i=0
\end{align*}

We see that resolving this gives the same expression for each $n_i$, with the only value differing for different $i$ is $\mu_i$. The KKT condition $\mu_i(1-n_i)=0$ shows, that some $\mu_i$ can only have a non-zero influence on $n_i$, if $n_i=1$. Thus the local optimum can have at most two sizes of $n_i$, either 1 or some fixed value $r>0$.

Since the $\alpha_i$ are anti-proportional to $n_i$ and are only an expression depending on $n_i,\lambda_2$, we know that there are also at most two different values of $\alpha_i$, and that the larger one belongs to the $n_i$ of size 1.
\medskip

Recall now that we only have one set with $\alpha_i\leq \alpha$. This implies that the local optimum has one set that contains $n-k$ clique vertices, and all the others contain exactly one clique vertex. For this very particular case, we can explicitly compute the density to show that it is too small.  The function $f$ in this case simplifies to \[f(x)=\sum_{i=1}^{k+1}\frac{1+n_i}{1+\alpha_i}=\frac{n-k+1}{1+\alpha_1}+\frac{2k}{1+\alpha_2}\]

Plugging this into the density expression, we get:

\[d({\cal P}^*)-\rho_k(n,\alpha)=\frac12\left(\frac{n+k}{1+\alpha}-\frac{n-k+1}{1+\alpha_1}-\frac{2k}{1+\alpha_2}\right)\]

Calculating with $a_i=1+\alpha_i$ and $a=\alpha+1$ for easier representation, we get some useful inequalities. From the definition of $\alpha=\frac{|I|}{|K|}=\frac{|I|}{n}$, and the fact that we have $k$ sets with the same ratio $\alpha_2$ that can together contain at most $|I|$ vertices: $\alpha_2\leq \frac{|I|}{k}=n\alpha\frac1k$, thus $a_2\leq n\alpha\frac1k+1$, and thus since $n> k$:
\begin{equation}\label{eqnph1}a_2< na\frac1k\end{equation}

From $\alpha_1(n-k)+\alpha_2k=\alpha n$ which implies  $a_1(n-k)+a_2k=a n$, we get
\begin{equation}\label{eqnph2}a_1=\frac{an-a_2k}{n-k}\end{equation}
(Observe that $a_1> 0$, since by our above bounds on $a_2$ we have $an-a_2k>an-na\frac1k k=0$.). 

Reformulating the density difference with $a, a_1, a_2$ gives:
\[d({\cal P}^*)-\rho_k(n,\alpha)=\frac12\left(\frac{n+k}{a}-\frac{n-k+1}{a_1}-\frac{2k}{a_2}\right)\]
\noindent
This value is negative because (for some $\beta>0$ that comes from getting rid of fractions): 

\begin{align*} \ &\frac{n+k}{a}-\frac{n-k+1}{a_1}-\frac{2k}{a_2}\\
&=\frac{n+k}{a}-\frac{(n-k+1)(n-k)}{an-a_2k}-\frac{2k}{a_2} & \text{by} \ (\ref{eqnph2})\\
&< \frac{n+k}{a}-\frac{n^2-k^2}{an-a_2k}-\frac{2k}{a_2}\\
&=\beta((n+k)(an-a_2k)a_2-(n^2-k^2)aa_2-2k(an-a_2k)a)\\
&=\beta(n^2aa_2-nka_2^2 +nkaa_2-k^2a_2^2 -n^2aa_2+k^2aa_2-2nka^2+2k^2aa_2)\\
&=\beta(nk(aa_2-a_2^2-2a^2)  +k^2(3aa_2-a_2^2)\\
&= \beta(nk(aa_2-a_2^2)-nk2a^2+k^22aa_2 +k^2(aa_2-a_2^2)\\
&\leq \beta((aa_2-a_2^2)(nk+k^2)-nk2a^2+k^22a(na\frac1k))& \text{by} \ (\ref{eqnph1})\\
&= \beta(a_2(a-a_2)(nk+k^2))
\end{align*}
Which is smaller than 0 since $n>k$ by our counterexample and $a<a_2$ by $\alpha<\alpha_2$ from Claim~\ref{claimsize}.

This shows that our assumed best partition into $k+1$ non trivial parts ${\cal P}^*$ has density less than $\rho_k(n,\alpha)$ which yields the overall contradiction to the assumption that the Lemma statement is false.
\end{proof}

\paragraph*{Proof of Theorem 3}

\splitnph*
\begin{proof}
 We give a reduction from the $\mathsf{NP}$-hard problem  {\sc Dominating Set} to  $\text{MDGP}$. In the  {\sc Dominating Set} problem, the input is a pair $(G,k)$ where $G=(V,E)$ is a graph and $k\in\mathbb N$, and the task is to decide if there exists a subset $C\subseteq V$ with $|C|\leq k$ such that $V=N[C]$. Indeed, {\sc Dominating Set} is $\mathsf{NP}$-hard \cite{DBLP:conf/coco/Karp72}.
 
Let $(G,k)$ with $G=(V,E)$ and $V=\{v_1,\dots,v_n\}$ be an instance of {\sc Dominating Set}. We assume without loss of generality that $k\leq \frac{n}{2}$ (observe that if $k> \frac{n}{2}$, we can delete isolated vertices decrementing $k$ and simply reply yes if some  $(G',k')$ remains where the minimum degree in $G'$ is 1 and the number of vertices is still less than $2k$).
  As reduction, we construct a split graph $G':=(K\cup I,E')$ as instance of $\text{MDGP}$ as follows.

Overall idea: The sets of vertices of $G'$ contain copies of $V$ on both sides plus $n$ large sets of filler-vertices to force a certain structure of any (almost) proportional split. An illustration of this construction can be found in Figure~\ref{npfig}. 

To better describe the construction, we use variables $y,h,r,r'$ as shorthands for $y=n^{20}$, $h=n^4$, $r=n^7$ and $r'=n^{14}$ (these names indicate the role of the vertex sets with these sizes, and the differing sizes allow for the control over the structure of solutions).
Before we define the sets $K$ and $I$ that will correspond to the clique and the independent set of the split graph of the construction, let us first define the following auxiliary sets.

For every vertex $v_{i}$ of $G$ we define the set
$V^{K}_{i}=\{v_{i}^{j}\mid 1\leq j\leq r\}$ containing $r$ copies of $v_{i}$, $V^{K} =  \cup_{i\in [n]}V^{K}_{i}$ is the union of all $V^{K}_{i}$, $i\in [n]$, the set $V^{I}_{i}=\{\overline{v}_{i}^{j}\mid 1\leq j\leq r'\}$ and $V^{I}  =  \cup_{i\in [n]}V^{I}_{i},
$ is the union of all $V^{I}_{i}$, $i\in [n]$.
Moreover,  for every $i\in [k]$, we define the sets 
\begin{eqnarray*}
W_{i} & = &\{w_{i}^{j}\mid 1\leq j\leq h\}~~~(W_{i} \text{ is a set of } h \text{ vertices})\\
Y_{i} & = &\{y_{i}^{j}\mid 1\leq j\leq y\}~~~~(Y_{i} \text{ is a set of } y \text{ vertices})\\
W & = & \displaystyle \bigcup_{i\in [k]} W_{i},~~~~~~~~~~~~~(W \text{ is the union of all } W_{i})\\
Y & = & \bigcup_{i\in [k]} Y_{i},~~~~~~~~~~~~~~(Y \text{ is the union of all } Y_{i})
\end{eqnarray*}
and for every $i\in [n-k]$, we define the sets
\begin{eqnarray*}
X_{i} & = &\{x_{i}^{j}\mid 1\leq j\leq h\}~~~~~~~~~~(X_{i} \text{ is a set of } h \text{ vertices})\\
Z_{i} & = &\{z_{i}^{j}\mid 1\leq j\leq nr'+y\}~~~(Z_{i} \text{ is a set of } nr'+y \text{ vertices})\\
X & = & \bigcup_{i\in [n-k]}X_{i},~~~~~~~~~~~~~~~~~~(X \text{ is the union of all } X_{i}) \text{ and}\\ 
Z & = & \bigcup_{i\in [n-k]}Z_{i}~~~~~~~~~~~~~~~~~~~~(Z \text{ is the union of all } Z_{i}).
\end{eqnarray*}
Last, we define the set $Y'=\{y_{i}'\mid 1\leq i\leq n(k-1)r'\}$.
The vertex sets $K$ and $I$ are defined as follows:
\begin{eqnarray*}
K & = & V^{K} \cup W\cup X\\
I & = & V^{I} \cup Y\cup Z \cup Y'
\end{eqnarray*}

The edge set of this graph is the union of the following sets:
\begin{eqnarray*}
    E_{K} & = & \{\{u,v\}\mid u,v\in {K}\}\\
    E_{V} & = & \{\{u,v\}\mid u\in V^{K},v\in Y\cup Y'\cup Z\}\\
    E_{d} & = & \{\{v_i^t,\overline{v}_j^s\}\mid \{v_i,v_j\}\in E, t\in [r], j\in [r'] \text{ or } i=j \}\\
    E_{V^{I}w} & = & \{\{v,w\}\mid v\in V^{I},w\in W\}\\
    E_{wy'} & = &\{\{w,y'\}\mid w\in W, y'\in Y'\}\\
    E_{wy} & = & \bigcup_{i\in [k]} \{\{w_{i}^{j},y_{i}^{t}\}\mid j\in [h], t\in [y]\}\\
    E_{xz} & = & \bigcup_{i\in [n-k]} \{\{x_{i}^{j},z_{i}^{t}\}\mid j\in [h], t\in [nr'+y]\}
\end{eqnarray*}

Given all the chosen values, the construction can be done in polynomial time. Moreover,
observe that $|K|=(r+h)n$ and $|I|=(nr'+y)n$ and therefore $\alpha:= \frac{|I|}{|K|}=\frac{r'n+y}{r+h}$. This concludes the construction.
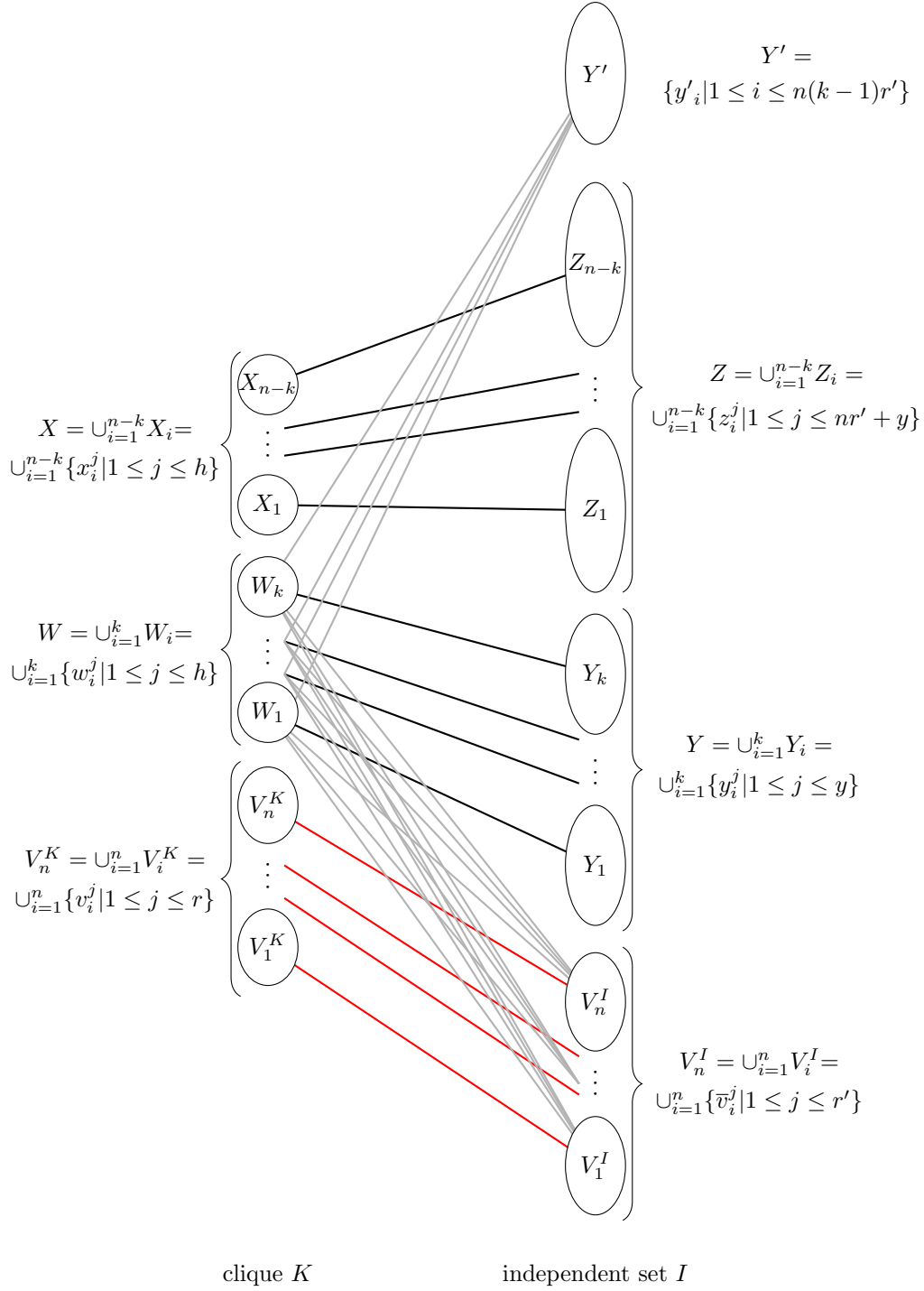
\begin{figure}[!h]
\begin{center}
\begin{tikzpicture}[every node/.style={fill=black,inner sep=1.5pt}, scale=0.8]
\draw [red, thick] (0,4)--(6,0);
\draw [red, thick] (0,6.6)--(6,3);
\draw [red, thick] (0.3,4.9)--(5.7,1.3);
\draw [red, thick] (0.3,5.5)--(5.7,2);
\draw [black, thick] (0,8.3)--(6,5.5);
\draw [black, thick] (0,10.6)--(6,9);
\draw [black, thick] (0.3,9)--(5.7,7);
\draw [black, thick] (0.3,9.6)--(5.7,7.8);
\draw [black, thick] (0,12.1)--(6,12);
\draw [black, thick] (0,14.3)--(6,16.5);
\draw [black, thick] (0.3,13)--(5.7,13.8);
\draw [black, thick] (0.3,13.5)--(5.7,14.5);
\draw [black!30, thick] (0.2,7.9)--(6,20);
\draw [black!30, thick] (0,10.6)--(6,20);
\draw [black!30, thick] (0.3,9)--(6,20);
\draw [black!30, thick] (0.3,9.6)--(6,20);
\draw [black!30, thick] (0,8.3)--(6,0);
\draw [black!30, thick] (0,8.3)--(6,3);
\draw [black!30, thick] (0,8.3)--(5.7,1.5);
\draw [black!30, thick] (0,10.6)--(6,0);
\draw [black!30, thick] (0,10.6)--(6,3);
\draw [black!30, thick] (0,10.6)--(5.7,1.5);
\draw [black!30, thick] (0.3,9)--(6,0);
\draw [black!30, thick] (0.3,9)--(6,3);
\draw [black!30, thick] (0.3,9)--(5.7,1.5);
\draw [black!30, thick] (0.3,9.6)--(6,0);
\draw [black!30, thick] (0.3,9.6)--(6,3);
\draw [black!30, thick] (0.3,9.6)--(5.7,1.5);
\filldraw[fill=white, draw=black] (0,4) ellipse (0.55cm and 0.7cm);
\node[fill=none, draw=none, text=black] at (0,4) {$V_1^K$};
\node[fill=none, draw=none, text=black] at (0,5.1) {$\cdot$};
\node[fill=none, draw=none, text=black] at (0,5.3) {$\cdot$};
\node[fill=none, draw=none, text=black] at (0,5.5) {$\cdot$};
\filldraw[fill=white, draw=black] (0,6.6) ellipse (0.55cm and 0.7cm);
\node[fill=none, draw=none, text=black] at (0,6.6) {$V_n^K$};
\draw [decorate,decoration={brace,amplitude=8pt}]
      (-0.5,3.1) -- (-0.5,7.4);
\node[fill=none, draw=none, text=black] at (-2.8,5.6) {$V^K=\cup_{i=1}^nV_i^K=$};
\node[fill=none, draw=none, text=black] at (-2.8,4.9) {$\cup_{i=1}^n\{v_i^j|1\leq j\leq r\}$};
\filldraw[fill=white, draw=black] (0,8.3) ellipse (0.55cm and 0.55cm);
\node[fill=none, draw=none, text=black] at (0,8.3) {$W_1$};
\node[fill=none, draw=none, text=black] at (0,9.2) {$\cdot$};
\node[fill=none, draw=none, text=black] at (0,9.4) {$\cdot$};
\node[fill=none, draw=none, text=black] at (0,9.6) {$\cdot$};
\filldraw[fill=white, draw=black] (0,10.6) ellipse (0.55cm and 0.55cm);
\node[fill=none, draw=none, text=black] at (0,10.6) {$W_k$};
\draw[decorate, decoration={brace, amplitude=8pt}]
  (-0.5,7.7) -- (-0.5,11.2);
\node[fill=none, draw=none, text=black] at (-2.8,9.75) {$W=\cup_{i=1}^k W_i$=};
\node[fill=none, draw=none, text=black] at (-2.8,9.05) {$\cup_{i=1}^k\{w_i^j|1\leq j\leq h\}$};
\filldraw[fill=white, draw=black] (0,12.1) ellipse (0.55cm and 0.55cm);
\node[fill=none, draw=none, text=black] at (0,12.1) {$X_1$};
\node[fill=none, draw=none, text=black] at (0,13) {$\cdot$};
\node[fill=none, draw=none, text=black] at (0,13.2) {$\cdot$};
\node[fill=none, draw=none, text=black] at (0,13.4) {$\cdot$};
\filldraw[fill=white, draw=black] (0,14.3) ellipse (0.55cm and 0.55cm);
\node[fill=none, draw=none, text=black] at (0,14.3) {$X_{n-k}$};
\draw[decorate, decoration={brace, amplitude=8pt}]
  (-0.5,11.5) -- (-0.5,14.9);
\node[fill=none, draw=none, text=black] at (-2.8,13.5) {$X=\cup_{i=1}^{n-k} X_i$=};
\node[fill=none, draw=none, text=black] at (-2.8,12.8) {$\cup_{i=1}^{n-k}\{x_i^j|1\leq j\leq h\}$};
\node[fill=none, draw=none, text=black] at (0,-2) {clique $K$};
\node[fill=none, draw=none, text=black] at (6,-2) {independent set $I$};
\filldraw[fill=white, draw=black] (6,0) ellipse (0.55cm and 0.9cm);
\node[fill=none, draw=none, text=black] at (6,0) {$V_1^I$};
\node[fill=none, draw=none, text=black] at (6,1.3) {$\cdot$};
\node[fill=none, draw=none, text=black] at (6,1.5) {$\cdot$};
\node[fill=none, draw=none, text=black] at (6,1.7) {$\cdot$};
\filldraw[fill=white, draw=black] (6,3) ellipse (0.55cm and 0.9cm);
\node[fill=none, draw=none, text=black] at (6,3) {$V_n^I$};
\draw[decorate, decoration={brace, mirror, amplitude=8pt}]
  (6.5,-1) -- (6.5,4);
\node[fill=none, draw=none, text=black] at (9,1.9) {$V^I=\cup_{i=1}^nV_i^I$=};
\node[fill=none, draw=none, text=black] at (9,1.2) {$\cup_{i=1}^n\{\overline{v}_{i}^{j}|1\leq j\leq r'\}$};
\filldraw[fill=white, draw=black] (6,5.5) ellipse (0.55cm and 1.1cm);
\node[fill=none, draw=none, text=black] at (6,5.5) {$Y_1$};
\node[fill=none, draw=none, text=black] at (6,7) {$\cdot$};
\node[fill=none, draw=none, text=black] at (6,7.2) {$\cdot$};
\node[fill=none, draw=none, text=black] at (6,7.4) {$\cdot$};
\filldraw[fill=white, draw=black] (6,9) ellipse (0.55cm and 1.1cm);
\node[fill=none, draw=none, text=black] at (6,9) {$Y_k$};
\draw[decorate, decoration={brace, mirror, amplitude=8pt}]
  (6.5,4.3) -- (6.5,10.2);
\node[fill=none, draw=none, text=black] at (9,7.7) {$Y=\cup_{i=1}^k Y_i=$};
\node[fill=none, draw=none, text=black] at (9,7) {$\cup_{i=1}^k\{{y}_{i}^{j}|1\leq j\leq y\}$};
\filldraw[fill=white, draw=black] (6,12) ellipse (0.55cm and 1.5cm);
\node[fill=none, draw=none, text=black] at (6,12) {$Z_1$};
\node[fill=none, draw=none, text=black] at (6,14) {$\cdot$};
\node[fill=none, draw=none, text=black] at (6,14.2) {$\cdot$};
\node[fill=none, draw=none, text=black] at (6,14.4) {$\cdot$};
\filldraw[fill=white, draw=black] (6,16.5) ellipse (0.55cm and 1.5cm);
\node[fill=none, draw=none, text=black] at (6,16.5) {$Z_{n-k}$};
\draw[decorate, decoration={brace, mirror, amplitude=8pt}]
  (6.5,10.5) -- (6.5,18);
\node[fill=none, draw=none, text=black] at (9.5,14.5) {$Z=\cup_{i=1}^{n-k} Z_i=$};
\node[fill=none, draw=none, text=black] at (9.5,13.7) {$\cup_{i=1}^{n-k}\{{z}_{i}^{j}|1\leq j\leq nr'+y\}$};
\filldraw[fill=white, draw=black] (6,20) ellipse (0.55cm and 1.3cm);
\node[fill=none, draw=none, text=black] at (6,20) {$Y'$};
\node[fill=none, draw=none, text=black] at (9.5,20.35) {$Y'=$};
\node[fill=none, draw=none, text=black] at (9.5,19.65) {$\{{y'}_{i}|1\leq i\leq n(k-1)r'\}$};
\end{tikzpicture}
\caption{Illustration of the construction for Theorem~\ref{thm::splitnph}. Black and gray edges denote complete connections, red edges are drawn according to the edges in $E$. The edges connecting the $v_i^j$ to everything else and the edges of the clique $K$ are omitted.}
\label{npfig}
\end{center}
\end{figure}

We claim that $(G,k)$ is a yes-instance for {\sc Dominating Set} if and only if $(G',n,D)$ is a yes-instance for \DGP{} with $D=d_l(|K|,\alpha)+nd_c(\alpha)$. \\

(\textbf{Forward direction}) For the forward direction, let us assume  that $(G,k)$ with $G=(V,E)$ and $V=\{v_1,\dots,v_n\}$  and $k<\frac n2$ is a yes-instance for  {\sc Dominating Set}. Let $S\subset V$ be a dominating set for $G$ with $|S|=k$. Let $S=\{v_{p_1},\dots, v_{p_k}\}$ and let $V_1,\dots, V_k$ be any partition of $V$ such that $\{v_{p_i},u\}\in E$ or $u=v_{p_i}$ for each $u\in V_i$, for all $1\leq i\leq k$. Such a partition exists, because $S$ is a dominating set and thus each $u$ in $V$ is equal or adjacent to at least one $v_{p_i}$-vertex. For every $i\in [k]$, let $l_{i}=|V_{i}|$. Denote further $V\setminus S=\{v_{p_{k+1}},\dots, v_{p_{n}}\}$. Let $Y'_{i}$, $i\in [k]$ be a partition of $Y'$ as follows:
$$Y'_{i}=\left\{y'_{t}\mid r'\left(n(i-1)-\sum_{p=1}^{i-1}l_{p}\right)+1\leq t \leq r'\left(ni-\sum_{p=1}^{i}l_{p}\right)\right\},$$
where $\sum_{p=1}^{0}l_{p}=0$. That is $Y'_{i}$ contains $(n-l_{i})r'$ elements of $Y'$. (Observe that by definition $Y'_{i}\cap Y'_{j}=\emptyset$ for every $i\neq j$. Moreover, since $\sum_{i=1}^{k}|Y_{i}|=|Y'|$, they indeed form a partition of $Y'$.)
We will construct a partition with $n$ sets of density $D$ for the corresponding graph $G'$ as follows. 

We create the partition  $\mathcal P= \{D_1,\dots, D_k, P_1,\dots, P_{n-k}\}$ defined as follows:
For every $i\in [k]$, let
\begin{align*}
D_i & =V^{K}_{p_{i}} & \text{all $r$ copies of $v_{p_i}$ in $K$} \\
&\cup W_{i} & \text{the $h$ vertices of $W_{i}$} \\
&\cup_{v_{t}\in V_{i}}V^{I}_{t} & \text{all copies of vertices assigned to $v_{p_i}$ by $V_i$ that are in $V^{I}$} \\
&\cup Y_{i} & \text{the $y$ vertices of $Y_{i}$}\\ 
&\cup Y'_{i}& \text{buffer from $Y$ to get $r'n+y$ vertices of $I$ in $D_{i}$}\end{align*}

\noindent Moreover, for every $i\in [n-k]$, let
\begin{align*}
P_i &=V^{K}_{p_{k+i}}
 & \text{all copies of $v_{p_{i+k}}$ in $K$}\\
&\cup X_{i} & \text{the $h$ vertices of $X_{i}$} \\ 
&\cup Z_{i} & \text{the $nr'+y$ vertices of $Z_{i}$}
\end{align*}

We now claim that $\mathcal{P}$ is a proportional partition of $G'$ where each part induces a complete split graph.
Towards showing that it is proportional
let first $i\in [k]$. Observe that by construction $K\cap D_{i}=V_{p_{i}}^{K}\cup W_{i}$, $I\cap D_{i}=Y_{i}\cup Y_{i}'\cup \cup_{v_{t}\in V_{i}}V_{t}^{I}$ and thus, also $|K\cap D_{i}|=r+h$ + $|I\cap D_{i}|=y+r'(n-l_{i})+l_{i}r'=nr'+y$.
Let now $i\in [n-k]$. Observe that by construction $K\cap P_{i}=V^{K}_{p_{k+i}}\cup X_{i}$, $I\cap P_{i}=Z_{i}$, and thus, also $|K\cap P_{i}|=r+h$ and $|I\cap P_{i}|=nr'+y$.
Therefore, for every $i\in [k]$ and every $j\in [n-k]$, $|D_i\cap K|=|P_j\cap K|=r+h=\frac{|K|}{n}$ and  $|D_i\cap I|=|P_j\cap I|=nr'+y=\frac{|I|}{n}$.

Towards showing that each part is 
complete let first $i\in [k]$, then in 
order to show that $D_{i}$ induces a 
complete split graph, 
since $K$ is a clique, it is enough to 
show that every vertex in 
$Y_{i}\cup Y_{i}'\cup \cup_{v_{t}\in V_{t}^{I}}$ 
is connected to every vertex of 
$V_{p_{i}}^{K}\cup W_{i}$. 
The edges between $W_{i}$ and $Y_{i}$ 
have been added to $G'$ by the edge 
set $E_{wy}$, the edges between $W_{i}$ 
and $Y_{i}'$ have been added to $G'$ by the edge set $E_{wy'}$, and the 
edges between $W_{i}$ and 
$\cup_{v_{t}\in V_{i}}V_{t}^{I}$ have 
been added to $G'$ by the edge set 
$E_{V^{I}w}$.
The edges between $V_{p_{i}}$ and the  
sets $Y_{i}$ and $Y_{i}'$ have been 
added to $G'$ by the edge set 
$E_{V}$. Furthermore, let $v_{p_{i}}^{q}\in V_{p_{i}}$, for some $q\in [r]$ and $v_{t}^{b} \in V_{t}^{I}$ where $t$ is such that $v_{t}\in V_{i}$ and $b\in [r']$. Then by the construction of $V_{i}$, $\{v_{p_{i}},v_{t}\}$ is an edge of $G$ or $v_{p_{i}}=v_{t}$ and thus $\{v_{p_{i}}^{q},v_{t}^{b}\}$ has been added to $G$ by the edge set $E_{d}$.
Let now $i\in [n-k]$.
As before, in 
order to show that $P_{i}$ induces a 
complete split graph, 
since $K$ is a clique, it is enough to 
show that every vertex in 
$Z_{i}$ 
is connected to every vertex of 
$V_{p_{k+i}}^{K}\cup X_{i}$.
The edges between $X_{i}$ and $Z_{i}$ have been added to $G'$ by the edge
set $E_{xz}$ and the edges between 
$V^{K}_{p_{k+i}}$ and $Z_{i}$ have been added to $G'$ by the edge set $E_{V}$.

Thus $\mathcal P$ is a proportional partition into $n$ complete split graphs, which by Observation~\ref{obs::splitdense_proportional} shows that $d(\mathcal P)=d_l(|K|,\alpha)+nd_c(\alpha)=D$.

Therefore, $(G',n,D)$ is a yes-instance for \DGP{} and this concludes the forward direction of the reduction. \\

(\textbf{Reverse direction}) Let us now assume that  $(G',n,D)$ is a yes-instance and thus there exists a partition $\mathcal P$ of $G'$ with $|\mathcal P|=n$ and $d(\mathcal P)\geq D$. 

Before we go into the numerical 
details of the proof, let us first 
illustrate the intuition behind it. Essentially, our proof builds on the following claims.

We use the estimates form Observation~\ref{obs::splitdense_missing} and  Lemma~\ref{lem::splitdense2} to derive bounds on missing edges and the ratio of each part in $\mathcal P$:
\missingedges*

\medskip
We then use these estimates to establish a particular structure of $\mathcal P$ with respect to the filler-vertices $W,X,Y,Z$.

\structural*

\medskip

Finally, we show how the parts of $\mathcal{P}$ interact with the sets $V^{K}$ and $V^{I}$. In particular,
\ViVkbounds*

\medskip

Assuming the above claims hold, we show how to construct a dominating set of $G$ of size at most $k$.\\

{\bf Creating the Dominating Set:}
 By the second item of Claim~\ref{clm:ViVkbounds} each part $D_i$, $i\in [k]$, contains at least $\frac{r}{n^3}$ vertices from $V^{K}$.
 Let $\mathcal{D}=\{v_{r_{1}},v_{r_{2
}},\dots, v_{r_{k}}\}$ be some set of $k$ vertices of $G$ such that $v_{r_{i}}^{j}\in D_{i}$ for some $j\in [r]$ (where it might be the case that $v_{r_{i}}=v_{r_{j}}$ for $i\neq j$.)

We prove that $\mathcal{D}$ is a dominating set of $G$.
Assume, towards a contradiction, 
that there is some $v_t\in V$ that 
is not equal or adjacent to any of the 
vertices in $\mathcal{D}$.

By the first item of 
Claim~\ref{clm:ViVkbounds}, at most 
$(n-k)n^{6}$ vertices of $V^{I}$ belong 
to the parts $P_{i}$, $i\in [n-k]$. Further, by item two of Claim~\ref{clm:missingedges} only $\mathcal O(\frac {r'}{n})$ vertices of $V^{I}$ belong to $R$.
Since the set $V_{t}^{I}$ that contains 
all copies of $v_{t}$ in the set $I$ has
size $r'=n^{14}$,
there are $\Omega(r')$ of the 
vertices of $V_{t}^{I}$ contained 
in the parts 
$D_{1}$, $D_{2}$, $\dots$, $D_{k}$.
By the pigeonhole principle, there is some $i_{0}\in [k]$ such that $D_{i_{0}}$ contains at least a $\frac{1}{k}$ fraction of the vertices of $V_{t}^{I}$.

Since $v_{t}$ is not adjacent to any 
of the vertices in $\mathcal{D}$, its respective copies in $V_{t}^{I}$ are not 
adjacent to $v_{r_{i_{0}}}$ in the constructed graph.
However, by construction of the set $\mathcal{D}$, we have
$v_{r_{i_{0}}^{j}}\in D_{i_{0}}$. This means that there are  $\frac 1k\Omega(r')$ missing edges in  
$D_{i_{0}}$. Since we assumed  $k<\frac{n}{2}$ in the very beginning, this gives $\Theta(\frac {r'}{n})=\Theta(n^{13})$ missing edges.

Thus, $D_{i_{0}}$ has more that 
$n^{9}$ missing edges, a contradiction 
to Claim~\ref{clm:missingedges}. Thus 
$\mathcal{D}$ is a dominating set of $G$ whose 
size is by construction at most $k$, 
implying that $(G,k)$ is a yes-instance 
for \textsc{Dominating Set}. \\

Thus, what remains is to prove the above claims. \\

\noindent For the calculations, we first argue that for any part $P$ of $\mathcal{P}$ the following three inequalities hold.

\begin{eqnarray} 
\!\!\!D-d(\mathcal P)\!\!&\!\!>&\!\!\!\!d_l(|K|,\alpha)+nd_c(\alpha)-\left(d(P)+d_l\left(|K|-|P\cap K|,\tfrac{|I|-|P\cap I|}{|K|-|P\cap K|}\right)\right) \text{ and}
\label{overest}\\
\!\!\!D-d(\mathcal P)\!\!&\!\!>&\!\!\!\!d_l(|K|,\alpha)+nd_c(\alpha)-d_l\left(|P\!\cap\! K|,\tfrac{|P\cap I|}{|P\cap K|}\right)-d_l\left(|K|-|P\!\cap\! K|,\tfrac{|I|-|P\cap I|}{|K|-|P\cap K|}\right)\label{overest2}\\
\!\!\!D-d(\mathcal P)\!\!&\!\!>&\!\!\!\!\frac{1}{2(1+\frac{|P\cap I|}{|P\cap K|})}-\frac{n}{2(1+\alpha)}.
\label{overest1}
\end{eqnarray} 
Notice that for any partition $\mathcal P$ of a split graph we have $d(\mathcal P)\leq \sum_{P\in\mathcal P} d_l(P)$ with the equality holding only if the graph is complete and $|\mathcal P|=1$. Furthermore, due to Lemma~\ref{lem::splitdense2} for any partition $\mathcal P$ of a graph $G$ we have that $\sum_{P\in\mathcal P}d_l(P)\leq \sum_{P\in\mathcal P} d_l(|V(P)|,\alpha)= d_l(G)\leq d_l(G_c)$, where $G_c$ is the complete split graph on the same vertex set as $G$. Hence, Inequality~\eqref{overest} holds.

Towards Inequalities~\eqref{overest2} and~\eqref{overest1},
observe first that the maximum possible 
density of $P$ is obtained if $P$ is 
a complete split graph which would 
imply that 
\begin{equation*} d(P) = d_l\left(|P\cap K|,\frac{|P\cap I|}{|P\cap K|}\right)+ d_c\left(\frac{|P\cap I|}{|P\cap K|}\right).
\end{equation*}
Then,
\begin{eqnarray*}\displaystyle
    D-d(\mathcal P) & > & d_l(|K|,\alpha)+nd_c(\alpha)-\left(d(P)+d_l\left(|K|-|P\cap K|,\tfrac{|I|-|P\cap I|}{|K|-|P\cap K|}\right)\right)\\
& > & d_l(|K|,\alpha)+nd_c(\alpha)-d_l\left(|P\cap K|,\tfrac{|P\cap I|}{|P\cap K|}\right)-d_c\left(\frac{|P\cap I|}{|P\cap K|}\right)-d_l\left(|K|-|P\cap K|,\tfrac{|I|-|P\cap I|}{|K|-|P\cap K|}\right)
\end{eqnarray*}
Recall that, by definition, $d_c(l)<0$ for any $l$ and thus we can simplify the above expression by removing the positive term $-d_c\left(\frac{|P\cap I|}{|P\cap K|}\right)$ and obtain Inequality~\eqref{overest2}. 
Towards Inequality~\eqref{overest1}, 
by Lemma~\ref{lem::splitdense2}, 
it holds that
\[d_l(|K|,\alpha)\geq d_l\left(|P\cap K|,\tfrac{|P\cap I|}{|P\cap K|}\right)+d_l\left(|K|-|P\cap K|,\tfrac{|I|-|P\cap I|}{|K|-|P\cap K|}\right).\]
 Plugging this and the maximum possible value of $d(P)$ into Inequality~\eqref{overest}, we get:
\[D-d(\mathcal P)>nd_c(\alpha)- d_c\left(\tfrac{|P\cap I|}{|P\cap K|}\right).\]
By the definition of $d_c$ and $\alpha$ this results in Inequality~\eqref{overest1}. 

Before we start proving the claims, note that we will often use asymptotic arguments claiming that those only hold \emph{for large enough $n$}. Whenever we write something like this, we only assume that $n$ is larger than some constant. Since instances of constant size can be solved trivially, this is not a restriction to the correctness of the overall reduction.


Let us now prove Claim~\ref{clm:missingedges}.
\begin{proofofclaim}[Proof of Claim~\ref{clm:missingedges}]
For the first item, let us assume towards contradiction that $\mathcal P$ 
contains some part $P$ with at least
$n^9$ missing edges.
Observation~\ref{obs::splitdense_missing}
and Inequality~\eqref{overest} yield
that 
\begin{eqnarray*}
D-d(\mathcal P)& > &  d_l(|K|,\alpha)+nd_c(\alpha)-\left((d(P)+d_l\left(|K|-|P\cap K|,\tfrac{|K|-|P\cap K|}{|I|-|P\cap I|}\right)\right)\\
 & \geq &  d_l(|K|,\alpha)+nd_c(\alpha)-d_l\left(\tfrac{|P\cap I|}{|P\cap K|},|P\cap K|\right)-d_c\left(\tfrac{|P\cap I|}{|P\cap K|}\right)\\ & & +\tfrac{n^{9}}{|P|}-d_l\left(|K|-|P\cap K|,\tfrac{|I|-|P\cap I|}{|K|-|P\cap K|}\right).\end{eqnarray*}
 
\noindent By 
Lemma~\ref{lem::splitdense2} it holds
that $d_l(|K|,\alpha)\geq d_l(|P\cap K|,\tfrac{|P\cap I|}{|P\cap K|})+d_l(|K|-|P\cap K|,\tfrac{|I|-|P\cap I|}{|K|-|P\cap K|})$
and thus the $d_l$-parts cancel each other out. Therefore,
 \begin{equation}D-d(\mathcal P) >nd_c(\alpha)-d_c\left(\frac{|P\cap I|}{|P\cap K|}\right)+\frac{n^{9}}{|P|}.\end{equation}
Recall that, by definition, $d_c(l)<0$ for any $l$ and thus we can simplify the above expression by removing the positive term $-d_c\left(\frac{|P\cap I|}{|P\cap K|}\right)$ and obtain that:  
\[D-d(\mathcal P) \geq \frac{n^{9}}{|P|}-\frac{n}{2(1+\frac{|I|}{|K|})}=\frac{n^{9}}{|P|}-\frac{n|K|}{2(|K|+|I|)}\]
 
 Since $P$ is a part of a partition of a 
 graph with $|K|+|I|$ vertices 
 $|P|<|K|+|I|$. Moreover, recall that 
 the size of $K$ is $n(r+h)=n^{8}+n^{5}$
 and thus $n|K|=n(n^8+n^5)<2n^{9}$. 
 We conclude that $D-d(\mathcal{P})>\frac{n^{9}}{|K|+|I|}-\frac{n|K|}{2(|K|+|I|)}>0$, 
 a contradiction to the assumption that $d(\mathcal{P})\geq D$. 

\medskip

\noindent
For the second item, assume towards contradiction that $\mathcal P$ 
contains some part $P$   with $K\cap P\neq \emptyset$ and $\frac{|P\cap I|}{|P\cap K|}<\frac{\alpha}{2n}$. Then by Inequality~(\ref{overest1}), we get that 
\begin{align*}
D-d(\mathcal P)& >\frac{1}{2(1+\frac{|P\cap I|}{|P\cap K|})}-\frac{n}{2(1+\alpha)}\\
& >\frac{1}{2(1+\frac{\alpha}{2n})}-\frac{n}{2(1+\alpha)}\\
& = \frac{2n}{2(2n+\alpha)}-\frac{n}{2(1+\alpha)}\\
\end{align*}
Since $\alpha>n^{13}$ we see that this is positive, since:  
\begin{eqnarray*}
&0& > \frac{2n}{2(2n+\alpha)}-\frac{n}{2(1+\alpha)}\\
\Leftrightarrow & 2n(1+\alpha)& > n(2n+\alpha)\\
\Leftrightarrow & 2n+2n\alpha& > 2n^2+n\alpha\\
\Leftrightarrow & 2n+n\alpha& > 2n^2
\end{eqnarray*}
Thus we get a contradiction to the existence of such a part $P$ which shows the second item of the claim.
\end{proofofclaim}

We now proceed with the next proof.
\begin{proofofclaim}[Proof of Claim~\ref{clm:structural}]
By Lemma~\ref{lemma::kparts}, we already know that $\mathcal P$ contains at most $n$ non-trivial parts, which we for now name $S_1,\dots, S_n$. Combining the trivial parts into one does not affect the density so without loss of generality we can assume that $\mathcal P$ only contains one trivial part we name $R$.

We first show that each non-trivial part needs a certain amount of independent set vertices, in particular some of $Z\cup Y$. Let  $S$ be a part that contains at least $h$  vertices from $X\cup W$ (thus by pigeonhole principle at least one such part exists). Assume towards contradiction, that $S_1$ contains less than  $n^{10}$ vertices of $Z\cup Y$. Then $|S\cap I|<n^5+|Y'|+|V^I|=n^{10}+nr'+n(k-1)r'$  For large enough $n$, we get $|S\cap I|<nkr'\leq \frac 12n^2r'$ (recall that we assumed $k\leq \frac n2$ in the very beginning of the proof, and that by definition$r'=n^{14}$). This allows to estimate that the ratio of $S_1$ is too small (recall $\alpha>n^{13}$):

\[\frac{|S\cap I|}{|S\cap K|}<\frac{\frac 12n^2r'}{\frac 1n (|X\cup W|)}=\frac 12\frac{n^2n^{14}}{h}=\frac 12 \frac{n^{16}}{n^4}=\frac{n^{13}}{2n}<\frac{\alpha}{2n}\]

Thus we get a contradiction to item two of Claim~\ref{clm:missingedges}. It follows that $S$ needs to contain at least  $n^{10}$ vertices of $Z\cup Y$, and by pigeonhole principle, this means that either $|S\cap Y_i|\geq n^9$ for some $1\leq i\leq k$ or $|S\cap Z_j|\geq n^9$ for some $1\leq j\leq n-k$. If $|S\cap Y_i|\geq n^9$ for some $1\leq i\leq k$, then $S\cap K\subseteq V^K\cup W_i$; all other vertices in $K$ are not adjacent to the vertices in $Y_i$ which would give at least $n^9$ missing edges, a contradiction to   item one of Claim~\ref{clm:missingedges}. Similarly, if $|S\cap Z_j|\geq n^9$ for some $1\leq j\leq n-k$, then $S\cap K\subseteq V^K\cup Z_j$. In either case, $S$ in particular can only contain at most $h$ vertices of $W\cup Z$. 

Since the vertices of $W\cup Z$ need to distribute over the $n$ non-trivial sets, and any set $S$ with at least $h$ vertices can only contain at most $h$, thus  it follows that each $S_i$ needs to contain exactly $h$ vertices from $W\cup Z$. Further, each $S_i$ needs to contain at least $n^9$ vertices from some $Y_i$ or $Z_j$. In order to not exceed the limit of $n^9$ missing edges forbidden by the first item of Claim~\ref{clm:missingedges}, we conclude that $\mathcal{P}$ can be written as $\mathcal P=\{D_1,\dots, D_k, P_{1},\dots, P_{n-k},R\}$ with $W_{i}\subseteq D_{i}$, $i\in [k]$, and $X_{i}\subseteq P_{i}$, $i\in [n-k]$, and $R\subset I$.

\medskip

It remains to first bound the size of $R$, and then to distribute the vertices in $(W\cup Z)\setminus R$. To this end, define $I'=I\setminus R$.
Since $\mathcal P$ contains $n$ non-trivial parts, we can use Lemma~\ref{lemma::kparts} to strictly bound its density with $\alpha'=\frac{|I'|}{|K|}$ by the next smaller $\rho$-value:

\[d(\mathcal{P})<\rho_{n-1}(|K|,\alpha')=|K|-\frac{|K|+n-1}{2(1+\alpha')}\]
Comparing this to the target $D=\rho_n(|K|,\alpha)$ we get:

\[|K|-\frac{|K|+n-1}{2(1+\alpha')}>\rho_n(|K|,\alpha)
=|K|-\frac{|K|+n}{2(1+\alpha)}\]
 
Thus we know that $\alpha'$ satisfies 
\begin{align*} &\frac{|K|+n-1}{2(1+\alpha')}<\frac{|K|+n}{2(1+\alpha)}\\
\Leftrightarrow &
(|K|+n-1)(1+\alpha)<(|K|+n)(1+\alpha')\\
\Leftrightarrow &
(|K|+n-1)\alpha-1<(|K|+n)\alpha'\\
\Leftrightarrow &
\frac{|K|+n-1}{|K|+n}\alpha-\frac{1}{|K|+n}<\alpha'\\
\end{align*}
Since $\alpha'=\frac{|I'|}{|K|}$ and $\alpha=\frac{|I|}{|K|}$, we get that:

\[|I'|>\frac{|K|+n-1}{|K|+n}|I|-\frac{|K|}{|K|+n}\]

With $|I'|$ being integer, it follows that  \[|I'|\geq \frac{|K|+n-1}{|K|+n}|I|=\left(1-\frac{1}{|K|+n}\right)|I|\]

Plugging in the sizes $|K|=(r+h)n$ and $|I|=(nr'+y)n$  with values $y=n^{20}$, $h=n^4$, $r=n^7$ and $r'=n^{14}$, we get that \[\mathcal O(|R|)=\mathcal O(|I\setminus I'|)=\mathcal O\left(\frac{|I|}{|K|+n}\right)=\mathcal O\left(\frac{yn}{rn}\right)=\mathcal O\left(\frac{n^{21}}{n^{8}}\right)=\mathcal O(n^{13})=\mathcal O\left(\frac{r'}{n}\right)\]
 
\medskip 

This firstly gives the claimed bound on $|R|$. Secondly, it follows that $R$ can only contain an insignificant fraction of the vertices in $Y_{i}$ or $Z_i$ for each $i$ - recall here that by our chosen values $\mathcal O(\frac{r'}{n})=\mathcal O(n^{13})$ while $|Y_i|=y=n^{20}$ and $|Z_i|>y=n^{20}$, thus  $\Omega(y)$ vertices of each of these sets needs to be contained in the non-trivial parts. Using our above confirmed structure  $\mathcal P=\{D_1,\dots, D_k, P_{1},\dots, P_{n-k},R\}$ with $W_{i}\subseteq D_{i}$, $i\in [k]$, and $X_{i}\subseteq P_{i}$, $i\in [n-k]$, and the fact that no part can contain more than $n^9$ missing edges, we can only distribute $nn^9$ of the $\Omega({y})$ vertices of each $Y_i$ to some part other than $D_i$ (similarly for $Z_j$ being distributed to some part other than $P_j$). Since $y=n^{20}$, for large enough $n$ this is an insignificant number of vertices of $Y_i$ that can go to some part other than $D_i$. Further, the $Y_i$ and $Z_i$ sets make up the dominant part of $I$ (the total of the remaining sets is in $\mathcal O(n^2r')=\mathcal O(n^{16})$), so they dominate the size of the independent set part of each partition part which gives the claimed bounds:
\begin{itemize}
\item $|Y_{i}\cap D_i|\in \Theta(y)$, $i \in [k]$, and
\item $|Z_{i}\cap P_i|\in \Theta(y)$, $i\in [n-k]$.
\end{itemize}

\end{proofofclaim}

We are now ready to move on to the proof of the 
last claim that will complete the reverse 
direction of our reduction and the overall proof.

\begin{proofofclaim}[Proof of Claim~\ref{clm:ViVkbounds}]
For the first bound of the claim, towards a
contradiction, let us assume that for some $t$, 
$|P_{t}\cap V^{I}|>n^6$. Since the vertices of
$V_{I}$ are, by construction, not adjacent to 
any vertex in $Z$, and $P_{t}$ contains $h$ 
vertices of $X_{t}$, it follows that $P_{t}$ contains 
at least $hn^6=n^{10}$ missing edges, which is a 
contradiction to Claim~\ref{clm:missingedges}.

For the second bound of the claim, towards a 
contradiction again, let us assume that for some 
part $P\in \mathcal{P}$ it holds that 
$|P\cap V^{K}|<\frac{r}{n^3}$. 
By Claim~\ref{clm:structural}, it holds that 
\begin{equation}
n^{4}=h\leq |P\cap K|=h+|P\cap V^{K}|<h+\frac{r}{n^3}=n^4+n^4.
\label{prop6}
\end{equation}
Meanwhile, from one of the items of 
Claim~\ref{clm:structural} (depending on whether
$P=D_{i}$ for some $i\in [k]$ or $P=P_{j}$ for 
some $j\in [n-k]$, it is true that
$$|P\cap I|\in \Theta(n^{20}).$$
We prove that this yields a split that is too far 
from proportional to satisfy the claimed density. 
Recall Inequality~(\ref{overest2}), that is, 
\begin{equation*} 
D-d(\mathcal P) 
 > nd_c(\alpha)+d_l(|K|,\alpha)-d_l\left(|P\cap K|,\tfrac{|P\cap I|}{|P\cap K|}\right)-d_l\left(|K|-|P\cap K|,\tfrac{|I|-|P\cap I|}{|K|-|P\cap K|}\right)
 \end{equation*}

\noindent By Lemma~\ref{lem::splitdense2},
\[d_l(|K|,\alpha)-d_l\left(|P\cap K|,\tfrac{|P\cap I|}{|P\cap K|}\right)-d_l\left(|K|-|P\cap K|,\tfrac{|I|-|P\cap I|}{|K|-|P\cap K|}\right)=\gamma(\alpha_1,\alpha_2,\alpha)n_2(\alpha_2-\alpha_1)(\alpha_2-\alpha)\]
where  $\alpha=\frac{y+nr'}{r+h},  \alpha_1=\frac{|I|-|P\cap I|}{|K|-|P\cap K|},  \alpha_2=\frac{|P\cap I|}{|P\cap K|} $ and $n_2=|P\cap K|$.

By Claim~\ref{clm:structural}, $n_2=|P\cap K|\geq h=n^4$. Furthermore, from $|P\cap I|\in \Theta(n^{20})$ and the further structural properties of Claim~\ref{clm:structural} we have also  that  $$|I|-|P\cap I|=\sum_{S\in \mathcal P, S\neq P}|S\cap I|\in \Omega((n-1)y).$$
Combining the above two inequalities with Inequality~\eqref{prop6} we obtain the following upper and lower bounds for $\alpha_1$ and $\alpha_2$, keeping only the dominant terms of the respective polynomials
\[\alpha_1=\frac{|I|-|P\cap I|}{|K|-|P\cap K|}\geq \frac{\Omega ((n-1)y)}{nr}=\frac{n^{21}}{n^8}\in \Omega (n^{13})\]
\[\alpha_1=\frac{|I|-|P\cap I|}{|K|-|P\cap K|}<\frac{ny}{nr}=\frac{n^{21}}{n^8}\in \mathcal{O}(n^{13})\]
\[\alpha_2=\frac{|P\cap I|}{|P\cap K|} \geq \frac{\Omega(y)}{2n^4}\in  \Omega(n^{16})\]
\[\alpha_2=\frac{|P\cap I|}{|P\cap K|} \leq \frac{\mathcal O(y)}{n^4}\in \mathcal{O}(n^{16})\] 
This means $\alpha_1\in\Theta(n^{13})$ and $\alpha_2\in\Theta(n^{16})$. By definition $\alpha=\frac{y+nr'}{r+h}\in \Theta(n^{13})$, thus 
\begin{align*}D-d(\mathcal P) > &  \ \gamma(\alpha_1,\alpha_2,\alpha)n_2(\alpha_2-\alpha_1)(\alpha_2-\alpha)+nd_c(\alpha)\\
= & \ \frac{1}{2(1+\alpha)(1+\alpha_1)(1+\alpha_2)}n_2(\alpha_2-\alpha_1)(\alpha_2-\alpha)-\frac{n}{2(1+\alpha)}\\
= & \ \frac{1}{2(1+\alpha)}\left(\frac{1}{(1+\alpha_1)(1+\alpha_2)}n_2(\alpha_2-\alpha_1)(\alpha_2-\alpha)-n\right)\end{align*}

By the above bounds on $\alpha, \alpha_1$ and $\alpha_2$ we have \[\displaystyle \frac{1}{(1+\alpha_1)(1+\alpha_2)}n_2(\alpha_2-\alpha_1)(\alpha_2-\alpha)\in \Theta\left(\frac{1}{n^{13}n^{16}}n^4n^{16}n^{16}\right)=\Theta(n^{7})\,.\]
Thus we can conclude that $D-d(\mathcal P)>0$ and therefore any partition of density at least $D$ needs to satisfy the second item of Claim~\ref{clm:ViVkbounds}. 
\end{proofofclaim}

We have concluded the proofs of all the claims of 
the reverse direction, concluding the reverse 
direction and the overall proof of the reduction.

\end{proof}

\section{Bounded Treewidth Connected Partition}\label{appendix:bounded-treewidth}
\paragraph*{Leaf nodes.}
Let $t$ be a leaf node of $\mathcal T$ with $X_t=\{v\}$. For each $i\in [k]$, initialize $S_i=\{v\}$, $S_j=\emptyset$ for all $j\in [k]$ with $j\neq i$ and $ v_1=\cdots=v_k=0$ and $E=\emptyset$. Denote also this state as a possibly valid partial solution. Formally for some $i\in [k]$, $\mathcal T_t[S_1,\ldots, S_{i-1},v,S_{i+1},\ldots, S_k, 0_1,\ldots, 0_k, \emptyset]=1$.
Set all other states to $0$.

\paragraph*{Introduce nodes.}
Let $t'\in\mathcal X$ be an introduce node with a child node $t$ and hence $X_{t'}=X_{t}\cup \{v\}$. For every feasible state $\mathcal T_t[S_1,\ldots, S_k, f_1,\ldots, f_k, E]=1$ and every $i\in[k]$ assign $v$ to part $S_i$ meaning: $S_i'=S_i\cup \{v\}$, $S_j'=S_j$ for $j\neq i$, and $E'=E\cup\{vu\mid u\in S_i\cap N_G(v)\}$. The new state is feasible if and only if both $|S_i'|+f_i\leq n_i$ and if $S_i'=\{v\}$ then $f_i=0$. Then, $T_t'[S_1',\ldots, S_k', f_1,\ldots, f_k, E']=1$.

\paragraph*{Forget Nodes}
Let $t'$ be a forget node of $\mathcal T$ with child $t$ and hence $X_{t'}=X_t\setminus \{v\}$. For each feasible state $\mathcal T_{t'}[S_1,\ldots, S_k,f_1,\ldots, f_k,E]=1$ with $v\in S_i$ we update $S_i'=S_i\setminus\{v\}$, $f_i'=f_i+1$ and $f_j'=f_j$ for $j\in[k]$ and $j\neq i$, we remove all edges incident to $v$ from $E$, $E^-=E\setminus\{vu\in E \mid u\in X_{t}\}$ and also add all missing edges in order for $N_{S_i}(v)$ to induce a clique, $E^+=\{uv:u,v\in N_{S_i}(v), u\neq v\}$. Denote $E'=E^-\cup E^+$. The state $\mathcal T_{t'}[S_1',\ldots, S_k', f_1',\ldots, f_k', E']$ is always denoted as valid unless forgetting $v$ results in losing a part of a connected component that can never be reconnected later and has not yet met its size demand, formally $N_{S_i}(v)=\emptyset$ and, ($S_i'\neq \emptyset$ or $f_i+1<n_i$).

\paragraph*{Join Nodes}
Let $t$ be a join node with children $t_1$ and $t_2$, meaning $X_t=X_{t_1}=X_{t_2}$. For each pair $\mathcal T_{t_1},\mathcal T_{t_2}$ of feasible states such that $\mathcal T_{t_1}[S_1,\ldots, S_k, f_1^1,\ldots, f_k^1, E^1]=1$, $\mathcal T_{t_2}[S_1, S_k, f_1^2,\ldots, f_k^2, E_2]=1$ define $f_i=f_i^1+f_i^2$ and $E=E^1\cup E^2$. The combined state is feasible if and only if for all $i\in[k]$ it is not the case that part $i$ is empty in the bag while both subtrees forgot one of its vertices. Formally, $S_i\neq\emptyset$ or $f_i^1=0$ or $f_i^2=0$.

\paragraph*{Root} At the root, $r$, of $\mathcal T$ we accept if and only if there exists a state $\mathcal T_r[\emptyset_1,\ldots, \emptyset_k, f_1,\ldots, f_k, E]=1$, where $f_i=n_i$ for all $i\in[k]$.

\paragraph*{Correctness}
Now we are ready to show that our algorithm outputs yes if and only if $G$ admits a partition $V_1,\ldots, V_k$ with $|V_i|=n_i$ for all $i\in [k]$ and each $G[V_i]$ is connected. We prove correctness on the nice tree decomposition showing that the following four conditions are met throughout the valid states of our algorithm. For every node $t$ of $\mathcal T$ and every dynamic programming table entry $\mathcal T_{t}[S_1,\ldots, S_k, f_1, \ldots, f_k, E]$ that is denoted as feasible there exists an assignment of the vertices of $G$ to parts $V_1,\ldots, V_k$ such that
\begin{enumerate}
    \item $S_i=V_i\cap X_t$
    \item $|V_i\cap (V(G_t)\setminus X_t)|=f_i$
    \item $|S_i|+f_i\leq n_i$
    \item Let $H_i$ be the induced subgraph $G_t[V_i]$.
    If $|S_i|+f_i<n_i$, then every connected component of $H_i$, intersects $S_i$ and if $|S_i|+f_i=n_i$, then $H_i$ is connected.
\end{enumerate}

Each dynamic programming state at a node $t$ represents a partial assignment of the vertices of $G_t$ to the $k$ parts, together with the counters indicating the amount of vertices that have already been forgotten for each part. Notice that the size constraints are respected and that, for every incomplete part, each connected component still has at least one vertex in the current bag, while a part may become fully forgotten only after it has reached its prescribed size and is already connected.

It is straightforward to verify that the four conditions hold for the initial leaf states. Each introduce step preserves the conditions by assigning the new vertex to exactly one part and checking the size constraint.
At forget nodes, by forcing the remaining neighbors to induce a clique, we preserve all the possible connections through that vertex and also the dynamic programming algorithm rejects all the states where an incomplete part loses its last vertex. At join nodes, the algorithm combines two partial solutions only when no incomplete part is separated across the two subtrees without a vertex in the bag to connect them and hence the four conditions stated above are again met.
Lastly, at the root all vertices are forgotten. Thus a feasible state at the root satisfies the four conditions stated if and only if each part has exactly the required size and induces a connected subgraph in $G$.
\paragraph*{Running time}
Let $w=\text{tw}(G)$. Assigning the vertices of some bag of the tree decomposition into $k$ parts can be done in at most $k^{w+1}$ ways. Moreover the amount of forgotten vertices $f_i$ is at most $n_i$ for all $i\in [k]$ giving $n+1$ possibilities for each $f_i$, while the edge set that ensures that the global connectivity is encoded correctly even after forgetting vertices can be done in at most $2^{\binom{w+1}{2}}$ ways. Hence the amount of states per bag is $k^{w+1}(n+1)^k2^{\binom{w+1}{2}}$.
Introduce and forget nodes process each state in time polynomial in $w$ and $k$. Join nodes are the ones that are mainly responsible for the running time of the algorithm, as they require combining pairs of compatible states from two children. Hence, for each $n_i$ there are $n_i^2$ possible pairs of amount of forgotten vertices to be considered. In general since the possible states of a bag are $k^{w+1}(n+1)^k2^{\binom{w+1}{2}}$ then at a join node we need to examine all possible pairs and hence get a running time of $\mathcal O(k^{2(w+1)}(n+1)^{2k}4^{\binom{w+1}{2}})$. Lastly, we process each bag of the decomposition once and at each forget node we add at most $\binom{w}{2}$ edges and also at the join nodes we examine those edges to transfer the connectivity to the union of the parts giving an factor of $w^2$ to the running time. Hence the overall running time of our algorithm is $\mathcal O(nw^2k^{2(w+1)}(n+1)^{2k}4^{\binom{w+1}{2}})$.

Notice that it is quite straightforward to modify this algorithm to realize Gy\H{o}ri-Lov\'{a}sz theorem by simply adding a condition that rejects a state if a terminal vertex is assigned to a ``wrong'' part.

\glpart*

Moreover, by slightly modifying it to not encode at each state (and receive as an input) the size of each part, but to encode the amount of forgotten edges and vertices from each part it is also possible to obtain an algorithm for dense partition on graphs with bounded treewidth. Exploiting the fact that each part in an optimum MDGP solution induces a connected subgraph, one can also solve MDGP without given $k$ this way. More precisely, within each bag, at most treewidth many parts are \emph{active}, in the sense that some vertices are already chosen, and more might be added; thus replacing $k$ with $w$ in the runtime. These observations directly yield the following.

\mdgpkk*

Hence we conclude the following theorem, which is the main result of this section.

\thicktrees*
\end{document}